%% file: tocl.tex
\documentclass[UKenglish,acmsmall,authorversion]{acmart}

\usepackage{virginialake}
\usepackage{xypic}
\vlsmallopstrue

\usepackage{booktabs} 
\usepackage{proof}
\usepackage{enumitem}

\newcommand{\tohide}[2]{#1} 
\newcommand{\fullproof}[1]{} 

\usepackage{lipsum}                     
\usepackage{xargs}                      
\usepackage[colorinlistoftodos,prependcaption,textsize=small]{todonotes}
\newcommandx{\ross}[2][1=]
{\todo[linecolor=blue,backgroundcolor=blue!25,bordercolor=blue,#1]{Ross: #2}} 
\newcommandx{\bogdan}[2][1=]{\todo[linecolor=red,backgroundcolor=red!25,bordercolor=red,#1]{Bogdan: #2}}
\newcommandx{\gabriel}[2][1=]{\todo[linecolor=green,backgroundcolor=green!25,bordercolor=green,#1]{Gabriel: #2}}
\newcommandx{\alwen}[2][1=]{\todo[linecolor=purple,backgroundcolor=purple!25,bordercolor=purple,#1]{Alwen: #2}}

\usepackage[T2A,T1]{fontenc}

\usepackage{wrapfig}
\usepackage{etex}
\usepackage{amsmath,amsfonts,amsthm}
\theoremstyle{plain}
\usepackage{verbatim}
\usepackage{fancyhdr}
\usepackage{microtype}
\usepackage{cmll}
\usepackage{prooftree}
\usepackage{rotating}
\usepackage{syndmeta}

\usepackage{flushend}

\usepackage{fancybox,tikz}
\usetikzlibrary{positioning}

\def\proof{{\noindent \bf Proof.~}}

\title{De Morgan Dual 
 Nominal Quantifiers Modelling Private Names in Non-Commutative Logic}

\acmJournal{TOCL}
\acmVolume{20}
\acmNumber{4}
\acmArticle{22}
\acmYear{2019}
\acmMonth{7}
\copyrightyear{2019}

\setcopyright{acmlicensed}

\acmDOI{10.1145/3325821}

\received{November 2017}
\received[revised]{September 2018}
\received[accepted]{April 2019}

\author[R.~Horne]{Ross Horne}
\affiliation{
Computer Science and Communications,
University of Luxembourg
}
\email{ross.horne@uni.lu}

\author[A.~Tiu]{Alwen Tiu}
\affiliation{
Research School of Computer Science,
The Australian National University, Australia
}
\email{alwen.tiu@anu.edu.au}

\author[B.~Aman]{Bogdan Aman}
\affiliation{
Alexandru Ioan Cuza University of Ia\c si, Romania
}
\email{bogdan.aman@iit.academiaromana-is.ro}

\author[G.~Ciobanu]{Gabriel Ciobanu}
\affiliation{
Alexandru Ioan Cuza University of Ia\c si, Romania
}
\email{gabriel@info.uaic.ro}

\input{abstract}

\begin{document}

%
%


%
%



\thanks{
R.\ Horne\ and A.\ Tiu\ acknowledge support from Ministry of Education, Singapore, Tier 2 grant MOE2014-T2-2-076.

Authors' addresses: 
R.~Horne,
Faculty of Science, Technology and Communication,
6 avenue de la Fonte,
L-4364 Esch-sur-Alzette,
Luxembourg,
email:
\texttt{ross.horne@uni.lu};
A.~Tiu,
Research School of Computer Science,
The Australian National University,
Building 108, North Road,
Canberra, ACT 2061, Australia;
email: \texttt{alwen.tiu@anu.edu.au};
B.~Aman and G.~Ciobanu,
Alexandru Ioan Cuza University of Ia\c si, Romania;
email: \texttt{bogdan.aman@iit.academiaromana-is.ro}
and
\texttt{gabriel@info.uaic.ro}.
}

\maketitle

\input{introduction}

\input{self-dual}

\input{syntax}

\input{implication}

\input{nominal}

\input{lemmas}

\tohide{
\input{multisets}
\input{bound}
}{}
\input{splitting}

\input{context}
\input{corule}
\input{symmetric}

\input{decidable}

\input{conclusion}

\bibliographystyle{plain} 
\bibliography{WWV15}

\input{appendix}

\end{document}

%% file: abstract.tex
\begin{abstract}
This paper explores the proof theory necessary for recommending an expressive but decidable first-order system, named \textsf{MAV1}, featuring a de Morgan dual pair of nominal quantifiers. These nominal quantifiers called `new' and `wen' are distinct from the self-dual Gabbay-Pitts and Miller-Tiu nominal quantifiers. The novelty of these nominal quantifiers is they are polarised in the sense that `new' distributes over positive operators while `wen' distributes over negative operators. This greater control of bookkeeping enables private names to be modelled in processes embedded as formulae in \textsf{MAV1}. The technical challenge is to establish a cut elimination result, from which essential properties including the transitivity of implication follow.
Since the system is defined using the calculus of structures, a generalisation of the sequent calculus, novel techniques are employed. The proof relies on an intricately designed multiset-based measure of the size of a proof, which is used to guide a normalisation technique called \textit{splitting}. The presence of equivariance, which swaps successive quantifiers, induces complex inter-dependencies between nominal quantifiers, additive conjunction and multiplicative operators in the proof of splitting. Every rule is justified by an example demonstrating why the rule is necessary for soundly embedding processes and ensuring that cut elimination holds.
\end{abstract}

%% file: introduction.tex
\section{Introduction}
\label{sec:intro}

This paper investigates the proof theory of a novel pair of de Morgan dual nominal quantifiers.
These quantifiers are motivated by the desire to model private name binders in processes by embedding the processes directly as formulae in a suitable logical system.
The logical system in which this investigation is conducted is sufficiently expressive to soundly embed the finite fragment of several process calculi.

A requirement of directly embedding processes as formulae is that the logic should be able to capture causal dependencies.
To do so, we employ a non-commutative multiplicative operator, which can be used to model the fact that `$a$ happens before $b$' is not equivalent to `$b$ happens before $a$'.
Such non-commutative operators are problematic for traditional proof frameworks such as the sequent calculus;
hence we adopt a formalism called the \textit{calculus of structures}~\cite{Guglielmi2007,Tiu2006,Guglielmi2011,Strassburger2011,Strassburger2003thesis}.
The calculus of structures permits more proofs than the sequent calculus, by allowing inference rules to be applied in any context; while still satisfying proof theoretic properties, notably cut elimination.
An advantage of the calculus of structures is that it can express proof systems combining connectives for sequentiality and parallelism.
%
The calculus of structures was motivated by a need for understanding
why pomset logic~\cite{Retore1997} could not be expressed in the
sequent calculus. Pomset logic is inspired by
pomsets~\cite{Pratt1986} and linear logic~\cite{Girard1987}, the
former being a model of concurrency respecting
causality, while the latter can be interpreted in
various ways as a logic of resources and
concurrency~\cite{Caires2010,Kobayashi1993,Wadler2014}.

These observations lead to the propositional system \textsf{MAV}~\cite{Horne2015} and its first-order extension presented in this work, named \textsf{MAV1}.
Related work establishes that linear implication in such logical systems is sound with respect to both pomset ideals~\cite{Horne2017}
and weak simulation~\cite{MSCS}.
These results tighten results in initial investigations concerning a minimal calculus \textsf{BV} and trace inclusion~\cite{Bruscoli2002}.
Hence reasoning using linear implication is sound with respect to most useful (weak) preorders over processes,
for a range of languages not limited to \textsf{CCS}~\cite{Milner1982} and $\pi$-calculus~\cite{Milner1992}.

This paper resolves the fundamental logical problem of whether cut elimination holds for \textsf{MAV1}.
Cut elimination, the corner stone of a proof system, is essential for confidently recommending a proof system.
In the setting of the calculus of structures, cut elimination is formalised quite differently compared to traditional proof frameworks;
hence the proof techniques employed in this paper are of considerable novelty.
Furthermore, this paper is the first paper to establish cut elimination for a de Morgan dual pair of nominal quantifiers in any proof framework.
These nominal quantifiers introduce intricate interdependencies between other operators in the calculus, reflected in the technique of \textit{splitting} (Lemma~\ref{lemma:split-times}) which is the key lemma required to establish cut elimination (Theorem~\ref{theorem:cut}).

Logically speaking, nominal quantifiers $\new$ and $\wen$, pronounced `new' and `wen' respectively, sit between $\forall$ and $\exists$ such that
$
\forall x. P \multimap \new x. P
$
and
$
\new x. P \multimap \wen x. P
$ and
$
\wen x. P \multimap \exists x. P
$, where $\multimap$ is linear implication.
The quantifier $\new$ is similar in some respects to $\forall$, whereas $\wen$
is similar to $\exists.$ A crucial difference between $\exists x. P$
and $\wen x. P$ is that variable $x$ in the latter cannot be instantiated with arbitrary terms,
but only `fresh' names introduced by $\new.$
Our \textit{new} quantifier $\new$, distinct from the Gabbay-Pitts quantifier, addresses limitations of established self-dual nominal quantifiers for modelling private names in embeddings of processes as formulae. 
In particular, our $\new$ quantifier does not distribute over parallel composition in either direction.
In \textsf{MAV1}, the formulae
$
\mathopen{\new x.}({\mathrm{event}(x)} \cpar {\mathrm{event}(x)})
$ and
$
\new x.{\mathrm{event}(x)} \cpar \new x.{\mathrm{event}(x)}
$
are unrelated by linear implication. This property is essential for soundly modelling private name binders in processes.
\textbf{Outline.}
For a new logical system it is necessary to justify correctness, which we approach in proof theoretic style by cut elimination.
Section~\ref{section:self-dual} illustrates why an established self-dual nominal quantifier~\cite{Gabbay2002,Pitts2003,Tiu2009,Gacek2011} is incapable of soundly modelling name restriction in a processes-as-formulae embedding.
Section~\ref{section:syntax} defines \textsf{MAV1}, explains cut elimination and discusses rules.
Section~\ref{section:nominal} presents an explanation of the rules for the nominal quantifiers.
Section~\ref{section:splitting} presents technical lemmas and the \textit{splitting} technique which is key to cut elimination.
Section~\ref{section:context} presents a context lemma which is used to eliminate \textit{co-rules} that form a cut; thereby establishing cut elimination.
Section~\ref{section:complexity} explains the complexity classes for various fragments of \textsf{MAV1}.

The cut elimination result in this article was announced at CONCUR 2016~\cite{Horne2016}, without full proofs.
This journal version of the paper explains the cut elimination proof, elaborates on the motivating discussion, and highlights further corollaries of cut elimination.
Since $\new$ is a Cyrillic vowel, we use another Cyrillic vowel $\wen$ for nominal quantifier 'wen'. This Cyrillic vowel is pronounced as the hard e in `wen' and reminds the reader of its existential nature.

Due to the space limitation, some proofs are omitted in the printed version of this article, but are available in the accompanying Electronic Appendix.

%% file: self-dual.tex
\section{Why not a self-dual nominal quantifier?}
\label{section:self-dual}

Nominal quantifiers in the literature are typically self-dual in the sense of de Morgan dualities.
That is, for a nominal quantifier, say $\nabla$, ``not $\nabla x$ $P$'' is equivalent to ``$\nabla x$ not $P$.''
Such self-dual nominal quantifiers have been successfully introduced in classical and intuitionistic frameworks, typically used to reason about higher-order abstract syntax with name binders.
Such nominal frameworks are therefore suited to program analysis, where the semantics of a programming language are encoded as a theory over terms in the logical framework.

Rather surprisingly, when processes themselves are directly embedded as formulae in a logic, where constructs are mapped directly to primitive logical connectives 
(as opposed to terms inside a logical encoding of the semantics of processes),
self-dual quantifiers do not exhibit typical properties expected of name binders.
To understand this problem, in this section we recall an established calculus \textsf{BVQ}~\cite{Roversi14} that can directly embed processes but features a self-dual nominal quantifier.
We explain that such a self-dual quantifier provides an unsound semantics for name binders.
This motivates the need for a finer polarised nominal quantifier, which leads to the calculus introduced in subsequent sections.

We assume the reader has a basic understanding of 
the semantics of the $\pi$-calculus~\cite{Milner1992} and $\textsf{CCS}$~\cite{Milner1982}.
This section provides necessary preliminaries for the calculus of structures.

\subsection{An established extension of \textsf{BV} with a self-dual quantifier}

An abstract syntax for formulae and the rules of \textsf{BVQ} are defined in Fig~\ref{figure:BVQ}.
In an inference rule, the formula appearing above the horizontal line is the premise and the formula below the horizontal line is the conclusion.
The key feature of the calculus of structures is \textit{deep inference}, which is the ability to apply
all rules in any context, i.e.\ formulae with a hole of the following form: $\context{} \Coloneqq \ehole{\,\cdot\,} \mid \context{} \odot P \mid P \odot \context{} \mid \nabla x . \context{}$, 
where $\odot \in \left\{ \cseq, \cpar, \tensor \right\}$.


\begin{figure}
 \begin{center}
\[
\begin{array}{l}
\mbox{\footnotesize Structural rules}
\\[4pt]
\begin{tabular}{c}
$\left(P, \cpar, \cunit\right)\mbox{ and }\left(P, \tensor, \cunit\right)\mbox{ are commutative monoids}$
\\[9pt]
$\left(P, \cseq, \cunit\right)\mbox{ is a monoid}$
\qquad
\mbox{$\alpha$-conversion for $\nabla$ quantifier}
\\[9pt]
$\nabla x .\nabla y. P \equiv \nabla y .\nabla x. P
~~\mbox{(equivariance)}$
\\[9pt]
$\nabla x . P \equiv P
~~\mbox{only if $\nfv{x}{P}$}
~~\mbox{(vacuous)}$
\end{tabular}
\end{array}
\qquad\quad
\begin{array}{l}
\mbox{\footnotesize Syntax}
\\[4pt]
\begin{tabular}{l@{\hspace{1ex}}l@{\hspace{1ex}}l@{\hspace{1ex}}l}
$P$&$::=$ & $\cunit$ & (unit)\\
&& $\alpha$ & (atom)\\
&& $\overline{\alpha}$ & (co-atom)\\
&& $\nabla{x}. P$ & (nabla)\\
&& $P \cpar P$ & (par)\\
&& $P \tensor P$ & (times)\\
&& $P \cseq P$ & (seq)
\end{tabular}
\end{array}
\]
\[
\begin{array}{l}
\mbox{\footnotesize Inference rules}
\\[4pt]
\begin{array}{c}
\infer[\mbox{(atomic interaction)}]{
\context{ \co{\alpha} \cpar \alpha }
}{
\context{ \cunit }
}
\qquad\qquad
\infer[\mbox{(switch)}]{
\context{P \cpar \left(Q \tensor S\right)}
}{
\context{(P \cpar Q) \tensor S}
}
\\[9pt]
\infer[\mbox{(sequence)}]{
\context{(P \cseq Q) \cpar (R \cseq S)}
}{
\context{(P \cpar R) \cseq (Q \cpar S)}
}
\qquad\qquad
\infer[\mbox{(unify)}]{
\context{ \nabla x. P \cpar \nabla x. Q }
}{
\context{ \mathopen{\nabla x .}\left( P \cpar Q \right) }
}
\end{array}\end{array}
\]
\end{center}
\caption{Syntax and rules of system \textsf{BVQ}~\cite{Roversi14}: which is \textsf{BV} extended with a self-dual nominal quantifier.
}\label{figure:BVQ}
\end{figure}

Inference rules are defined \textit{modulo a structural congruence}, where a congruence is an equivalence relation that holds in any context.
A \textit{derivation} is a sequence of rules from Fig.~\ref{figure:BVQ}, where the structural congruence can be applied at any point in a derivation.
The length of a derivation involving only the structural congruence is zero. 
The length of a derivation involving one inference rule instance is one.
Given a derivation $\vcenter{\infer{Q}{P}}$ of length $m$ and another $\vcenter{\infer{R}{Q}}$ of length $n$, the derivation $\vcenter{\infer{R}{P}}$ is of length $m+n$.
Unless we make it clear in the context that we refer to a specific rule, this horizontal line notation is generally used to represent derivations of any length.
For example, since $\nabla x.\cunit \equiv \cunit$, derivation $\infer{\nabla x.\cunit}{\cunit}$  of length 0, and derivation 
$\vcenter{\infer{\left(P \tensor Q\right) \cpar R \cpar S}{\left(P \cpar R\right) \tensor \left(Q \cpar S\right)}}$ is of length 2, since two instances of \textit{switch} are applied.


The congruence, $\equiv$ in Fig.~\ref{figure:BVQ}, makes \textit{par} and \textit{times} commutative and \textit{seq} non-commutative in general. 
For the nominal quantifier $\nabla$, the congruence enables: $\alpha$-conversion for renaming bound names;
\textit{equivariance} which allows names bound by successive nominal quantifiers to be swapped; and \textit{vacuous} that allows the nominal quantifier to be introduced or removed whenever the bound variable does not appear in the formula. 
As standard, we define a freshness predicate such that a variable $x$ is fresh for a formulae $P$, written $\nfv{x}{P}$, if and only if $x$ is not a member of the set of free variables of $P$, where $\nabla x.P$ binds occurrences of $x$ in $P$.

Consider the syntax and rules of $\textsf{BVQ}$ in Figure~\ref{figure:BVQ}.
The three rules \textit{atomic interaction} and $\textit{switch}$ and \textit{sequence} define the basic system \textsf{BV}~\cite{Guglielmi2007} that also forms the core of the system \textsf{MAV1} investigated in later sections.
The only additional inference rule for $\nabla$ is called \textit{unify}.

\textbf{Atomic interaction.}
The atomic interaction rule should remind the reader of the classical tautology $\neg \alpha \vee \alpha$ or intuitionistic axiom $\alpha \Rightarrow \alpha$, applied only to the predicates forming the atoms of the calculus.
Since there is no contraction rule for $\cpar$, once atoms are consumed by \textit{atomic interaction} they cannot be reused. Thus \textit{atomic interaction} is useful for modelling communication in process, where $\alpha$ models a receive action or event and $\co{\alpha}$ is the complementary send, which cancel each other out.

\textbf{Switch and sequence.}
The \textit{atomic interaction} and $\textit{switch}$ rules together provide a model for multiplicative linear logic (with \textit{mix})~\cite{Girard1987}.
The difference between $\cpar$ and $\tensor$ is that $\cpar$ allows interaction, but $\tensor$ does not. In this sense the switch rule restricts where which atoms may interact.
The \textit{seq} rule also restricts where interactions can take place, but, since \textit{seq} is non-commutative, it can be used to capture causal dependencies between atoms.
The \textit{sequence} rule preserves these causal dependencies, while permitting new causal dependencies. In terms of process models, the \textit{sequence} rule appears in the theory of pomsets~\cite{Gischer1988} and can refine parallel composition to its interleavings.

\textbf{Unify.}
The novel rule for \textsf{BVQ} is \textit{unify} for nominal quantifier $\nabla$.
The \textit{unify} rule should be admissible in a well-designed extension of linear logic with a self-dual quantifier. To see why, consider the following auxiliary definitions.
Observe that the following definition of linear implication ensures that 
$\nabla$ is self-dual in the sense that the de Morgan dual of $\nabla$ is $\nabla$ itself. Similarly, \textit{seq} and the unit are self-dual, while $\tensor$ and $\cpar$ are a de Morgan dual pair of operators. 
\begin{definition}
\textit{Linear negation} is defined by the following function over formulae.
\begin{gather*}
\co{\cunit} = \cunit
\qquad
\co{\co{\alpha}} = \alpha
\qquad
\co{ P \tensor Q } = \co{P} \cpar \co{Q}
\qquad
\co{ P \cpar Q } = \co{P} \tensor \co{Q}
\qquad
\co{P \cseq Q} = \co{P} \cseq \co{Q}
\qquad
\co{\nabla x. P} = \nabla x. \co{P}
\end{gather*}
\textit{Linear implication}, written $P \multimap Q$, is defined as $\co{P} \cpar Q$.
\end{definition}

We are particularly interested in special derivations, called proofs.
\begin{definition}
A \textit{proof} is a derivation of any length with conclusion $P$ and premise $\cunit$. When such a derivation exists, we say that $P$ is provable, and write $\vdash P$ holds.
\end{definition}

As a basic property of linear implication $\vdash P \multimap P$ must hold for any $P$. Now assume that $\vdash Q \multimap Q$ is provable in \textsf{BVQ} (hence, by the above definitions, there exists a derivation with conclusion $\co{Q} \cpar Q$ and premise $\cunit$), and consider formula $\nabla x. Q$. Using the \textit{unify} rule and the definition of linear implication, we can construct the following proof of $\vdash \nabla x.Q \multimap \nabla x.Q$.
\[
\infer[\mbox{by the \textit{unify} rule}]{
\nabla x.\co{Q} \cpar \nabla x.Q 
}{
\infer[\mbox{by the assumption $\vdash \co{Q} \cpar Q$}]{
\mathopen{\nabla x.}\left(\co{Q} \cpar Q\right)
}{
\infer[\mbox{by the \textit{vacuous} rule}]{
\nabla x.\cunit
}{
\cunit
}}}
\]
The above illustrates why \textit{unify} should be admissible in order to guarantee \textit{reflexivity} --- the most basic property of implication --- for an extension of \textsf{BV} with a self-dual nominal quantifier.
In the next section, we explain why the \textit{unify} rule is problematic for modelling processes as formulae.

\subsection{Fundamental problems with a self-dual nominal for embeddings of processes}

Initially, it seems that desirable properties of name binding, typical of process calculi, are achieved in \textsf{BVQ}.
For example, we expect that if $\nfv{x}{Q}$ then $\vdash \nabla x.\left( P \cpar Q \right) \multimap \nabla x.P \cpar Q$,
indicating that the scope of a name can be \textit{extruded} as long as another name is not captured,
which is provable using the \textit{vacuous} and \textit{unify} rules.
The \textit{equivariance} rule that swaps name binders is also a property preserved by most equivalences over processes.

Another strong property of \textsf{BVQ}, expected of all nominal quantifiers, is that we avoid the \textit{diagonalisation} property.
Diagonalisation $\vdash \forall x .\forall y .P(x,y) \multimap \forall z .P(z,z)$ holds in any system with universal quantifiers, as does the converse for existential quantifiers.
However, for nominals such at $\nabla$, \textbf{neither} $\nabla x .\nabla y .P(x,y) \multimap \nabla z .P(z,z)$ \textbf{nor} its converse $\nabla z .P(z,z) \multimap \nabla x .\nabla y .P(x,y)$ hold.
This is a critical feature of all nominal quantifiers that ensures that distinct fresh names in the same scope never collapse to the same name, and explains why universal and existential quantifiers are not suited modelling fresh name binders.
It is precisely the absence of diagonalisation for nominals that is used in classical~\cite{Pitts2003,Gabbay2002} and intuitionistic frameworks~\cite{Tiu2009,Gacek2011}
to logically manage the bookkeeping of fresh name in, so called, \textit{deep embeddings} of processes as terms in a theory.
Avoiding diagonalisation is sufficient in such deep embeddings since nominal quantifiers cannot appear inside a term representation of a process, so are always pushed to the outermost level where formulae are used to define the operational semantics of processes as a theory over process terms.

\textbf{Soundness criterion.}
The problem with \textsf{BVQ} is that when processes are directly embedded as formulae $\nabla$ quantifiers may appear inside embeddings of processes, which can result in unsound behaviours.
To see why the \textit{unify} rule induces unsound behaviours consider the following $\pi$-calculus terms. 
$\mathopen{\nu x.}(\co{z}x \mathrel{|} \co{y}x)$ is a $\pi$-calculus process that can output a fresh name twice, once on channel
$z$ and once on channel $y$; but cannot output two distinct names in any execution.
In contrast, observe that $\nu x.\co{z}x \mathrel{|} \nu x.\co{y}x$ is a $\pi$-calculus process that outputs
two distinct fresh names before terminating, but cannot output the same name twice in any execution.
As a soundness criterion, since the processes $\mathopen{\nu x.}(\co{z}x \mathrel{|} \co{y}x)$ and $\nu x.\co{z}x \mathrel{|} \nu x.\co{y}x$ do not have any complete traces in common, these processes must not be related by any sound preorder over processes.

Now consider an embedding of these processes in \textsf{BVQ}, where the parallel composition operator of the $\pi$-calculus is encoded as \emph{par} and $\nu$ is encoded as $\nabla$.
This gives us the formulae
$\mathopen{\nabla x.} \left(\co{\mathrm{act}(z, x)} \cpar \co{\mathrm{act}(y, x)}\right)$
and
$\nabla x.\co{\mathrm{act}(z, x)} \cpar \nabla x.\co{\mathrm{act}(y, x)}$.
Note that output action prefixes are encoded as negated predicates, e.g., $\co{z}x$ is
encoded $\co{\mathrm{act}(z,x)}$. 

Observe that 
$\vdash
\mathopen{\nabla x.} \left(\co{\mathrm{act}(z,x)} \cpar \co{\mathrm{act}(y,x)}\right)
\multimap
\nabla x.\co{\mathrm{act}(z,x)} \cpar \nabla x.\co{\mathrm{act}(y,x)}
$
is provable, as follows. 
\[
\infer[\mbox{by \textit{unify}}]{
\mathopen{\nabla x.} \left({\mathrm{act}(z,x)} \tensor {\mathrm{act}(y,x)}\right)
\cpar
\nabla x.\co{\mathrm{act}(z,x)} \cpar \nabla x.\co{\mathrm{act}(y,x)}
}{
\infer[\mbox{by \textit{unify}}]{
\mathopen{\nabla x.} \left({\mathrm{act}(z,x)} \tensor {\mathrm{act}(y,x)}\right)
\cpar
\mathopen{\nabla x.}\left(\co{\mathrm{act}(z,x)} \cpar \co{\mathrm{act}(y,x)}\right)
}{
\infer[\mbox{by \textit{switch}}]{
\mathopen{\nabla x.} \left(\left({\mathrm{act}(z,x)} \tensor {\mathrm{act}(y,x)}\right)
\cpar
\co{\mathrm{act}(z,x)} \cpar \co{\mathrm{act}(y,x)}\right)
}{
\infer[\mbox{by \textit{switch}}]{
\mathopen{\nabla x.} \left(\left(\left({\mathrm{act}(z,x)} \cpar \co{\mathrm{act}(z,x)}\right) \tensor {\mathrm{act}(y,x)}\right)
 \cpar \co{\mathrm{act}(y,x)}\right)
}{
\infer[\mbox{by \textit{atomic interaction}}]{
\mathopen{\nabla x.} \left(\left({\mathrm{act}(z,x)} \cpar \co{\mathrm{act}(z,x)}\right)
\tensor
\left({\mathrm{act}(y,x)} \cpar \co{\mathrm{act}(y,x)}\right)\right)
}{
\infer[\mbox{by \textit{atomic interaction}}]{
\mathopen{\nabla x.} \left({\mathrm{act}(y,x)} \cpar \co{\mathrm{act}(y,x)}\right)
}{
\infer[\mbox{by \textit{vacuous}}]{
\mathopen{\nabla x.} \cunit
}{
\cunit
}}}}}}}
\]
The above implication is \textbf{unsound} with respect to trace inclusion for the $\pi$-calculus. 
The implication wrongly suggests that the process $\nu x.\co{z}x \mathrel{|} \nu x.\co{y}x$, that cannot output the same names twice, can be refined to a process $\mathopen{\nu x.}(\co{z}x \mathrel{|} \co{y}x)$, that outputs the same name twice.
This is exactly the contradiction that we avoid by using polarised nominal quantifiers investigated in subsequent sections.

As a further example of unsoundness issues for a self-dual nominal, consider the following criterion: an embedding of a process is provable if and only if there is a series of internal transitions leading to a successful termination state.
A successful termination state is a state without any unconsumed actions.
Now consider the process $\mathopen{\nu x.} \left(x.y\right) \mathrel{|} \nu z.\co{z} \mathrel{|} \co{y}$ in process calculus \textbf{CCS}~\cite{Milner1982}.
We can attempt to embed this process in \textsf{BVQ} as $\mathopen{\nabla x.} \left(\mathrm{event}(x) \cseq \mathrm{event}(y)\right) \cpar \nabla z.\co{\mathrm{event}(z)} \cpar \co{\mathrm{event}(y)}$, where $\mathrm{event}(x)$ is a unary predicate representing an event identified by variable $x$.
This embedding \textbf{violates} our soundness criterion. Under the semantics of \textsf{CCS} the process is immediately deadlocked; hence none of the four actions are consumed.
However, the embedding is a provable formula, by the following derivation.
\[
\infer[\mbox{by \textit{unify} and $\alpha$-conversion}]{
\mathopen{\nabla x.} \left(\mathrm{event}(x) \cseq \mathrm{event}(y)\right) \cpar \nabla z.\co{\mathrm{event}(z)} \cpar \co{\mathrm{event}(y)}
}{
\infer[\mbox{by  \textit{vacuous} and \textit{unify}}]{
\mathopen{\nabla x.} \left(\left(\mathrm{event}(x) \cseq \mathrm{event}(y)\right) \cpar \co{\mathrm{event}(x)}\right) \cpar \co{\mathrm{event}(y)}
}{
\infer[\mbox{by  \textit{sequence}}]{
\mathopen{\nabla x.} \left(\left(\mathrm{event}(x) \cseq \mathrm{event}(y)\right) \cpar \co{\mathrm{event}(x)} \cpar \co{\mathrm{event}(y)}\right)
}{
\infer[\mbox{by  \textit{sequence}}]{
\mathopen{\nabla x.} \left(\left(\mathrm{event}(x) \cseq \mathrm{event}(y)\right) \cpar \left(\co{\mathrm{event}(x)} \cseq \co{\mathrm{event}(y)}\right)\right)
}{
\infer[\mbox{by  \textit{atomic interaction}}]{
\mathopen{\nabla x.} \left(\left(\mathrm{event}(x) \cpar \co{\mathrm{event}(x)}\right) \cseq \left(\mathrm{event}(y) \cpar \co{\mathrm{event}(y)}\right)\right)
}{
\infer[\mbox{by  \textit{atomic interaction} and \textit{vacuous}}]{
\mathopen{\nabla x.} \left(\mathrm{event}(y) \cpar \co{\mathrm{event}(y)}\right)
}{
\cunit
}}}}}}
\]

The above observations lead to a specification of the properties desired for a nominal quantifier suitable for direct embeddings of processes as formulae.
We desire a nominal quantifier, say $\new$, such that properties such as \textit{no diagonalisation}, \textit{equivariance} and \textit{extrusion} hold except that
also 
\textbf{neither} 
$
\mathopen{\new x.} \left(P \cpar Q \right)
\multimap
\new x.P \cpar \new x.Q
$
\textbf{nor}
$
\new x.P \cpar \new x.Q
\multimap
\mathopen{\new x.} \left(P \cpar Q \right)
$
hold in general.
Also, by the arguments above the quantifier cannot be self-dual;
and hence, as a side effect, we expose another nominal quantifier, called ``wen'', denoted $\wen$, that is de Morgan dual to $\new$.
The rest of this paper is devoted to establishing that indeed there does exist a logical system with such a pair of nominal quantifiers.

%% file: syntax.tex
\section{Introducing a Proof System with a Pair of Nominal Quantifiers}
\label{section:syntax}

Soundness issues associated with a self-dual nominal quantifier in embeddings of processes as formulae, can be resolved by instead using a pair of de Morgan dual nominal quantifiers.
This section introduces a proof system for such a pair of nominal quantifiers, 
building on the core system \textsf{BV}, further extended with: additives useful for expressing non-deterministic choice; and first-order quantifiers which range over terms not only fresh names.
Investigating the pair of nominal quantifiers in the presence of these operators is essential for understanding  the interplay between nominal quantifiers and other operators, showing that this pair of nominal quantifiers can exist in a system sufficiently expressive to embed rich process models.
This section also summarises the main proof theoretic result, although lemmas are postponed until later sections.


\subsection{The inference rules and structural rules}

\begin{figure}[t]
\begin{tabular}{c}
\begin{tabular}{c}
$x$ a variable
\\[9pt]
$c$ a constant
\\[9pt]
$f$ a function symbol
\\[9pt]
$\mathrm{p}$ a predicate symbol
\\[9pt]
\end{tabular}
\\
\begin{tabular}{l@{\hspace{1ex}}l@{\hspace{1ex}}l@{\hspace{1ex}}l}
$t$&$::=$ & $x$ & (variable)\\
&& $c$ & (constant)\\
&& $f(t, \hdots t)$ & ($n$-ary function)
\\[9pt]
$\alpha$&$::=$ & $\mathrm{p}(t, \hdots t)$ & ($n$-ary predicate)\\
\end{tabular}
\qquad\qquad\qquad
\end{tabular}
\begin{tabular}{l@{\hspace{1ex}}l@{\hspace{1ex}}l@{\hspace{1ex}}l}
$P$&$::=$ & $\cunit$ & (unit)\\
&& $\alpha$ & (atom)\\
&& $\overline{\alpha}$ & (co-atom)\\
&& $\forall x. P$ & (all)\\
&& $\exists x. P$ & (some)\\
&& $\new x. P$ & (new)\\
&& $\wen x. P$ & (wen)\\
&& $P \wwith P$ & (with)\\
&& $P \ooplus P$ & (plus) \\
&& $P\cpar P$ & (par)\\
&& $P \tensor P$ & (times)\\
&& $P \cseq P$ & (seq)
\end{tabular}
\caption{Syntax for \textsf{MAV1} formulae.}\label{figure:syntax}
\end{figure}

\begin{figure}[t]
\begin{tabular}{c}
$\left(P, \cpar, \cunit\right)\mbox{ and }\left(P, \tensor, \cunit\right)\mbox{ are commutative monoids}
\mbox{ and }
\left(P, \cseq, \cunit\right)\mbox{ is a monoid.}$
\\[9pt]
$\new x .\new y. P \equiv \new y .\new x. P
\qquad
\wen x .\wen y .P \equiv \wen y .\wen x .P
\qquad\mbox{(equivariance)}$
\end{tabular}
\caption{Structural congruence ($\equiv$)
 for \textsf{MAV1} formulae,
plus $\alpha$-conversion for all quantifiers.
}\label{figure:congruence}
\end{figure}

\begin{figure}[t]
\[
\begin{array}{c}
\infer[\mbox{(atomic interaction)}]{
\context{ \co{\alpha} \cpar \alpha }
}{
\context{ \cunit }
}
\qquad\quad\quad
\infer[\mbox{(switch)}]{
\context{P \cpar \left(Q \tensor S\right)}
}{
\context{(P \cpar Q) \tensor S}
}
\\ [10pt]
\infer[\mbox{(sequence)}]{
\context{(P \cseq Q) \cpar (U \cseq V)}
}{
\context{(P \cpar U) \cseq (Q \cpar V)}
}
\\[10pt]
\hline
\\
\infer[\mbox{(external)}]{
\context{ (P \wwith Q) \cpar S }
}{
\context{ (P \cpar S) \wwith (Q \cpar S) }
}
\qquad\qquad
\infer[\mbox{(medial)}]{
\context{ \left(P \cseq Q\right) \wwith \left(U \cseq V\right) }
}{
\context{ \left(P \wwith U\right) \cseq \left(Q \wwith V\right) }
}
\\[10pt]
\infer[\mbox{(tidy)}]{
\context{ \cunit \wwith \cunit }
}{
\context{ \cunit }
}
\qquad\qquad
\infer[\mbox{(left)}]{
\context{ P \ooplus Q }
}{
\context{ P}
}
\qquad\qquad
\infer[\mbox{(right)}]{
\context{ P \ooplus Q }
}{
\context{ Q}
}
\\[10pt]
\hline
\\
\infer[\mbox{(extrude1)}]{
\context{ \forall x. P \cpar R }
}{
\context{ \mathopen{\forall x.}\left( P \cpar R \right) }
}
\qquad\qquad\quad
\infer[\mbox{(medial1)}]{
\context{ \mathopen{\forall x .}\left(P \cseq S \right) }
}{
\context{ \forall x .P \cseq \forall x .S }
}
\\[10pt]
\infer[\mbox{(tidy1)}]{
\context{ \forall x. \cunit }
}{
\context{ \cunit }
}
\qquad\qquad\quad
\infer[\mbox{(select1)}]{
\context{ \exists x. P }
}{
\context{ P \sub{x}{t} }
}
\\[10pt]
\hline
\\
\infer[\mbox{(extrude new)}]{
\context{ \new x .P \cpar R }
}{
\context{ \mathopen{\new x .}\left( P \cpar R \right) }
}
\qquad
\infer[\mbox{(medial new)}]{
\context{ \mathopen{\new x .}\left( P \cseq S \right) }
}{
\context{ \new x. P \cseq \new x. S }
}
\\[10pt]
\infer[\mbox{(tidy name)}]{
\context{ \new x. \cunit }
}{
\context{ \cunit }
}
\qquad
\infer[\mbox{(close)}]{
\context{ \new x. P \cpar \wen x. Q }
}{
\context{ \mathopen{\new x .}\left( P \cpar Q \right) }
}
\\[10pt]
\infer[\mbox{(fresh)}]{
\context{ \wen x. P }
}{
\context{ \new x. P }
}
\qquad
\infer[\mbox{(new wen)}]{
\context{ \new x .\wen y .P }
}{
\context{ \wen y. \new x .P }
}
\qquad
\infer[\mbox{(all name)}]{
\context{ \forall x .\quantifier y. P }
}{
\context{ \quantifier y. \forall x .P }
}
\\[10pt]
\infer[\mbox{(suspend)}
]{
\context{ \wen x. P \odot \wen x. S }
}{
\context{ \mathopen{\wen x .}\left( P \odot S \right) }
}
\qquad
\infer[\mbox{(left wen)}
]{
\context{ \wen x .P \odot R }
}{
\context{ \mathopen{\wen x .}\left( P \odot R \right) }
}
\qquad
\infer[\mbox{(right wen)}
]{
\context{ R \odot \wen x. Q }
}{
\context{ \mathopen{\wen x .}\left( R \odot Q \right) }
}
\\[10pt]
\infer[\mbox{(with name)}
]{
\context{ \quantifier x .P \wwith \quantifier x. S }
}{
\context{ \mathopen{\quantifier x .}\left( P \wwith S \right) }
}
\quad
\infer[\mbox{(left name)}
]{
\context{ \quantifier x .P \wwith R }
}{
\context{ \mathopen{\quantifier x .}\left( P \wwith R \right) }
}
\quad
\infer[\mbox{(right name)}]{
\context{ R \wwith \quantifier x. Q }
}{
\context{ \mathopen{\quantifier x .}\left( R \wwith Q \right) }
}
\\\\
\end{array}
\]
\begin{flushright}
where $\protect{\quantifier \in \left\{ \new, \wen \right\}}$, $\odot \in \left\{ \cpar, \cseq \right\}$ and $\nfv{x}{R}$, in all rules containing $R$
\end{flushright}
\caption[Rules for formulae in system \textsf{MAV1}.]{Rules for formulae in system \textsf{MAV1}. 
Notice the figure is divided into four parts. The first part defines sub-system \textsf{BV}~\cite{Guglielmi2007}. The first and second parts define sub-system \textsf{MAV}~\cite{Horne2015}.
}
\label{figure:rewrite}
\end{figure}

We present the syntax and rules of a first-order system expressed in the calculus of structures, with the technical name \textsf{MAV1}.
The derivations of the system are defined by the \textit{abstract syntax} in Fig.~\ref{figure:syntax}, \textit{structural congruence} in Fig.~\ref{figure:congruence},
and the \textit{inference rules}, in Fig~\ref{figure:rewrite}.
We emphasise that, in contrast to the sequent calculus, 
rules can be applied in any context, i.e.\ \textsf{MAV1} formulae from Fig.~\ref{figure:syntax} with a hole of the form 
\[
\context{} \Coloneqq \ehole{\,\cdot\,} \mid \context{} \odot P \mid P \odot \context{} \mid \quantifier x . \context{}, 
\mbox{ where $\odot \in \left\{ \cseq, \cpar, \tensor, \wwith, \ooplus \right\}$ and $\quantifier \in \left\{ \exists, \forall, \new, \wen \right\}$. }
\] 
We also assume the standard notion of capture avoiding substitution of a variable for a term.
Terms may be constructed from variables, constants and function symbols.

To explore the theory of proofs, two auxiliary definitions are introduced: linear negation and linear implication. 
Notice in the syntax in Fig.~\ref{figure:syntax} linear negation applies only to atoms.
\begin{definition}
\textit{Linear negation} is defined by the following function from formulae to formulae.
\begin{gather*}
\co{\co{\alpha}} = \alpha
\qquad
\co{ P \tensor Q } = \co{P} \cpar \co{Q}
\qquad
\co{ P \cpar Q } = \co{P} \tensor \co{Q}
\qquad
\co{ P \ooplus Q } = \co{P} \wwith \co{Q}
\qquad
\co{ P \wwith Q } = \co{P} \ooplus \co{Q}
\\
\co{\cunit} = \cunit
\quad
\co{P \cseq Q} = \co{P} \cseq \co{Q}
\qquad
\co{\forall x. P} = \exists x. \co{P}
\qquad
\co{\exists x. P} = \forall x. \co{P}
\qquad
\co{\new x. P} = \wen x. \co{P}
\qquad
\co{\wen x. P} = \new x. \co{P}
\end{gather*}
\textit{Linear implication}, written $P \multimap Q$, is defined as $\co{P} \cpar Q$.
\end{definition}
Linear negation defines de Morgan dualities.
As in linear logic, the multiplicatives $\tensor$ and $\cpar$ are de Morgan dual; as are the additives $\wwith$ and $\ooplus$, the first-order quantifiers $\exists$ and $\forall$, and the nominal quantifiers $\new$ and $\wen$. As in \textsf{BV}, \textit{seq} and the unit are self-dual.

A basic, but essential, property of implication can be established immediately.
The following proposition is simply a reflexivity property of linear implication in \textsf{MAV1}.
\begin{proposition}[Reflexivity]\label{proposition:reflexivity}
For any formula $P$, $\vdash \co{P} \cpar P$ holds, i.e., $\vdash P \multimap P$.
\end{proposition}
The proof of the above follows by a straightforward induction over the structure of $P$.

\subsection{Intuitive explanations for the rules of \textsf{MAV1}.}

We briefly recall the established system \textsf{MAV},
before explaining the rules for quantifiers.
This paper focuses on necessary proof theoretical prerequisites, and hints at result for process embeddings in \textsf{MAV1}.
Details on the soundness of process embeddings appear in a companion paper~\cite{MSCS}.

\textbf{The additives.}
The rules of the basic system \textsf{BV} in the top part of Fig.~\ref{figure:rewrite} are as described previously in Section~\ref{section:self-dual}.
The first and second parts of Fig.~\ref{figure:rewrite} define multiplicative-additive system \textsf{MAV}~\cite{Horne2015}.
The additives are useful for modelling non-deterministic choice in processes~\cite{Abramsky1993}: the \textit{left} rule $\vcenter{\infer{P \ooplus Q}{P}}$ suggests we chose the left branch $P$ \textbf{or} alternatively the right branch $Q$ by using the \textit{right} rule; the \textit{external} rule 
$\vcenter{\infer{\left(P \wwith Q\right) \cpar R}{\left(P  \cpar R\right) \wwith \left(Q \cpar R\right)}}$
suggests that we try both branches $P \cpar R$ \textbf{and} $Q \cpar R$ separately; and the \textit{tidy} rule indicates a derivation is successfully only if both branches explored are successful.
The \textit{medial} rule is a partial distributivity property between the additives and \textit{seq} (in concurrency theory, this is a property expected of most preorders over processes).
The role of the additives as a form of \textit{internal} and \textit{external} choice has been investigated in related work~\cite{Ciobanu2015}.

\textbf{The first-order quantifiers.}
The rules for the first-order quantifiers in the third part of Fig.~\ref{figure:rewrite} follow a similar pattern to the additives.
The \textit{select1} rule allows a variable to be replaced by any term. Notice we stick to the first-order case, since variables only appear in atomic formulae and may only be replaced by terms.
The \textit{extrude1}, \textit{tidy1} and \textit{medial1} rules follow a similar pattern to the rules for the additives \textit{external}, \textit{tidy} and \textit{medial} respectively.
In process embeddings, first-order quantifiers are useful as input binders.
For example we can soundly embed the $\pi$-calculus process $\co{y}z \pipar y(x).\co{x}w \pipar z(x)$ as the following provable formula:
\[
\infer[\mbox{by \textit{select1}}]{
\co{\mathrm{act}(y, z)} \cpar \mathopen{\exists x.}\left(\mathrm{act}(y, x)\cseq \co{\mathrm{act}(x, w)}\right) \cpar \exists v.\mathrm{act}(z, v)
}{
\infer[\mbox{by \textit{sequence}}]{
\co{\mathrm{act}(y, z)} \cpar \left(\mathrm{act}(y, z)\cseq \co{\mathrm{act}(z, w)}\right) \cpar \exists v.\mathrm{act}(z, v)
}{
\infer[\mbox{by \textit{atomic interaction}} ]{
\left(\left(\co{\mathrm{act}(y, z)} \cpar \mathrm{act}(y, z)\right)\cseq \co{\mathrm{act}(z, w)}\right) \cpar \exists v.\mathrm{act}(z, v) 
}{
\infer[\mbox{by \textit{select1}}]{
\co{\mathrm{act}(z, w)} \cpar \exists v.\mathrm{act}(z, v)
}{
\infer[\mbox{by \textit{atomic interaction}}]{
\co{\mathrm{act}(z, w)} \cpar \mathrm{act}(z, w)
}{
\cunit
}}}}}
\]
Notice, that the above process can also reach a successfully terminated state using $\tau$ transitions in the $\pi$-calculus semantics.
Indeed the cut elimination result established in this paper is a prerequisite in order to prove this soundness criterion holds for finite $\pi$-calculus processes. 

\textbf{The polarised nominal quantifiers.}
The rules for the de Morgan dual pair of nominal quantifiers are more intricate.
For first-order quantifiers many properties are derivable, e.g., the following implications hold (appealing to Prop.~\ref{proposition:reflexivity}):
$\vdash \forall x. \forall y. P \multimap \forall y. \forall x. P$,
$\vdash \exists x. \forall y. P \multimap \forall y. \exists x. P$
and
$\vdash \mathopen{\forall x.}\left( P \cpar Q \right) \multimap \forall x. P \cpar \exists x. Q$.
The three proofs proceed as follows.
\[
\infer[]{
\exists x. \exists y. \co{P} \cpar \forall y. \forall x. P
}{
\infer[]{
\mathopen{\forall y. \forall x.}\left(\exists x. \exists y. \co{P} \cpar  P\right)
}{
\infer[]{
\mathopen{\forall y. \forall x.}\left(\co{P} \cpar  P\right)
}{
\infer[]{
\mathopen{\forall y. \forall x.}\cunit
}{
\cunit
}}}}
\qquad\qquad
\infer[]{
\forall x. \exists y. \co{P} \cpar \forall y. \exists x. P
}{
\infer[]{
\mathopen{\forall x. \forall y.} \left(\exists y. \co{P} \cpar \exists x. P\right)
}{
\infer[]{
\mathopen{\forall x. \forall y.} \left(\co{P} \cpar P\right)
}{
\infer[]{
\mathopen{\forall x. \forall y.} \cunit
}{
\cunit
}}}}
\qquad\qquad
\infer[]{
\mathopen{\exists x.}\left( \co{P} \tensor \co{Q} \right) \cpar \forall x. P \cpar \exists x. Q
}{
\infer[]{
\mathopen{\forall x.}\left(\mathopen{\exists x.}\left( \co{P} \tensor \co{Q} \right) \cpar P \cpar \exists x. Q\right)
}{
\infer[]{
\mathopen{\forall x.}\left( \co{P \cpar Q} \cpar P \cpar Q\right)
}{
\infer[]{
\mathopen{\forall x.}\cunit
}{
\cunit
}}}}
\]
We desire analogous properties for the nominals $\new$ and $\wen$.
However, in contrast to first-order quantifiers, these properties must be induced for our pair of nominals. 
The first property is induced for $\new$ and $\wen$ by \textit{equivariance} in the structural congruence. 
The other rules analogous to the above derived implications are induced by the rules: \textit{new wen}, which allow a weaker quantifier $\wen$ to commute over a stronger quantifier $\new$;
and \textit{close} which models that $\wen$ can select a name as long as it is fresh as indicated by $\new$.

We avoid \textit{new} distributing over $\cpar$, i.e., in general \textbf{neither} 
$
\mathopen{\new x.} \left(P \cpar Q \right)
\multimap
\new x.P \cpar \new x.Q
$
\textbf{nor}
$
\new x.P \cpar \new x.Q
\multimap
\mathopen{\new x.} \left(P \cpar Q \right)
$
hold.
Hence $\new$ is suitable for embedding the name binder $\nu$ of the $\pi$-calculus.
Interestingly, the dual quantifier $\wen$ is also useful for embedding a variant of the $\pi$-calculus called the $\pi I$-calculus, where every communication creates a new fresh name.
For example, $\pi I$-calculus process $\mathopen{\oprivate{v}{x}}.\iprivate{x}{y} \pipar \mathopen{\iprivate{v}{z}}.\oprivate{z}{w}$ can be embedded as the following provable formula.\footnote{
To disambiguate from the $\pi$-calculus we use square brackets as binders for the $\pi I$-calculus.
So $\mathopen{\oprivate{v}{x}}.P$ denotes a process that outputs a fresh name $x$
and $\mathopen{\iprivate{v}{x}}.P$ denotes a process that receives a name $x$ only if it is fresh.
}
\[
\infer[\mbox{by \textit{close} and $\alpha$-conversion}]{
\mathopen{\new x.}\left( \co{\mathrm{act}(v, x)} \cseq \wen y.\mathrm{act}(x, y)\right) \cpar \mathopen{\wen z.}\left(\mathrm{act}(v, z) \cseq \new w.\co{\mathrm{act}(z, w)}\right)
}{
\infer[\mbox{by \textit{sequence}}]{
\mathopen{\new x.}\left(\left( \co{\mathrm{act}(v, x)} \cseq \wen y.\mathrm{act}(x, y)\right) \cpar \left(\mathrm{act}(v, x) \cseq \new w.\co{\mathrm{act}(x, w)}\right)\right)
}{
\infer[\mbox{by \textit{atomic interaction}}]{
\mathopen{\new x.}\left(\left( \co{\mathrm{act}(v, x)} \cpar \mathrm{act}(v, x)\right) \cseq \left(\wen y.\mathrm{act}(x, y) \cpar \new w.\co{\mathrm{act}(x, w)}\right)\right)
}{
\infer[\mbox{by \textit{close}}]{
\mathopen{\new x.}\left(\wen y.\mathrm{act}(x, y) \cpar \new w.\co{\mathrm{act}(x, w)}\right)
}{
\infer[\mbox{by \textit{atomic interaction}}]{
\mathopen{\new x. \new w.}\left(\mathrm{act}(x, w) \cpar \co{\mathrm{act}(x, w)}\right)
}{
\infer[\mbox{by \textit{tidy name}}]{
\mathopen{\new x. \new w.}\cunit
}{
\cunit
}}}}}}
\]
Note that $\alpha$-renaming is implicitly applied in the derivation above. 

There is no \textit{vacuous} rule in Fig.~\ref{figure:syntax}, in contrast to the presentation of \textsf{BVQ} in Fig.~\ref{figure:BVQ}.
This is because the \textit{vacuous} rule creates problems for proof search, since arbitrarily many nominal quantifiers can be introduced at any point in the proof leading to unnecessary infinite search paths. Instead we build the introduction and elimination of fresh names into rules only where required.
For example, \textit{extrude new} is like \textit{close} with a vacuous $\wen$ implicitly introduced;
similarly, for \textit{left wen}, \textit{right wen}, \textit{left name} and \textit{right name} a vacuous $\wen$ is implicitly introduced.
Also the \textit{tidy name} allows vacuous $\new$ operators to be removed from a successful proof in order to terminate with $\cunit$ only.
The reason why the rules \textit{medial new}, \textit{suspend}, \textit{all name} and \textit{with name} are required are in order to make cut elimination work; hence we postpone their explanation until after the statement of the cut elimination result.

In addition to forbidding the \textit{vacuous} rule, 
the following restrictions are placed on the rules to avoid meaningless infinite paths in proof search.
\begin{itemize}
\item For the \textit{switch}, \textit{sequence}, \textit{medial1}, \textit{medial new} and \textit{extrude new} rules, $P \not \equiv \cunit$ and $S \not \equiv \cunit$.
\item The \textit{medial} rule is such that either $P \not\equiv \cunit$ or $R \not\equiv \cunit$
and also either $Q \not\equiv \cunit$ or $S \not\equiv \cunit$.
\item The rules \textit{external}, \textit{extrude1}, \textit{extrude new}, \textit{left wen} and \textit{right wen} are such that $R \not \equiv \cunit$.
\end{itemize}
Avoiding infinite search paths is important for the termination of our cut elimination procedure.
Essentially, we desire that our system for \textsf{MAV1} is in a sense \textit{analytic}~\cite{Bruscoli2009}.

\paragraph{Note on term ``medial''} Medials were introduced, historically, to make contraction local (reducing contraction to a rule acting only over atoms)~\cite{Brunnler2001}.
Although the rules in Fig.~\ref{figure:rewrite} do not define such a local system, we discovered these rules by first
defining a local system, and then designing a more controlled system retaining only the medials of the local system that are not admissible.
Related work~\cite{Tubella2018} shows that medials are a ubiquitous recipe underlying the rules of proof systems.

%% file: implication.tex
\subsection{Cut elimination and its consequences}
\label{section:implication}

This section confirms that the rules of \textsf{MAV1} indeed define a logical system, as established by a cut elimination theorem.
Surprisingly, prior to this work, the only direct proof of cut elimination involving quantifiers in the calculus of structures was for \textsf{BVQ}~\cite{Roversi14}.
Related cut elimination results involving first-order quantifiers in the calculus of structures relied on a correspondence with the sequent calculus~\cite{Brunnler2006,Strassburger09}.
However, due to the presence of the non-commutative operator \textit{seq} there is no sequent calculus presentation~\cite{Tiu2006} for \textsf{MAV1}; hence we pursue here a direct proof.

The main result of this paper is the following, which is a generalisation of \textit{cut elimination} to the setting of the calculus of structures.
\begin{theorem}[Cut elimination]\label{theorem:cut}
For any formula $P$, if $\vdash \context{ P \tensor \co{P} }$ holds,
then $\vdash \context{ \cunit }$ holds.
\end{theorem}
The above theorem can be stated alternatively by supposing that there is a proof in \textsf{MAV1} extended with the extra inference rule: 
$\vcenter{\infer[\mbox{(cut)}]{\context{\cunit}}{\context{P \tensor \co{P}}}}.$
Given such a proof, a new proof can be constructed that uses only the rules of \textsf{MAV1}.
In this formulation, we say that \textit{cut} is \textit{admissible}.

Cut elimination for the propositional sub-system \textsf{MAV} has been previously established~\cite{Horne2015}.
The current paper advances cut-elimination techniques to tackle first-order system \textsf{MAV1}, as achieved by the lemmas in later sections.
Before proceeding with the necessary lemmas, we provide a corollary that demonstrates that one of many consequences of cut elimination is indeed that linear implication defines a precongruence --- a reflexive transitive relation that holds in any context.
\begin{corollary}~\label{corollary:precongruence}
Linear implication defines a precongruence.
\end{corollary}
\begin{proof}
For transitivity, if $\vdash P \multimap Q$ and $\vdash Q \multimap R$ hold, we have the following.
\[
\infer[\mbox{by the \textit{switch} rule}]{
 \left(\co{P} \cpar \left( Q \tensor \co{Q} \right) \cpar R \right)
}{
\infer[\mbox{by the assumptions $\vdash \co{P} \cpar Q$ and $\vdash \co{Q} \cpar R$}]{
 \left(\co{P} \cpar Q \right) \tensor \left( \co{Q} \cpar R \right)
}{
\cunit
}}
\]
Hence, by Theorem~\ref{theorem:cut}, $\vdash P \multimap R$ as required.

For contextual closure, if $\vdash P \multimap Q$ holds, we have the following.
\[
\infer[\mbox{by the \textit{switch} rule}]{
\co{\context{P}} \cpar \context{ \left( P \tensor \co{P} \right) \cpar Q}
}{
\infer[\mbox{by the assumption $\vdash P \multimap Q$}]{
\co{\context{P}} \cpar \context{ P \tensor \left( \co{P}  \cpar Q \right) } 
}{
\infer[\mbox{by Proposition~\ref{proposition:reflexivity}}]{
\co{\context{P}} \cpar \context{ P } 
}{
 \cunit
}}}
\]
Hence by Theorem~\ref{theorem:cut}, $\vdash \context{P} \multimap \context{Q}$ as required.
Reflexivity holds by Proposition~\ref{proposition:reflexivity}.
\end{proof}

%% file: nominal.tex
\subsection{Discussion on logical properties of the rules for nominal quantifiers}
\label{section:nominal}

The rules for the nominal quantifiers \textit{new} and \textit{wen} require justification.
The \textit{close} and \textit{tidy name} rules ensure the reflexivity of implication for nominal quantifiers.
Using the \textit{extrude new} rule (and Proposition~\ref{proposition:reflexivity}) we can establish the following proof of $\vdash \wen x. P \multimap \exists x. P$.
\[
\infer[\mbox{by the \textit{extrude new} rule}]{
 \exists x. P \cpar \new x. \co{P}
}{
\infer[\mbox{by the \textit{select1} rule}]{
 \mathopen{\new x.} \left(\exists x .P \cpar \co P\right)
}{
\infer[\mbox{by Proposition~\ref{proposition:reflexivity}}]{
 \mathopen{\new x.} \left( P \cpar \co P \right)
}{
\infer[\mbox{by the \textit{tidy name} rule}]{
 \new x. \cunit
}{
 \cunit
}}}}
\]
The above also serves as a proof of the dual statement $\vdash \forall x. P \multimap \new x. P$.

Using the \textit{fresh} rule we can establish the following implication $\vdash \new x. P \multimap \wen x. P$, as follows.
\[
\infer[\mbox{by the \textit{fresh} rule}]{
 \wen x .\co{P} \cpar \wen x .P
}{
\infer[\mbox{by Proposition~\ref{proposition:reflexivity}}]{
 \new x .\co{P} \cpar \wen x .P
}{
 \cunit
}}
\]
This completes the chain $\vdash \forall x. P \multimap \new x. P$, $\vdash \new x. P \multimap \wen x .P$ and $\vdash \wen x. P \multimap \exists x. P$. These linear implications are strict unless $\nfv{x}{P}$, in which case, for $\quantifier \in \left\{ \forall, \exists, \new, \wen \right\}$, $\quantifier x .P$ is logically equivalent to $P$. For example, using the \textit{fresh} rule followed by the \textit{extrude new} and \textit{tidy name} rules, $\vdash
 \new x .P \multimap P$
 holds, whenever $\nfv{x}{P}$.
Thus the implication corresponding to the \textit{vacuous} rule as in Fig.~\ref{figure:BVQ} is provable for any quantifier.

\textbf{The medial rules for nominals.}
The \textit{medial new} rule is particular to handling nominals in the presence of the self-dual non-commutative operator \textit{seq}. To see why this medial rule cannot be excluded, consider the following
formulae, where $x$ is free for atoms $\bbb$, $\ccc$, $\eee$ and $\ffff.$
\[
\begin{array}{rl}
 \left( \aaa \cseq \mathopen{\wen x.}\left( \bbb \cseq \ccc \right)\right) \tensor \left( \ddd \cseq \mathopen{\wen x.}\left( \eee \cseq \ffff \right) \right)
 \!\!\!&\multimap
 \left( \aaa \cseq \exists x. \bbb \cseq \exists x. \ccc \right) \tensor \left( \ddd \cseq \exists x. \eee \cseq \exists x .\ffff \right) 
\\
 \left( \aaa \cseq \exists x. \bbb \cseq \exists x .\ccc \right) \tensor \left( \ddd \cseq \exists x. \eee \cseq \exists x .\ffff \right) 
 \!\!\!&\multimap
 \left(\left( \aaa \cseq \exists x .\bbb \right) \tensor \left( \ddd \cseq \exists x. \eee \right)\right) \cseq \left( \exists x. \ccc \tensor \exists x. \ffff \right)
\end{array}
\]
Without using the \textit{medial new} rule, the above formulae are provable.
The first is as follows.
\[
\infer[\mbox{\textit{switch}}]{
 \left( \co{\aaa} \cseq \mathopen{\new x.}\left( \co{\bbb} \cseq \co{\ccc} \right)\right) \cpar \left( \co{\ddd} \cseq \mathopen{\new x.}\left( \co{\eee} \cseq \co{\ffff} \right) \right)
 \cpar
 \left( \aaa \cseq \exists x. \bbb \cseq \exists x. \ccc \right) \tensor \left( \ddd \cseq \exists x. \eee \cseq \exists x .\ffff \right) 
}{
\infer[\mbox{\textit{sequence}}]{
 \left(
  \left( \co{\aaa} \cseq \mathopen{\new x.}\left( \co{\bbb} \cseq \co{\ccc} \right)\right) \cpar \left( \aaa \cseq \exists x. \bbb \cseq \exists x. \ccc \right)
 \right)
 \tensor
 \left(
  \left( \co{\ddd} \cseq \mathopen{\new x.}\left( \co{\eee} \cseq \co{\ffff} \right) \right) \cpar \left( \ddd \cseq \exists x. \eee \cseq \exists x .\ffff \right) 
 \right)
}{
\infer[\mbox{\textit{extrude}}]{
 \left(
  \left( \co{\aaa} \cpar \aaa \right) \cseq \left( \mathopen{\new x.}\left( \co{\bbb} \cseq \co{\ccc} \right) \cpar \left( \exists x. \bbb \cseq \exists x. \ccc \right)\right)
 \right)
 \tensor
 \left(
  \left( \co{\ddd} \cpar \ddd \right) \cseq \left(\mathopen{\new x.}\left( \co{\eee} \cseq \co{\ffff} \right) \cpar \left( \exists x. \eee \cseq \exists x .\ffff \right) \right)
 \right)
}{
\infer[\mbox{\textit{select1}}]{
 \left(
  \left( \co{\aaa} \cpar \aaa \right) \cseq  \mathopen{\new x.}\left(\left( \co{\bbb} \cseq \co{\ccc} \right) \cpar \left( \exists x. \bbb \cseq \exists x. \ccc \right)\right)
 \right)
 \tensor
 \left(
  \left( \co{\ddd} \cpar \ddd \right) \cseq \mathopen{\new x.}\left(\left( \co{\eee} \cseq \co{\ffff} \right) \cpar \left( \exists x. \eee \cseq \exists x .\ffff \right) \right)
 \right)
}{
\infer[\mbox{by Proposition~\ref{proposition:reflexivity}}]{
 \left(
  \left( \co{\aaa} \cpar \aaa \right) \cseq  \mathopen{\new x.}\left(\left( \co{\bbb \cseq \ccc} \right) \cpar \left( \bbb \cseq \ccc \right)\right)
 \right)
 \tensor
 \left(
  \left( \co{\ddd} \cpar \ddd \right) \cseq \mathopen{\new x.}\left(\co{\eee \cseq \ffff} \cpar \left( \eee \cseq \ffff \right) \right)
 \right)
}{
\infer[\mbox{by \textit{tidy name}}]{
 \left(
  \mathopen{\new x.}\cunit
 \right)
 \tensor
 \left(
  \mathopen{\new x.}\cunit
 \right)
}{
\cunit
}}}}}}
\]
The proof of the second formula above is as follows.
\[
\infer[\mbox{by \textit{sequence}}]{
 \left( \co{\aaa} \cseq \forall x. \co{\bbb} \cseq \forall x .\co{\ccc} \right) \cpar \left( \co{\ddd} \cseq \forall x. \co{\eee} \cseq \forall x .\co{\ffff} \right) 
 \cpar
 \left(\left( \aaa \cseq \exists x .\bbb \right) \tensor \left( \ddd \cseq \exists x. \eee \right)\right) \cseq \left( \exists x. \ccc \tensor \exists x. \ffff \right)
\\\quad
}{
\infer[\mbox{by \textit{sequence}}]{
 \left(\left( \co{\aaa} \cseq \forall x. \co{\bbb}\right) \cpar \left(\co{\ddd} \cseq \forall x. \co{\eee}\right) \right) \cseq \left( \forall x .\co{\ccc} \cpar \forall x .\co{\ffff} \right) 
 \cpar
 \left(\left( \aaa \cseq \exists x .\bbb \right) \tensor \left( \ddd \cseq \exists x. \eee \right)\right) \cseq \left( \exists x. \ccc \tensor \exists x. \ffff \right)
}{
\infer[\mbox{by Prop.~\ref{proposition:reflexivity}}]{
 \left( \co{\left( \aaa \cseq \exists x .\bbb \right) \tensor \left( \ddd \cseq \exists x. \eee \right)} \cpar \left(\left( \aaa \cseq \exists x .\bbb \right) \tensor \left( \ddd \cseq \exists x. \eee \right)\right)\right) \cseq \left( \co{\exists x. \ccc \tensor \exists x. \ffff} \cpar \left(\exists x. \ccc \tensor \exists x. \ffff \right)\right)
}{
 \cunit
}}}
\]

However, the issue is that
the following formula would not be provable without using the \textit{medial new} rule; hence cut elimination cannot hold without the \textit{medial new} rule.
\[
\begin{array}{l}
 \left( \aaa \cseq \mathopen{\wen x.}\left( \bbb \cseq \ccc \right)\right) \tensor \left( \ddd \cseq \mathopen{\wen x.}\left( \eee \cseq \ffff \right) \right)
 \multimap
 \left(\left( \aaa \cseq \exists x. \bbb \right) \tensor \left( \ddd \cseq \exists x. \eee \right)\right) \cseq \left( \exists x .\ccc \tensor \exists x. \ffff \right)
\end{array}
\]
In contrast, with the \textit{medial new} rule the above formula is provable, as verified by the proof in Figure~\ref{fig:medial-ex}.
Notice the above proofs use only the \textit{medial new}, \textit{extrude new} and \textit{tidy name} rules for nominals. These rules are of the same form as rules \textit{medial1}, \textit{extrude1} and \textit{tidy1} for universal quantifiers, hence the same argument holds for the necessity of the \textit{medial1} rule by replacing $\new$ with $\forall$.

\begin{figure}[h!]
\[
\infer[]{
 \left( \co{\aaa} \cseq \new x. \left( \co{\bbb} \cseq \co{\ccc} \right)\right) \cpar \left( \co{\ddd} \cseq \new x. \left( \co{\eee} \cseq \co{\ffff} \right) \right)
 \cpar 
 \left(\left( \aaa \cseq \exists x .\bbb \right) \tensor \left( \ddd \cseq \exists x .\eee \right)\right) \cseq
 \left( \exists x .\ccc \tensor \exists x. \ffff \right)
}{
\infer[]{
  \left( \co{\aaa} \cseq \new x. \co{\bbb} \cseq \new x. \co{\ccc} \right) \cpar \left( \co{\ddd} \cseq \new x. \co{\eee} \cseq \new x .\co{\ffff} \right)
 \cpar 
 \left( \left( \aaa \cseq \exists x. \bbb \right) \tensor \left( \ddd \cseq \exists x. \eee \right) \right) \cseq
 \left( \exists x.\ccc \tensor \exists x. \ffff \right)
}{
\infer[]{
  \left( \left( \co{\aaa} \cseq \new x. \co{\bbb} \right) \cpar \left( \co{\ddd} \cseq \new x. \co{\eee} \right)\right)
  \cseq
  \left( \new x. \co{\ccc} \cpar \new x .\co{\ffff} \right)
  \cpar
  \left( \left( \aaa \cseq \exists x. \bbb \right) \tensor \left( \ddd \cseq \exists x. \eee \right) \right) \cseq
  \left( \exists x .\ccc \tensor \exists x. \ffff \right)
}{
\infer[]{
  \left( \left( \co{\aaa} \cseq \new x .\co{\bbb} \right) \cpar \left( \co{\ddd} \cseq \new x. \co{\eee} \right)
  \cpar
  \left( \left( \aaa \cseq \exists x. \bbb \right) \tensor \left( \ddd \cseq \exists x .\eee \right) \right) \right) 
  \cseq 
  \left( \new x. \co{\ccc} \cpar \new x. \co{\ffff} \cpar \left( \exists x.\ccc \tensor \exists x. \ffff \right) \right)
}{
\infer[]{
 \left( \left( \left( \co{\aaa} \cseq \new x. \co{\bbb} \right) \cpar \left( \aaa \cseq \exists x .\bbb \right) \right)
  \tensor
  \left(  \left( \co{\ddd} \cseq \new x .\co{\eee} \right) \cpar \left( \ddd \cseq \exists x. \eee \right) \right) \right) 
 \cseq 
  \left(  \left( \new x. \co{\ccc} \cpar \exists x .\ccc \right) \tensor \left( \new x. \co{\ffff} \cpar \exists x. \ffff \right) \right)
}{
\infer[]{
  \left( \left( \left( \co{\aaa} \cpar \aaa \right) \cseq \left( \new x .\co{\bbb} \cpar \exists x .\bbb \right) \right)
  \tensor
  \left(  \left( \co{\ddd} \cpar \ddd \right) \cseq \left( \new x .\co{\eee} \cpar \exists x .\eee \right) \right) \right) 
  \cseq 
  \left(  \left( \new x. \co{\ccc} \cpar \exists x .\ccc \right) \tensor \left( \new x. \co{\ffff} \cpar \exists x. \ffff \right) \right)
}{
\infer[]{
  \left( \left( \left( \co{\aaa} \cpar \aaa \right) \cseq \mathopen{\new x.}\left(  \co{\bbb} \cpar \exists x .\bbb \right) \right)
  \tensor
  \left(  \left( \co{\ddd} \cpar \ddd \right) \cseq  \mathopen{\new x.} \left(\co{\eee} \cpar \exists x. \eee \right) \right) \right) 
  \cseq 
  \left(  \mathopen{\new x.} \left( \co{\ccc} \cpar \exists x. \ccc \right) \tensor \mathopen{\new x.}\left( \co{\ffff} \cpar \exists x .\ffff \right) \right)
}{
\infer[]{
  \left( \left( \left( \co{\aaa} \cpar \aaa \right) \cseq \mathopen{\new x.} \left(  \co{\bbb} \cpar \bbb \right) \right)
  \tensor
  \left(  \left( \co{\ddd} \cpar \ddd \right) \cseq  \mathopen{\new x.}\left(\co{\eee} \cpar \eee \right) \right) \right) 
  \cseq
  \left(  \mathopen{\new x.}\left( \co{\ccc} \cpar \ccc \right) \tensor \mathopen{\new x.}\left( \co{\ffff} \cpar \ffff \right) \right)
}{
\infer[]{
  \left(
  \new x .\cunit
  \tensor
  \new x .\cunit \right) 
  \cseq
  \left(  \new x. \cunit \tensor \new x. \cunit \right)    
}{
  \cunit
}}}}}}}}}
 \]
 \caption{A proof of  $ \left( \aaa \cseq \mathopen{\wen x.}\left( \bbb \cseq \ccc \right)\right) \tensor \left( \ddd \cseq \mathopen{\wen x.}\left( \eee \cseq \ffff \right) \right)
 \multimap
 \left(\left( \aaa \cseq \exists x. \bbb \right) \tensor \left( \ddd \cseq \exists x. \eee \right)\right) \cseq \left( \exists x .\ccc \tensor \exists x. \ffff \right)$
}
 \label{fig:medial-ex}
 \end{figure}

Including the \textit{medial new} rule forces the \textit{suspend} rule to be included. To see why, observe that the following linear implications are provable.
\[
\begin{array}{rl}
\left( \new x. \aaa \cseq \new x. \bbb \right) \tensor \left( \new x .\ccc \cseq \new x .\ddd \right)
\multimap&
\!\!\!\!
\mathopen{\new x.}\left( \aaa \cseq \bbb \right) \tensor \mathopen{\new x.}\left( \ccc \cseq \ddd \right)
\\
\mathopen{\new x.}\left( \aaa \cseq \bbb \right) \tensor \mathopen{\new x.}\left( \ccc \cseq \ddd \right)
\multimap&
\!\!\!\!
\mathopen{\new x.}\left( \left( \aaa \cseq \bbb \right) \tensor \left( \ccc \cseq \ddd \right) \right)
\end{array}
\]
However, without the \textit{suspend} rule the following implication is not provable, which would contradict the cut elimination result of this paper.
\[
\left( \new x .\aaa \cseq \new x .\bbb \right) \tensor \left( \new x. \ccc \cseq \new x. \ddd \right)
\multimap
\mathopen{\new x.}\left( \left( \aaa \cseq \bbb \right) \tensor \left( \ccc \cseq \ddd \right) \right)
\]
Fortunately, including the \textit{suspend} rule ensures that the above implication is provable as follows.
\[
\infer[\mbox{by \textit{suspend}}]{
\left( \wen x .\co{\aaa} \cseq \wen x .\co{\bbb} \right) \cpar \left( \wen x. \co{\ccc} \cseq \wen x. \co{\ddd} \right)
\cpar
\mathopen{\new x.}\left( \left( \aaa \cseq \bbb \right) \tensor \left( \ccc \cseq \ddd \right) \right)
}{
\infer[\mbox{by \textit{suspend}}]{
\mathopen{\wen x.}\left( \co{\aaa} \cseq \co{\bbb} \right) \cpar \mathopen{\wen x.}\left( \co{\ccc} \cseq \co{\ddd} \right)
\cpar
\mathopen{\new x.}\left( \left( \aaa \cseq \bbb \right) \tensor \left( \ccc \cseq \ddd \right) \right)
}{
\infer[\mbox{by Proposition~\ref{proposition:reflexivity}}]{
\mathopen{\wen x.}\left(\left( \co{\aaa} \cseq \co{\bbb} \right) \cpar \left( \co{\ccc} \cseq \co{\ddd} \right)\right)
\cpar
\mathopen{\new x.}\left( \left( \aaa \cseq \bbb \right) \tensor \left( \ccc \cseq \ddd \right) \right)
}{
\cunit
}}}
\]
A similar argument justifies the inclusion of the \textit{left wen} and \textit{right wen} rules.

\textbf{Rules induced by equivariance.}
Interestingly, \textit{equivariance} is a design decision in the sense that cut elimination still holds if we drop the \textit{equivariance} rule from the structural congruence.
For such a system without \textit{equivariance}, also the rules \textit{all name}, \textit{with name}, \textit{left name} and \textit{right name} could also be dropped.
Perhaps there may be interesting applications for a non-equivariant nominal quantifiers; however, for embedding of process such as $\nu$ in the $\pi$-calculus, \textit{equivariance} is an essential property for scope extrusion. For example, \textit{equivariance} is used when proving the embedding of labelled transition $\nu x. \nu y.\co{z}y.p \lts{\co{z}{(y)}} \nu x. p$, assuming $z \not= x$ and $z \not= y$.

In our embedding of the $\pi$-calculus in $\mathsf{MAV1}$, addressed thoroughly in a companion paper~\cite{MSCS}, we assume process $p$ is embedded as formula $P$.
In this case, process $\nu x. \nu y.\co{z}y.p$ maps to $Q = \mathopen{\new x. \new y.}\left(\co{\textrm{act}(z, y)} \cseq P\right)$, process $\nu x. p$ maps to $R = \new x. P$.
In this embedding of processes as formulae, we can prove that whenever the above labelled transition is enabled,
we can prove the following implication
$
\mathopen{\new y.}\left( \co{\textrm{act}(z, y)} \cseq {R} \right)
\multimap
Q
$, where the binder $\new y$ and atom $\textrm{act}(z, y)$ indicate that the process can commit to a bound output.
Indeed this formula is provable, as follows, by using \textit{equivariance}.
\[
\infer[\mbox{by \textit{equivariance}}]{
\mathopen{\wen y.}\left( \textrm{act}(z, y) \cseq \wen x.\co{P} \right)
\cpar
\mathopen{\new x. \new y.}\left(\co{\textrm{act}(z, y)} \cseq P\right)
}{
\infer[\mbox{by \textit{close}}]{
\mathopen{\wen y.}\left( \textrm{act}(z, y) \cseq \wen x.\co{P} \right)
\cpar
\mathopen{\new y. \new x.}\left(\co{\textrm{act}(z, y)} \cseq P\right)
}{
\infer[\mbox{by \textit{medial new}}]{
\mathopen{\new y.} \left(\left( \textrm{act}(z, y) \cseq \wen x.\co{P}\right)
\cpar
\mathopen{\new x.} \left(\co{\textrm{act}(z, y)} \cseq P\right) \right)
}{
\infer[\mbox{by \textit{sequence}}]{
\mathopen{\new y.} \left(\left( \textrm{act}(z, y) \cseq \wen x.\co{P}\right)
\cpar
 \left(\mathopen{\new x.}\co{\textrm{act}(z, y)} \cseq \mathopen{\new x.}P\right) \right)
}{
\infer[\mbox{by \textit{extrude new}}]{
\mathopen{\new y.} \left(
\left( \textrm{act}(z, y) \cpar \mathopen{\new x.}\co{\textrm{act}(z, y)} \right)
\cseq
 \left(\wen x.\co{P} \cpar \mathopen{\new x.}P\right) \right)
}{
\infer[\mbox{by Proposition~\ref{proposition:reflexivity}}]{
\mathopen{\new y.} \left(
\mathopen{\new x.}\left( \textrm{act}(z, y) \cpar \co{\textrm{act}(z, y)} \right)
\cseq
 \left(\co{\new x.P} \cpar \mathopen{\new x.}P\right) \right)
}{
\infer[\mbox{by \textit{tidy name}}]{
\mathopen{\new y.} 
\mathopen{\new x.}\cunit
}{
\cunit
}}}}}}}
\]
In response to the above problem, modelling the $\pi$-calculus, \textsf{MAV1} includes equivariance.

The \textit{equivariance} rule forces additional distributivity properties for $\new$ and $\wen$ over $\wwith$ and $\forall$, given by the \textit{all name}, \textit{with name}, \textit{left name}, \textit{right name} rules.
These rules allow $\new$ and $\wen$ quantifiers to propagate to the front of certain contexts.
To see why these rules are necessary consider the following implications,
with matching formulae, respectively, after and before the implication.
\[
\vdash
\mathopen{\new x.}\left(\new y. \forall z. \aaa \cpar \mathopen{\wen y.}\left( \bbb \wwith \ccc \right)\right)
\multimap
\new x.\new y. \forall z. \aaa \cpar \mathopen{\wen x. \wen y.}\left( \bbb \wwith \ccc \right)
\]
\[
\vdash
\new x.\new y. \forall z. \aaa \cpar \mathopen{\wen x. \wen y.}\left( \bbb \wwith \ccc \right)
\multimap
\new y. \forall z. \new x. \aaa \cpar \mathopen{\wen y.}\left( \wen x. \bbb \wwith \wen x. \ccc \right)
\]
Any proof of the second implication does involve \textit{equivariance}; but neither proof requires
\textit{all name} or \textit{with name}.
A proof of the first implication above is as follows.
\[
\infer[\mbox{by \textit{close}}]{
\mathopen{\wen x.}\left(\wen y. \exists z. \co{\aaa} \tensor \mathopen{\new y.}\left( \co{\bbb} \ooplus \co{\ccc} \right)\right)
\cpar
\new x.\new y. \forall z. \aaa \cpar \mathopen{\wen x. \wen y.}\left( \bbb \wwith \ccc \right)
}{
\infer[\mbox{by Proposition~\ref{proposition:reflexivity}}]{
\mathopen{\wen x.}\left(
\wen y. \exists z. \co{\aaa} \tensor \mathopen{\new y.}\left( \co{\bbb} \ooplus \co{\ccc} \right)\right)
\cpar
\mathopen{\new x.}\left(\new y. \forall z. \aaa \cpar \mathopen{\wen y.}\left( \bbb \wwith \ccc \right)\right)
}{
\cunit
}}
\]
A proof of the second implication above is given in Figure~\ref{fig:equiv-ex}.

\begin{figure}[h!]
\[
\infer[\mbox{by \textit{switch}}]{
\left(
\wen x.\wen y. \exists z. \co{\aaa} \tensor \mathopen{\new x. \new y.}\left( \co{\bbb} \ooplus \co{\ccc} \right)
\right)
\cpar
\new y. \forall z. \new x. \aaa \cpar \mathopen{\wen y.}\left( \wen x. \bbb \wwith \wen x.\ccc \right)
}{
\infer[\mbox{by \textit{equivariance} and \textit{close}}]{
\left(\wen x.\wen y. \exists z. \co{\aaa} \cpar \new y. \forall z. \new x. \aaa \right)
\tensor
\left(
\mathopen{\new x. \new y.}\left( \co{\bbb} \ooplus \co{\ccc} \right)
\cpar \mathopen{\wen y.}\left( \wen x. \bbb \wwith \wen x.\ccc \right)
\right)
}{
\infer[\mbox{by \textit{extrude1}}]{
\mathopen{\new y.}
\left(\wen x. \exists z. \co{\aaa} \cpar \forall z. \new x. \aaa \right)
\tensor
\left(
\mathopen{\new x. \new y.}\left( \co{\bbb} \ooplus \co{\ccc} \right)
\cpar \mathopen{\wen y.}\left( \wen x. \bbb \wwith \wen x.\ccc \right)
\right)
}{
\infer[\mbox{by \textit{close}}]{
\mathopen{\new y.\forall z. }
\left(\wen x. \exists z. \co{\aaa} \cpar \new x. \aaa \right)
\tensor
\left(
\mathopen{\new x. \new y.}\left( \co{\bbb} \ooplus \co{\ccc} \right)
\cpar \mathopen{\wen y.}\left( \wen x. \bbb \wwith \wen x.\ccc \right)
\right)
}{
\infer[\mbox{by \textit{select1}}]{
\mathopen{\new y.\forall z. \new x.}
\left(\exists z. \co{\aaa} \cpar  \aaa \right)
\tensor
\left(
\mathopen{\new x. \new y.}\left( \co{\bbb} \ooplus \co{\ccc} \right)
\cpar \mathopen{\wen y.}\left( \wen x. \bbb \wwith \wen x.\ccc \right)
\right)
}{
\infer[\mbox{by \textit{equivariance} and \textit{close}}]{
\mathopen{\new y.\forall z. \new x.}
\left(\co{\aaa} \cpar  \aaa \right)
\tensor
\left(
\mathopen{\new x. \new y.}\left( \co{\bbb} \ooplus \co{\ccc} \right)
\cpar \mathopen{\wen y.}\left( \wen x. \bbb \wwith \wen x.\ccc \right)
\right)
}{
\infer[\mbox{by \textit{external}}]{
\mathopen{\new y.\forall z. \new x.}
\left(\co{\aaa} \cpar \aaa \right)
\tensor
\mathopen{\new y.}
\left(
\mathopen{\new x.}\left( \co{\bbb} \ooplus \co{\ccc} \right)
\cpar \left( \wen x. \bbb \wwith \wen x.\ccc \right)
\right)
}{
\infer[\mbox{by \textit{close}}]{
\mathopen{\new y.\forall z. \new x.}
\left(\co{\aaa} \cpar \aaa \right)
\tensor
\mathopen{\new y.}
\left(
\left(
\mathopen{\new x.}\left( \co{\bbb} \ooplus \co{\ccc} \right) \cpar \wen x. \bbb
\right) \wwith
\left(
\mathopen{\new x.}\left( \co{\bbb} \ooplus \co{\ccc} \right) \cpar \wen x.\ccc
\right)
\right)
}{
\infer[\mbox{by \textit{left} and \textit{right}}]{
\mathopen{\new y.\forall z. \new x.}
\left(\co{\aaa} \cpar  \aaa \right)
\tensor
\mathopen{\new y.}
\left(
\mathopen{\new x.}\left(\left( \co{\bbb} \ooplus \co{\ccc} \right) \cpar \bbb\right)
 \wwith
\mathopen{\new x.}\left(\left( \co{\bbb} \ooplus \co{\ccc} \right) \cpar \ccc\right)
\right)
}{
\infer[\mbox{by \textit{atomic interaction}}]{
\mathopen{\new y.\forall z. \new x.}
\left(\co{\aaa} \cpar \aaa \right)
\tensor
\mathopen{\new y.}
\left(
\mathopen{\new x.}\left( \co{\bbb} \cpar \bbb\right)
 \wwith
\mathopen{\new x.}\left(\co{\ccc}  \cpar \ccc\right)
\right)
}{
\infer[\mbox{by \textit{tidy name} and \textit{tidy1}}]{
\mathopen{\new y.\forall z. \new x.}
\cunit
\tensor
\mathopen{\new y.}
\left(
\mathopen{\new x.}\cunit
 \wwith
\mathopen{\new x.}\cunit
\right)
}{
\cunit
}}}}}}}}}}}
\]
\caption{A proof of $\new x.\new y. \forall z. \aaa \cpar \mathopen{\wen x. \wen y.}\left( \bbb \wwith \ccc \right) \multimap
\new y. \forall z. \new x. \aaa \cpar \mathopen{\wen y.}\left( \wen x. \bbb \wwith \wen x. \ccc \right)$}
\label{fig:equiv-ex}
\end{figure}

By the implications above, if cut elimination holds, it must be the case that the following is provable.
\[
\mathopen{\new x.}\left(\new y. \forall z. \aaa \cpar \mathopen{\wen y.}\left( \bbb \wwith \ccc \right)\right)
\multimap
\new y. \forall z. \new x. \aaa \cpar \mathopen{\wen y.}\left( \wen x. \bbb \wwith \wen x. \ccc \right)
\]
However, without the \textit{all name} and \textit{with name} rules,
the above implication is not provable and hence cut elimination would not hold in the presence of \textit{equivariance}.
Fortunately, using both the \textit{all name} and \textit{with name} rules the above implication is provable, as follows.
\[
\infer[\mbox{\textit{all name} and \textit{equivariance}}]{
\mathopen{\wen x.}\left(\wen y. \exists z. \co{\aaa} \tensor \mathopen{\new y.}\left( \co{\bbb} \ooplus \co{\ccc} \right)\right)
\cpar
\new y. \forall z. \new x. \aaa \cpar \mathopen{\wen y.}\left( \wen x. \bbb \wwith \wen x. \ccc \right)
}{
\infer[\mbox{\textit{with name} and \textit{equivariance}}]{
\mathopen{\wen x.}\left(\wen y. \exists z. \co{\aaa} \tensor \mathopen{\new y.}\left( \co{\bbb} \ooplus \co{\ccc} \right)\right)
\cpar
\new x. \new y. \forall z. \aaa \cpar \mathopen{\wen y.}\left( \wen x. \bbb \wwith \wen x. \ccc \right)
}{
\infer[\mbox{by \textit{close}}]{
\mathopen{\wen x.}\left(\wen y. \exists z. \co{\aaa} \tensor \mathopen{\new y.}\left( \co{\bbb} \ooplus \co{\ccc} \right)\right)
\cpar
\new x. \new y. \forall z. \aaa \cpar \mathopen{\wen x. \wen y.}\left( \bbb \wwith \ccc \right)
}{
\infer[\mbox{by Proposition~\ref{proposition:reflexivity}}]{
\mathopen{\wen x.}\left(\wen y. \exists z. \co{\aaa} \tensor \mathopen{\new y.}\left( \co{\bbb} \ooplus \co{\ccc} \right)\right)
\cpar
\mathopen{\new x.} \left(\new y. \forall z. \aaa \cpar \mathopen{\wen y.}\left( \bbb \wwith \ccc \right)\right)
}{
\cunit
}}}}
\]
A similar argument justifies the necessity of the \textit{left name} and \textit{right name} rules. 

\textbf{Polarities of the nominals.}
As with focussed proof search~\cite{Andreoli1992,Chaudhuri2011}, assigning a positive or negative polarity to operators explains certain distributivity properties.
Consider $\cpar$, $\wwith$, $\forall$ and $\new$ to be negative operators, 
and $\tensor$, $\ooplus$, $\exists$ and $\wen$ to be positive operators, where \textit{seq} is both positive and negative. The negative quantifier $\new$ distributes over all positive operators. 
Considering positive operator \textit{tensor} 
for example, $\vdash \new x .\alpha \tensor \new x. \beta  \multimap \new x. \left( \alpha \tensor \beta \right)$ holds but the converse implication does not hold. Furthermore, $\wen x. \alpha \tensor \wen x. \beta$ and $\wen x .\left(\alpha \tensor \beta\right)$ are unrelated by linear implication in general.
Dually, for the negative operator \textit{par} the only distributivity property that holds for nominal quantifiers is $\vdash \wen x .\left( \alpha \cpar \beta \right) \multimap \wen x .\alpha \cpar \wen x .\beta$.
The \textit{new wen} rule completes this picture of \textit{new} distributing over positive operators and \textit{wen} distributing over negative operators.
From the perspective of embedding name-passing process calculi in logic, the above distributivity properties of \textit{new} and \textit{wen} suggest that processes should be
encoded using negative operators $\new$ and $\cpar$ for private names and parallel composition
(or perhaps dually, using positive operators $\wen$ and $\tensor$),
so as to avoid private names distributing over parallel composition,
which we have shown to be problematic in Section~\ref{section:self-dual}.

The control of distributivity exercised by \textit{new} and \textit{wen} contrasts with the situation for universal and existential quantifiers, where $\exists$ commutes in one direction over all operators and $\forall$ commutes with all operators in the opposite direction, similarly to the additive $\ooplus$ and $\wwith$ which are also  insensitive to the polarity of operators with which they commute.
In the sense of control of distributivity~\cite{Blute2012}, \textit{new} and \textit{wen} behave more like multiplicatives than additives, but are unrelated to multiplicative quantifiers in the logic of bunched implications~\cite{BI}.

%% file: lemmas.tex
\section{The Splitting Technique for Renormalising Proofs}
\label{section:splitting}

This section presents the \textit{splitting} technique that is central to the cut elimination proof for \textsf{MAV1}.
Splitting is used to recover a syntax directed approach for sequent-like contexts.
Recall that in the sequent calculus rules are always applied to the root connective of a formula in a sequent, whereas deep inference rules can be applied deep within any context.
The technique is used to guide proof normalisation
leading to the cut elimination result at the end of Section~\ref{section:context}.

There are complex inter-dependencies between the nominals \textit{new} and \textit{wen} and other operators, particularly the multiplicatives \textit{times} and \textit{seq} and additive \textit{with}.
As such, the splitting proof is tackled as follows, as illustrated in Fig.~\ref{figure:strategy}:
\begin{itemize}
\item Splitting for the first-order universal quantifier $\forall$ can be treated independently of the other operators; hence a direct proof of splitting for this operator is provided first as a simple induction over the length of a derivation in Lemma~\ref{lemma:universal}. Splitting for all other operators are dependent on this lemma.

\item Due to inter-dependencies between $\new$, $\wen$, $\tensor$, $\cseq$ and $\wwith$, splitting for these operators are proven simultaneously by a (huge) mutual induction in Lemma~\ref{lemma:split-times}.
The induction is guided by an intricately designed multiset-based measure of the size of a proof in Definition~\ref{definition:size}. 
The balance of dependencies between operators in this lemma is, by far, the most challenging aspect of this paper.

\item Having established 
Lemma~\ref{lemma:universal} and Lemma~\ref{lemma:split-times},
splitting for the remaining operators $\exists$ and $\ooplus$ and the atoms can each be established independently of each other in Lemmas~\ref{lemma:split-exists},~\ref{lemma:split-plus} and~\ref{lemma:split-atoms} respectively.
\end{itemize}
\begin{figure}[h!]
\[
\xymatrix@C=1.5pc@R=2pc{
&& 
\mbox{\txt{
Splitting $\exists$
\\ (Lemma~\ref{lemma:split-exists})}}
\ar[dr]
\\
\mbox{
\txt{
Splitting $\forall$
\\ (Lemma~\ref{lemma:universal})}}
\ar[r]
&
\mbox{\txt{
Splitting $\new, \wen, \tensor, \cseq, \wwith$
\\
 (Lemma~\ref{lemma:split-times})}}
\ar[r]
\ar[ru]
\ar[rd]
&
\mbox{\txt{
Splitting $\ooplus$\\
 (Lemma~\ref{lemma:split-plus})}}
\ar[r]
&
\mbox{Section~\ref{section:context}}
\\
&&
\mbox{Splitting $\alpha, \co{\alpha}$ (Lemma~\ref{lemma:split-atoms})}
\ar[ur]
}
\]
\caption{
The proof strategy:
dependencies between splitting lemmas leading to cut elimination.}
\label{figure:strategy}
\end{figure}

\subsection{Elimination of universal quantifiers from a proof}

%

We employ a trick 
where universal quantification $\forall$ receives a more direct treatment than other operators.
The proof requires closure of rules under substitution of terms for variables, established as follows directly by induction over the length of a derivation using a function over formulae.
\begin{lemma}[Substitution]
\label{lemma:substitution}
If we have derivation $\vcenter{\infer{Q}{P}}$,
then we have derivation $\vcenter{\infer{Q\sub{x}{v}}{P\sub{x}{v}}}$.
\end{lemma}
\fullproof{
\begin{proof}
\input{sub-proof}
\end{proof}
}
%


We can now establish, the following lemma directly, which is a \textit{co-rule} elimination lemma. 
By a co-rule, we mean that, for \textit{select} rule $\vcenter{\infer[]{\context{\exists x.P}}{\context{P\sub{x}{v}}}}$, 
there is complementary rule $\vcenter{\infer[]{\context{P\sub{x}{v}}}{\context{\forall x.P}}}$ where the direction of inference is reversed and the formulae are complemented. 
Such a co-rule can always be eliminated from a proof, in which case we say \textit{co-select1} is \textit{admissible}, as established by the following lemma.
\begin{lemma}[Universal]\label{lemma:universal}
If $\vdash \context{ \forall x. P }$ holds
 then, for all terms~$v$, $\vdash \context{ P\sub{x}{v} }$ holds.
\end{lemma}

A corollary of Lemma~\ref{lemma:universal} is: if $\vdash \forall x. P \cpar Q$ then $\vdash P\sub{x}{y} \cpar Q$, where $\nfv{y}{\left(\forall x. P \cpar Q\right)}$.
This corollary is in the form of a \textit{splitting} lemma,
where we have a principal connective $\forall$ at the root of a formula inside a context of the form $\left\{\ \cdot\ \right\} \cpar Q$. This corollary of the above lemma should remind the reader of the (invertible) sequent calculus rule for universal quantifiers:
\[
\begin{prooftree}
\vdash P\sub{x}{y}, \Gamma
\justifies
\vdash \forall x. P, \Gamma
\using
\mbox{where $y$ is fresh for $\forall x. P$ and all formulae in $\Gamma$}
\end{prooftree}
\]
We discuss, the significance of splitting lemmas after some preliminary lemmas required for the main splitting result.

\subsection{Killing contexts and technical lemmas required for splitting}

We require a restricted form of context called a killing context (terminology is from~\cite{Chaudhuri2011}). 
A killing context is a context with one or more holes, defined as follows.
\begin{definition}
A \textit{killing context} is a context defined by the following grammar.
\[
 \tcontext{} \Coloneqq \{\ \cdot\ \} \mid \tcontext{} \wwith \tcontext{} \mid \forall x. \tcontext{}
\mid \new x.\tcontext{}
\]
In the above, $\{\ \cdot\ \}$ is a hole into which any formula can be plugged. An $n$-ary killing context is a killing context in which $n$ holes appear.
\end{definition}

For readability of large formulae involving an $n$-ary killing context, for $n > 1$, we represent the holes using a comma-separated list, so for example, 
instead of writing $\tcontextsym \{\cdot\} \{ \cdot\}$, we write $\tcontext{\cdot, \cdot}$ for a binary context. Given an $n$-ary killing context $\tcontext{\ldots}$, we write $\tcontext{Q_1,\ldots,Q_n}$ to denote the formula obtained by filling the holes in the context with formulas $Q_1,\ldots, Q_n.$ 
We also introduce the notation $\tcontext{Q_i \colon 1 \leq i \leq n}$ as shorthand for $\tcontext{Q_1, Q_2, \hdots, Q_n}$; and $\tcontext{ Q_i \colon i \in I }$ for a family of formulae indexed by finite subset of natural numbers $I$. 

A killing context represents a context that cannot in general be removed until all other rules in a proof have been applied, hence the corresponding \textit{tidy} rules are suspended until the end of a proof.  A killing context has properties that are applied frequently in proofs, characterised by the following lemma. 
\begin{lemma}\label{lemma:kill-distribute}
For any killing context $\tcontext{}$, $\vdash \tcontext{\cunit, \hdots, \cunit}$ holds; and, assuming the free variables of $P$ are not bound by $\tcontext{}$,
we have derivation
\[
\infer[.]{
P \cpar \tcontext{Q_1, Q_2, \hdots Q_n}
}{
\tcontext{P \cpar Q_1, P \cpar Q_2, \hdots P \cpar Q_n}
}
\]
\end{lemma}
\fullproof{
\input{distribute-proof}
}

Killing contexts also satisfy the following property that is necessary for handling the \textit{seq} operator, which interacts subtly with killing contexts.
\begin{lemma}~\label{lemma:medial}
Assume that $I$ is a finite subset of natural numbers, $P_i$ and $Q_i$ are formulae, for $i \in I$, and $\tcontext{}$ is a killing context. There exist killing contexts $\tcontextn{0}{}$ and $\tcontextn{1}{}$ and sets of natural numbers $J \subseteq I$ and $K \subseteq I$ such that the following derivation holds:
\[
\infer[.]{
\tcontext{ P_i \cseq Q_i \colon i \in I }
}{
\tcontextn{0}{ P_j \colon j \in J } \cseq
\tcontextn{1}{ Q_k \colon k \in K }
}
\]
\end{lemma}

The following lemma checks that \textit{wen} quantifiers can propagate to the front of a killing context.
Similarly, to the proof of the lemma above, the proof is by induction on the structure of a killing context, applying the \textit{all name}, $\textit{new wen}$, $\textit{with name}$, $\textit{left name}$ or $\textit{right name}$ rule, as appropriate.
\begin{lemma}\label{lemma:commute}
Consider an $n$-ary killing context $\tcontext{}$ and formulae such that $\nfv{x}{P_i}$ and either $P_i = \wen x. Q_i$ or $P_i = Q_i$,
for $1 \leq i \leq n$. If for some $i$ such that $1 \leq i \leq n$, $P_i = \wen x. Q_i$, then we have derivation 
$\vcenter{\infer{\tcontext{ P_1, P_2, \hdots P_n}}{\wen x. \tcontext{ Q_1, Q_2, \hdots, Q_n }}}$.
\end{lemma}

To handle certain cases in splitting the following definitions and property is helpful.
Assume $\vec{y}$ defines a possibly empty list of variables $y_1, y_2, \hdots, y_n$ and $\quantifier \vec{y}. P$ abbreviates $\quantifier y_1. \quantifier y_2. \hdots \quantifier y_n. P$. Let $\nfv{\vec{y}}{P}$ hold only if $\nfv{y}{P}$ for every $y \in \vec{y}$.
By induction over the length of $\vec{z}$ we can establish the following lemma, by repeatedly applying the \textit{close}, \textit{fresh} and \textit{extrude new} rules.
\begin{lemma}\label{lemma:close}
If $\vec{y} \subseteq \vec{z}$ and $\nfv{\vec{z}}{\wen \vec{y}. P}$, then we have derivations 
$\vcenter{\infer{\wen \vec{y}. P \cpar \new \vec{z}. Q}{\mathopen{\new \vec{z}.} \left( P \cpar Q \right)}}$ 
and $\vcenter{\infer{\new \vec{y}. P \cpar \wen \vec{z}. Q}{\mathopen{\new \vec{z}.} \left( P \cpar Q \right)}}$.
\end{lemma}

%% file: sub-proof.tex
Consider the case of the \textit{switch} rule. The following rule can be derived.
\[
\left(\left(P \tensor Q\right) \cpar V\right) \sub{x}{v}
\equiv
\left(P\sub{x}{v} \tensor Q\sub{x}{v}\right) \cpar V\sub{x}{v}
\longrightarrow
P\sub{x}{v} \tensor \left(Q\sub{x}{v} \cpar V\sub{x}{v}\right)
\equiv
\left(P \tensor \left(Q \cpar V\right)\right)\sub{x}{v}
\]
The cases are similar for the rules \textit{sequence}, \textit{left}, \textit{right}, \textit{external}, \textit{atomic interaction}, \textit{tidy}, \textit{tidy1} and \textit{tidy name}.

Consider the case of the \textit{extrude new} rule. The following rule can be derived, where $\nfv{y}{Q}$, $z \not = x$, $\nfv{z}{v}$ and $\nfv{x}{(\new y. P \cpar Q)}$.
\[
\begin{array}{rl}
\left(\new y. P \cpar Q\right)\sub{x}{v}
\equiv&
\mathopen{\new z.} \left(P\sub{y}{z}\sub{x}{v}\right) \cpar Q
\\
\longrightarrow&
\mathopen{\new z.} \left(P\sub{y}{z}\sub{x}{v} \cpar Q \right)
\equiv
\mathopen{\new z.} \left(\left(P \cpar Q \right)\sub{y}{z}\sub{x}{v}\right)
\equiv
\mathopen{\new y.} \left(\left(P \cpar Q \right)\right)\sub{x}{v}
\end{array}
\]
The cases for \textit{extrude1}, \textit{left wen}, and \textit{right wen} rules are similar.

In the case of the \textit{select} rule. The following rule can be derived, where $z \not = x$, $\nfv{z}{v}$ and $\nfv{x}{(\exists y. P)}$.
\[
\begin{array}{rl}
\left( \exists y. P \right) \sub{x}{v}
\equiv
\left( \exists z. P\sub{y}{z}\sub{x}{v} \right) 
\longrightarrow
P\sub{y}{z}\sub{x}{v}\sub{z}{t\sub{x}{v}}
\equiv
P\sub{y}{z}\sub{z}{t}\sub{x}{v}
\equiv
P\sub{y}{t}\sub{x}{v}
\end{array}
\]

Consider the case of the \textit{close} rule. The following rule can be derived, where $z \not = x$, $\nfv{z}{v}$ and $\nfv{x}{(\wen{y}. P \cpar \new{y}. Q)}$.
\[
\begin{array}{rl}
\left( \wen{y}. P \cpar \new{y}. Q \right) \sub{x}{v}
\equiv&
\mathopen{\wen{z}}.\left(P\sub{y}{z}\sub{x}{v}\right) \cpar \mathopen{\new{z}.} \left( Q\sub{y}{z}\sub{x}{v} \right)
\\
\longrightarrow&
\mathopen{\new{z}.}\left(P\sub{y}{z}\sub{x}{v} \cpar Q\sub{y}{z}\sub{x}{v} \right)
\equiv
\mathopen{\new{z}.}\left(\left(P \cpar Q \right) \sub{y}{z}\sub{x}{v}\right)
\equiv
\left(\mathopen{\new{y}.}\left(P \cpar Q \right) \right)\sub{x}{v}
\end{array}
\]
The cases for \textit{medial1}, \textit{medial new}, \textit{suspend} are similar.

Consider the case of the \textit{new wen} rule. Pick fresh names $y'$ and $z'$ such that $x \not= y'$, $x \not=z'$, $y' \not= z'$, $\nfv{y'}{v}$, $\nfv{z'}{v}$, $\nfv{y'}{P}$, and $\nfv{z'}{P}$, hence $\sub{y}{y'}\sub{z}{z'} = \sub{z}{z'}\sub{y}{y'}$; and thereby the following rule can be construced.
\[
\left( \new y. \wen z. P \right)\sub{x}{v} 
\equiv
\mathopen{\new y'. \wen z'.} \left( P \sub{z}{z'}\sub{y}{y'}\sub{x}{v}  \right)
\longrightarrow
\mathopen{\wen z'. \new y'.} \left( P \sub{y}{y'}\sub{z}{z'}\sub{x}{v}  \right)
\equiv
\left(\wen z. \new y. P \right)\sub{x}{v}
\] 
The lemma follows by induction on the length of a derivation.

%% file: distribute-proof.tex
\begin{proof}
The proofs follow by straightforward inductions over the structure of a killing context.

When the killing context is one hole only $P \cpar \ehole{Q} = \ehole{P \cpar Q}$ and $\ehole{\cunit} = \cunit$, as required.

Now assume that by the induction hypothesis the following hold for killing contexts $\tcontextn{1}{}$ and $\tcontextn{2}{}$, and also $\vdash \tcontextn{1}{\cunit, \hdots, \cunit}$ and $\vdash \tcontextn{2}{\cunit, \hdots, \cunit}$.
\[
\begin{array}{rl}
 P \cpar \tcontextn{1}{Q_1, \hdots, Q_m} &\longrightarrow \tcontextn{1}{P \cpar Q_1, \hdots, P \cpar Q_m}
 \\
 P \cpar \tcontextn{2}{Q_{m+1}, \hdots, Q_{m+n}} &\longrightarrow \tcontextn{2}{P \cpar Q_{m+1}, \hdots, P \cpar Q_{m+n}}
\end{array}
\]
Hence, by distributivity the following derivation can be constructed.
\[
\begin{array}{l}
P \cpar \left( \tcontextn{1}{Q_1, \hdots, Q_m} \wwith \tcontextn{2}{Q_{m+1}, \hdots, Q_{m+n}} \right)
\\\qquad
 \begin{array}{l}
 \longrightarrow P \cpar \tcontextn{1}{Q_1, \hdots, Q_m} \wwith P \cpar \tcontextn{2}{Q_{m+1}, \hdots, Q_{m+n}}
 \\
 \longrightarrow \tcontextn{1}{P \cpar Q_1, \hdots, P \cpar Q_m} \wwith \tcontextn{2}{P \cpar Q_{m+1}, \hdots, P \cpar Q_{m+n}}
 \end{array}
\end{array}
\]
Furthermore, $\tcontextn{1}{\cunit, \hdots, \cunit} \wwith \tcontextn{2}{\cunit, \hdots, \cunit} \longrightarrow \cunit \wwith \cunit \longrightarrow \cunit$ holds.

As the induction hypothesis, assume $P \cpar \tcontext{ Q_1, \hdots, Q_n } \longrightarrow \tcontext{ P \cpar Q_1, \hdots, P \cpar Q_n }$ and also that $\vdash \tcontext{ \cunit, \hdots, \cunit }$ . Hence, given $\nfv{x}{P}$, the following derivation holds.
\[
P \cpar \new x \tcontext{ Q_1, \hdots, Q_n } \longrightarrow \new x \left( P \cpar \tcontext{ Q_1, \hdots, Q_n }\right) \longrightarrow \new x \tcontext{ P \cpar Q_1, \hdots, P \cpar Q_n }
\]
Furthermore, $\new x \tcontext{ \cunit, \hdots, \cunit} \longrightarrow \new x \cunit \longrightarrow \cunit$.

A similar argument holds for a killing context of the form $\forall x \tcontext{}$.
\end{proof}

%% file: multisets.tex
\subsection{An Affine Measure for the Size of a Proof.}

As an induction measure in the splitting lemmas, we employ a multiset-based measure~\cite{Dershowitz1979} of the size of a proof.
An \textit{occurrence count} is defined in terms of a multiset of multisets. To give weight to nominals, a \textit{wen} and \textit{new} count is employed.
The measure of the size of a proof, Definition~\ref{definition:size}, is then given by the lexicographical order induced by the occurrence count, wen count and new count for the formula in the conclusion of a proof, and the derivation length of the proof itself.

In the sub-system \textsf{BV}~\cite{Guglielmi2007}, the occurrence count is simply the number of atom and co-atom occurrences.
For the sub-system corresponding to \textsf{MALL} (multiplicative-additive linear logic)~\cite{Strassburger2003thesis}, i.e. without \textit{seq}, 
a multiset of atom occurrences such that $\occ{\left(P \wwith Q\right) \cpar R} = \occ{\left(P  \cpar R\right) \wwith \left(Q \cpar R\right)}$ is sufficient, to ensure that the \textit{external} rule does not increase the size of the measure. 
The reason why a multiset of multisets is employed for extensions of \textsf{MAV}~\cite{Horne2015} is to handle subtle interactions between the unit, \textit{seq} and \textit{with} operators.
In particular, by applying the structural rules for units, such that $ \context{ P \wwith Q }
 \equiv
 \context{ \left(P \cseq \cunit\right) \wwith \left(\cunit \cseq Q\right) }$
and the \textit{medial} rule, we obtain the following inference.
\[
\infer[\mbox{by the \textit{medial} rule}]{
 \context{ P \wwith Q }
}{
 \context{ \left(P \wwith \cunit\right) \cseq \left(\cunit \wwith Q\right) }
}
\]
In the above derivation, the units cannot in general be removed from the formula in the premise; hence extra care should be taken that these units do not increase the size of the formula. This observation leads us to the notion of multisets of multisets of natural numbers defined below.

\begin{definition}
We denote the standard multiset disjoint union operator as $\discup$, a multiset sum operator defined such that $M + N = \left\{ m + n \colon m \in M \mbox{ and } n \in N \right\}$.
We also define pointwise plus and pointwise union over multisets of multisets of natural numbers, where $\mathcal{M}$ and $\mathcal{N}$ are multisets of multisets.
$\mathcal{M} \mplus \mathcal{N} = \left\{ M + N, M \in \mathcal{M} \mbox{ and } N \in \mathcal{N} \right\}$
and
$\mathcal{M} \mcup \mathcal{N} = \left\{ M \discup N, M \in \mathcal{M} \mbox{ and } N \in \mathcal{N} \right\}$.
\end{definition}

We employ two distinct multiset orderings over multisets and over multisets of multisets.
\begin{definition}
For multisets of natural numbers $M$ and $N$, define a multiset ordering $M \leq N$ if and only if there exists an injective multiset function $f \colon M \rightarrow N$ such that, for all $m \in M$, $m \leq f(m)$. Strict multiset ordering $M < N$ is defined such that $M \leq N$ but $M \not= N$.
\end{definition}

\begin{definition}
Given two multisets of multisets of natural numbers $\mathcal{M}$ and $\mathcal{N}$, $\mathcal{M} \sqsubseteq \mathcal{N}$ holds if and only if $\mathcal{M}$ can be obtained from $\mathcal{N}$ by repeatedly removing a multiset $N$ from $\mathcal{N}$ and replacing $N$ with zero or more multisets $M_i$ such that $M_i < N$. $\mathcal{M} \sqsubset \mathcal{N}$ is defined when $\mathcal{M} \sqsubseteq \mathcal{N}$ but $\mathcal{M} \not= \mathcal{N}$.
\end{definition}

\begin{definition}
The occurrence count is the following function from formulae to multiset of multisets of natural numbers.
\begin{gather*}
\begin{array}{l}
\occ{\cunit} = \left\{\left\{ 0 \right\}\right\}
\qquad
\occ{\alpha} = \occ{\co{\alpha}} = \left\{\left\{ 1 \right\}\right\}
\\[6pt]
\occ{P \wwith Q} = \occ{P \ooplus Q} = \occ{P} \sqcup \occ{Q}
\\[6pt]
\occ{P \cpar Q} = \occ{P} \boxplus \occ{Q}
\\[6pt]
\occ{ \forall x. P } = \occ{ \exists x. P } = \left\{\left\{ 0 \right\}\right\} \mcup \occ{ P }
\end{array}
\!\!\!\!\!\!\!\!\!\!
\begin{array}{r}
\occ{ \new x. P } = \occ{ \wen x. P } = 
\left\{
\begin{array}{lr}
\left\{\left\{ 0, 0 \right\}\right\} & \mbox{if $P \equiv \cunit$} \\
\occ{P} & \mbox{otherwise}
\end{array}
\right.
\\[6pt]
\occ{P \tensor Q} = \occ{P \cseq Q} = 
\left\{
\begin{array}{lr}
\occ{P} & \mbox{if $Q \equiv \cunit$} \\
\occ{Q} & \mbox{if $P \equiv \cunit$} \\
\occ{P} \discup \occ{Q} & \mbox{otherwise}
\end{array}
\right.
\end{array}
\end{gather*}
\end{definition}

\begin{definition}
The wen count is the following function from formulae to natural numbers.
\begin{gather*}
\wensize{ \wen x. P } = 1 + \wensize{ P }
\qquad
\wensize{ \exists x. P } = \wensize{ \forall x. P } = \wensize{ \new x. P } = \wensize{ P }
\qquad
\wensize\alpha = \wensize{\co\alpha} = \wensize\cunit = 1
\\[4pt]
\wensize{ P \cseq Q } = \wensize{ P \tensor Q } = \wensize{ P \cpar Q } = \wensize{ P }\wensize{ Q }
\qquad
\wensize{ P \ooplus Q } = \wensize{ P \wwith Q } = \wensize P + \wensize Q
\end{gather*}
\end{definition}

\begin{definition}
The new count is the following function from formulae to natural numbers.
\begin{gather*}
\newsize{ \new x. P } = 1 + \newsize{ P }
\qquad
\newsize{ \exists x. P } = \newsize{ \forall x. P } = \newsize{ \wen x. P } = \newsize{ P }
\qquad
\newsize\alpha = \newsize{\co\alpha} = \newsize\cunit = 1
\\[4pt]
\newsize{ P \cpar Q } = \newsize{ P }\newsize{ Q }
\quad
\newsize{ P \ooplus Q } = \newsize{ P \wwith Q } = \newsize P + \newsize Q
\quad
\newsize{ P \cseq Q } = \newsize{ P \tensor Q } = \mmax{\newsize{P}, \newsize{Q}}
\end{gather*}
\end{definition}

\begin{definition}
The size of a formula $\size{P}$ is defined as the triple $\left(\occ{P}, \wensize{P}, \newsize{P} \right)$ lexicographically ordered by $\prec$.
$\phi \preceq \psi$ is defined such that $\phi \prec \psi$ or $\phi = \psi$ pointwise.
\end{definition}

\begin{definition}\label{definition:size}
The size of a proof of $P$ with derivation of length $n$ is given by the tuple of the form $\left(\size{P}, n\right)$, subject to lexicographical ordering.
\end{definition}

%

%% file: bound.tex
\begin{lemma}\label{lemma:size-sub}
For any formula $P$ and term $t$, $\size{ P } = \size{ P\sub{x}{t} }$.
\end{lemma}

\begin{lemma}\label{lemma:multiset-congruence}
If $P \equiv Q$ then $\size{P} = \size{Q}$.
\end{lemma}
\fullproof{
\begin{proof}
For commutativity the following arguments hold for \textit{par} and \textit{times}.
\begin{gather*}
\begin{array}{rl}
\occ{P \cpar Q} = \occ{P} \mplus \occ{Q} =&\!\!\left\{ M + N \colon M \in \occ{P}, N \in \occ{Q} \right\}
\\
=&\!\!\left\{ N + M \colon M \in \occ{P}, N \in \occ{Q} \right\} = \occ{Q} \mplus \occ{P} = \occ{ Q \cpar P }
\end{array}
\\
\occ{P \tensor Q} = \occ{P} \discup \occ{Q} = \occ{Q} \discup \occ{P} = \occ{Q \tensor P}
\\
\wensize{ P \tensor Q } = \wensize{ P \cpar Q } = \wensize{ P }\wensize{ Q } = \wensize{ Q }\wensize{ P } = \wensize{ Q \cpar P } = \wensize{ Q \tensor P }
\\
\newsize{ P \cpar Q } = \newsize{ P }\newsize{ Q } = \newsize{ Q }\newsize{ P } = \newsize{ Q \cpar P }
\\
\newsize{ P \tensor Q } = \mmax{\newsize{ P } , \newsize{ Q }} = \mmax{\newsize{ Q }, \newsize{ P }} = \newsize{ Q \tensor P }
\qquad
\end{gather*}

The identity rules hold by the following reasoning, where $\odot \in \left\{ \tensor, \cseq \right\}$, since $\newsize{P} \geq 1$ always $\newsize{P \odot \cunit} = \mmax{\newsize{ P}, 1} = \newsize{P}$ and the following hold. 
\begin{gather*}
\occ{P \cpar \cunit} = \occ{P} \mplus \left\{\left\{0\right\}\right\} = \occ{P}
\quad
\occ{\cunit \odot P} = \occ{P}
\quad
\wensize{P \odot \cunit} 
= \wensize{P}
= \wensize{P \cpar \cunit}
\quad
\newsize{P \cpar \cunit} = \newsize{ P }
\end{gather*}

Associativity properties hold by extending associativity of multisets to multisets of multisets.
\[
\begin{array}{rl}
\size{\left(P \cpar Q\right) \cpar R}
 =
 \left(\size{P} \mplus \size{Q}\right) \mplus \size{R}
 =&
 \left\{ \left(M + N\right) + K \colon M \in \size{P}, N \in \size{Q}, K \in \size{R} \right\}
 \\
 =&
 \left\{ M + \left(N + K\right) \colon M \in \size{P}, N \in \size{Q}, K \in \size{R} \right\}
 =  \size{P} \mplus \left(\size{Q} \mplus \size{R}\right) 
 = \size{P \cpar \left(Q \cpar R\right)}
\end{array}
\]
If any one of $P \equiv \cunit$, $Q \equiv \cunit$ or $R \equiv \cunit$ hold, then $\size{\left(P \cseq Q\right) \cseq R} = \size{P \cseq \left(Q \cseq R\right)}$ by definition. If $P \not\equiv \cunit$ and $Q \not\equiv \cunit$ and $R \not\equiv \cunit$, then the following equalities hold:
$
\size{\left(P \cseq Q\right) \cseq R} = \left(\size{P} \discup \size{Q}\right) \discup \size{R} = \size{P} \discup \left(\size{Q} \discup \size{R}\right) = \size{P \cseq \left(Q \cseq R\right)}
$.
The same associativity argument works for the \textit{times} operator. Furthermore, the following wen count and new count preserve associativity as follows, where $\odot \in \left\{ \tensor, \cseq \right\}$.
\[
\wensize{ \left( P \cpar Q\right) \cpar R } =
\wensize{ \left( P \odot Q\right) \odot R } = \left(\wensize{ P }\wensize{ Q }\right)\wensize{ R } = \wensize{ P }\left(\wensize{ Q }\wensize{R}\right) = \wensize{ P \odot \left( Q \odot R\right) }
= \wensize{ P \cpar \left(Q \cpar R\right) }
\]
\[
\newsize{ \left( P \cpar Q\right) \cpar R } = \left(\newsize{ P }\newsize{ Q }\right)\newsize{ R } = \newsize{ P }\left(\newsize{ Q }\newsize{R}\right) = \newsize{ P \cpar \left( Q \cpar R\right) }
\]
\[
\newsize{ \left( P \odot Q\right) \odot R } = \mmax{\mmax{\newsize{ P },\newsize{ Q }},\newsize{ R }} = \mmax{\newsize{ P },\mmax{\newsize{ Q },\newsize{R}}} = \newsize{ P \odot \left( Q \odot R\right) }
\]
The lemma then follows by induction over the structure of the derivation of $P \equiv Q$.
\end{proof}
}

The following lemma we will appeal to regularly in the splitting proofs in subsequent sections to bound the size of a derivation.
\begin{lemma}[affine]\label{lemma:bound}
Any derivation $\vcenter{\infer{Q}{P}}$, is bound such that $\size{P} \preceq \size{Q}$.
\end{lemma}

%% file: splitting.tex
\subsection{The splitting technique for simulating sequent-like rules}

The technique called splitting~\cite{Guglielmi2007,Guglielmi2011} generalises the application of rules in the sequent calculus.
In the sequent calculus, any root connective in a sequent can be selected and some rule for that connective can be applied.
For example, consider the following rules in linear logic forming part of a proof in the sequent calculus, where $\nfv{x}{P,Q,U,V,W}$.
\[
\begin{prooftree}
\begin{prooftree}
   \vdash P, U
   \quad
   \vdash Q, R
\justifies
   \vdash P \tensor Q, R, U
\end{prooftree}
  \quad
\begin{prooftree}
\begin{prooftree}
    \vdash P, R, V
    \quad
    \vdash Q, W
\justifies
    \vdash P \tensor Q, R, V, W
\end{prooftree}
\justifies
   \vdash P \tensor Q, R, V \cpar W
\end{prooftree}
\justifies
\begin{prooftree}
\vdash P \tensor Q, R, U\wwith \left(V \cpar W\right)
\justifies
 \vdash P \tensor Q, \forall x.R, U\wwith \left(V \cpar W\right)
\end{prooftree}
\end{prooftree}
\]
In the setting of the calculus of structures, the sequent at the conclusion of the above proof corresponds to a \textit{shallow context} of the form $\left\{\ \cdot\ \right\} \cpar \forall x.R \cpar \left(U\wwith \left(V \cpar W\right)\right)$
where the \textit{times} operator at the root of $P \tensor Q$ is a \textit{principal formula} that is plugged into the shallow context.
Splitting proves that there is always a derivation reorganising a shallow context 
into a form such that a rule for the root connective of the principal formula may be applied.
In the above example, this would correspond to the following derivation over contexts:
\[
\infer[\mbox{by the \textit{extrude1} rule}]{
\left\{\ \cdot\ \right\} \cpar \forall x.R \cpar \left(U\wwith \left(V \cpar W\right)\right)
}{
\infer[\mbox{by the \textit{external} rule}]{
\left\{\ \cdot\ \right\} \cpar \mathopen{\forall x.}\left(R \cpar \left(U\wwith \left(V \cpar W\right)\right)\right)
}{
\left\{\ \cdot\ \right\} \cpar \mathopen{\forall x.}\left(\left(R \cpar U\right) \wwith \left(R \cpar V \cpar W\right)\right)
}}
\]
By plugging in the principal formula, $P \tensor Q$, into the hole in the premise of the above derivation and applying distributivity properties of a killing context (Lemma~\ref{lemma:kill-distribute}),
the \textit{switch} rule involving the principal connective can be applied as follows.
\[
\infer[\mbox{by Lemma~\ref{lemma:kill-distribute}}]{
\left(P \tensor Q\right) \cpar \mathopen{\forall x.}\left(\left(R \cpar U\right) \wwith \left(R \cpar V \cpar W\right)\right)
}{
\infer[\mbox{by the \textit{switch} rule}]{
\mathopen{\forall x.}\left(\left(\left(P \tensor Q\right) \cpar R \cpar U\right) \wwith \left(\left(P \tensor Q\right) \cpar R \cpar V \cpar W\right)\right)
}{
\mathopen{\forall x.}\left(\left(\left(P \cpar U\right)\tensor \left(Q \cpar R\right)\right) \wwith \left(\left(P \cpar R \cpar V\right)\tensor \left(Q \cpar W\right)\right)\right)
}}
\]
Notice that the final formula above holds when all of the following hold: $\vdash P \cpar U$, $\vdash Q \cpar R$, $\vdash P \cpar R \cpar V$ and $\vdash Q \cpar W$.
Notice that these correspond to the leaves of the example sequent above.

Splitting is sufficiently general that the technique can be applied to operators such as \textit{seq} that have no sequent calculus presentation~\cite{Tiu2006}.
The technique also extends to the pair of nominals \textit{new} and \textit{wen}, for which a sequent calculus presentation is an open problem.

The operators \textit{times}, \textit{seq}, \textit{new} and \textit{wen} are treated together in Lemma~\ref{lemma:split-times}.
These operators give rise to \textit{commutative cases}, where rules for these operators can permute with any principal formula, swapping the order of rules in a proof.
\textit{Principal cases} are where the root connective of the principal formula is directly involved in the bottommost rule of a proof.
As with \textsf{MAV}~\cite{Horne2015}, the \textit{principal cases} for \textit{seq} are challenging, demanding Lemma~\ref{lemma:medial}.
The principal case induced by \textit{medial new} demands Lemma~\ref{lemma:commute}. The cases where two nominal quantifiers commute are also interesting, particularly where the case arrises due to \textit{equivariance}.
\begin{lemma}[Core Splitting]
\label{lemma:split-times}
The following statements hold.
\begin{enumerate}
\item If $\vdash \left( P \tensor  Q \right) \cpar R$, then there exist formulae $V_i$ and $W_i$ such that $\vdash P \cpar V_i$ and $\vdash Q \cpar W_i$, where $1 \leq i \leq n$, and $n$-ary killing context $\tcontext{}$ such that $\vcenter{\infer{R}{\tcontext{V_1 \cpar W_1, V_2 \cpar W_2, \hdots, V_n \cpar W_n}}}$
and if $\tcontext{}$ binds $x$ then $\nfv{x}{\left(P \tensor Q\right)}$.

\item If $\vdash \left( P \cseq Q \right) \cpar R$, then there exist formulae $V_i$ and $W_i$ such that $\vdash P \cpar V_i$ and $\vdash Q \cpar W_i$, where $1 \leq i \leq n$, and $n$-ary killing context $\tcontext{}$ such that $\vcenter{\infer{R}{\tcontext{V_1 \andthen W_1, V_2 \andthen W_2, \hdots, V_n \andthen W_n}}}$
and if $\tcontext{}$ binds $x$ then $\nfv{x}{\left(P \cseq Q\right)}$.

\item If $\vdash \new x. P \cpar Q$, then there exist formulae $V$ and $W$
where $\nfv{x}{V}$ and $\vdash P \cpar W$ and either $V = W$ or $V = \wen x .W$,
such that there is a derivation $\vcenter{\infer{Q}{V}}$.

\item If $\vdash \wen x. P \cpar Q$, then there exist formulae $V$ and $W$ where $\nfv{x}{V}$ and $\vdash P \cpar W$ and either $V = W$ or $V = \new x. W$, such that there is a derivation $\vcenter{\infer{Q}{V}}$.

\item If $\vdash \left( P \wwith Q \right) \cpar R$, then $\vdash P \cpar R$ and $\vdash Q \cpar R$.
\end{enumerate}
Furthermore, for all $1 \leq i \leq n$, in the first two cases the size of the proofs
of ${P \cpar V_i}$ and ${Q \cpar W_i}$ are strictly bounded above by the size of the proofs of ${ \left( P \tensor Q \right) \cpar R}$ and ${ \left( P \cseq Q \right) \cpar R }$.
In the third and fourth cases, the size of the proof ${P \cpar W}$ is strictly bounded above by the size of the proofs of $\new x .P \cpar Q$ and $\wen x. P \cpar Q$.
The size of a proof is measured according to Definition~\ref{definition:size}.
\end{lemma}
\tohide{
\input{split-times}
}{}
The final three splitting lemmas mainly involve checking commutative cases.
The commutative cases follow a similar pattern to the commutative cases in Lemma~\ref{lemma:split-times}.
\begin{lemma}\label{lemma:split-exists}
If $\vdash \exists x. P \cpar Q$, then there exist formulae $V_i$ and values $v_i$ such that $\vdash P\sub{x}{v_i} \cpar V_i$, where $1 \leq i \leq n$, and $n$-ary killing context $\tcontext{}$ such that $\vcenter{\infer{Q}{ \tcontext{V_1, V_2, \hdots, V_n}}}$
and if $\tcontext{}$ binds $y$ then $\nfv{y}{\left(\exists x . P\right)}$.
\end{lemma}

%
%
%
%
The proofs of the splitting lemmas for \textit{plus} and atoms offer no new insight or difficulties compared to their treatment in \textsf{MAV}~\cite{Horne2015}.
Similarly, to the above lemma for existential quantifiers, the proofs mainly involve commutative cases of a standard form.
\begin{lemma}\label{lemma:split-plus}
If $\vdash \left( P \ooplus Q \right) \cpar R$, then there exist formulae $W_i$ such that either $\vdash P \cpar W_i$ or $\vdash Q \cpar W_i$ where $1 \leq i \leq n$, and $n$-ary killing context $\tcontext{}$ such that $\vcenter{\infer{R}{\tcontext{W_1, W_2, \hdots, W_n}}}$
and if $\tcontext{}$ binds $x$ then $\nfv{x}{\left(P \ooplus Q\right)}$.
\end{lemma}

\begin{lemma}\label{lemma:split-atoms}
The following statements hold, for any atom $\alpha$, where if $\tcontext{}$ binds $x$ then $\nfv{x}{\alpha}$.
\begin{itemize}
\item If $\vdash \co{\alpha} \cpar Q$, then there exist $n$-ary killing context $\tcontext{}$ such that 
$\vcenter{\infer{Q}{\tcontext{\alpha, \alpha, \hdots, \alpha}}}$.

\item If $\vdash \alpha \cpar Q$, then there exist $n$-ary killing context $\tcontext{}$ such that 
$\vcenter{\infer{Q}{\tcontext{\co{\alpha}, \co{\alpha}, \hdots, \co{\alpha}}}}$.
\end{itemize}
\end{lemma}

%% file: split-times.tex
\begin{proof}
The proof proceeds by induction on the size of the proof, as in Defn.~\ref{definition:size}. 
In each of the following base cases, the conditions for splitting are immediately satisfied.
For the base case for the \textit{tidy name} rule, the bottommost rule of a proof is of the form 
$\vcenter{\infer[]{\new x.\new \vec{y}. \cunit \cpar P}{\new \vec{y}.\cunit \cpar P}}$, where $\nfv{\vec{y}}{P}$.
For the base case for the \textit{tidy} rule, the bottommost rule is of the form 
$\vcenter{\infer[]{\left(\cunit \wwith \cunit\right) \cpar P}{\cunit \cpar P}}$, such that $\vdash \cunit \cpar P$.
For the base case for $\textit{times}$ and $\textit{seq}$, $\vdash \left(\cunit \tensor \cunit\right) \cpar \cunit$ and $\vdash \left(\cunit \cseq \cunit\right) \cpar \cunit$ hold.

\begin{enumerate}[label=\textbf{\Alph*},ref=\Alph*,leftmargin=*]
\item \textbf{Principal cases for wen.}
There are principal cases for \textit{wen} where the rules \textit{close}, \textit{suspend}, \textit{left wen}, \textit{right wen} and \textit{fresh} interfere directly with \textit{wen} at the root of a principal formula. Three representative cases are presented.

\begin{enumerate}[label*=\textbf{.\arabic*}]
\item The first principal case for $\textit{wen}$ is when the bottommost rule of a proof is an instance of the \textit{close} rule of the form 
$
\vcenter{\infer[]{\wen x. P \cpar \new x. Q \cpar R
}{
\mathopen{\new x.} \left( P \cpar Q \right) \cpar R}}
$,
where $\vdash \mathopen{\new x.} \left( P \cpar Q \right) \cpar R$ and $\nfv{x}{R}$.
By the induction hypothesis,
there exist $S$ and $T$ such that $\vdash P \cpar Q \cpar T$ and $\nfv{x}{S}$ and either $S = T$ or $S = \wen x. T$, and also we have derivation $\infer{R}{S}$.
Since $\nfv{x}{S}$, if $S = T$ then $\vcenter{\infer[]{\new x. Q \cpar S}{\mathopen{\new x.}\left( Q \cpar T \right)}}$.
Furthermore, the size of the proof of $P \cpar Q \cpar T$ is no larger than the size of the proof of $\mathopen{\new x.} \left( P \cpar Q \right) \cpar R$; hence strictly bounded by the size of the proof of $\wen x. P \cpar \new x. Q \cpar R$.
If $S = \wen x. T$ then by the \textit{close} rule $\vcenter{\infer[]{\new x. Q \cpar \wen x. T}{\mathopen{\new x.}\left( Q \cpar T \right)}}$.
If $S = T$ then, since $\nfv{x}{S}$, by the \textit{extrude new} rule,
$\vcenter{\infer[]{\new x. Q \cpar T}{\mathopen{\new x.}\left( Q \cpar T \right)}}$.
Hence in either case $\vcenter{\infer[]{\new x. Q \cpar S}{\mathopen{\new x.}\left( Q \cpar T \right)}}$
and thereby the derivation 
$
\vcenter{
\infer[]{
\new x. Q \cpar R
}{
\infer[]{
\new x. Q \cpar S
}{
\mathopen{\new x.} \left(Q \cpar T \right)
}}}
$
can be constructed,
meeting the conditions for splitting for \textit{wen}.

\item Consider the second principal case for \textit{wen} where the bottommost rule of a proof is an instance of the \textit{suspend} rule of the form
$
\vcenter{
\infer[]{
\wen x. P \cpar \wen x. Q \cpar R
}{
\mathopen{\wen x.} \left( P \cpar Q \right) \cpar R
}}
$, where $\vdash \mathopen{\wen x.} \left( P \cpar Q \right) \cpar R$ and $\nfv{x}{R}$.
By the induction hypothesis,
there exist $S$ and $T$ such that and $\vdash P \cpar Q \cpar T$ and $\nfv{x}{S}$ and either $S = T$ or $S = \new x. T$, and also $\vcenter{\infer[]{R}{S}}$.
Furthermore, the size of the proof of $P \cpar Q \cpar T$ is no larger than the size of the proof of $\mathopen{\wen x.} \left( P \cpar Q \right) \cpar R$; hence strictly bounded by the size of the proof of $\wen x. P \cpar \wen x. Q \cpar R$.
Since $\nfv{x}{S}$, if $S = T$ then, by the \textit{new wen} and \textit{extrude new} rules, 
$\vcenter{\infer[]{\wen x. Q \cpar T}{\infer[]{\new x. Q \cpar T}{\mathopen{\new x.} \left( Q \cpar T \right)}}}$.
If $S = \new x. T$ then, by the \textit{close} rule, $\vcenter{\infer[]{\wen x. Q \cpar \new x. T}{\mathopen{\new x.}\left( Q \cpar T \right)}}$.
So in either case, $\vcenter{\infer[]{\wen x. Q \cpar S}{\mathopen{\new x.} \left( Q \cpar T \right)}}$,
and hence the derivation
$
\vcenter{
\infer[]{
\wen x. Q \cpar R
}{
\infer[]{
\wen x. Q \cpar S
}{
\mathopen{\new x.}\left(Q \cpar T \right)
}}}
$
can be constructed, as required.
The principal cases for \textit{left wen} and \textit{right wen} are similar.

\fullproof{
Another similar principal cases for $\textit{wen}$ is where the first the bottommost rule of a proof is of the form
$
\infer[]{
\wen x. P \cpar Q \cpar R
}{
\mathopen{\wen x.} \left( P \cpar Q \right) \cpar R
}
$, where $\vdash \mathopen{\wen x.} \left( P \cpar Q \right) \cpar R$ and $\nfv{x}{Q \cpar R}$.
By induction,
there exist $S$ and $T$ such that and $\vdash P \cpar Q \cpar T$ and either $S = T$ or $S = \new x. T$, and also $\vcenter{\infer[]{R}{S}}$.
Furthermore, the size of the proof of $P \cpar Q \cpar T$ is no larger than the size of the proof of $\mathopen{\wen x.} \left( P \cpar Q \right) \cpar R$; hence strictly bounded by the size of the proof of $\wen x. P \cpar Q \cpar R$.
If $S = T$ define $U = Q \cpar T$, and if $S = \new x. T$ define $U = \mathopen{\new x.} \left( Q \cpar T \right)$.
In the case $S = \new x. T$, since $\nfv{x}{Q}$, $\infer[]{Q \cpar \new x. T}{ \mathopen{\new x.} \left( Q \cpar T \right) }$.
Hence the following derivation
$
\infer[]{
Q \cpar R
}{
\infer[]{
Q \cpar S
}{
U
}}
$ can be constructed, as required.
}
%

\item Consider the principal case for $\textit{wen}$ when the bottommost rule of a proof is an instance of the \textit{fresh} rule of the form 
$\vcenter{\infer[]{\wen x. \wen \vec{y}. P \cpar Q}{ \wen \vec{y}. \new x. P \cpar Q }}$, 
where $\vdash \wen \vec{y}. \new x. P \cpar Q$. Notice that $\vec{y}$ is required to handle the effect of \textit{equivariance}.
By applying the induction hypothesis inductively on the length of $\vec{y}$, there exist $\vec{z}$ and $\hat{Q}$ such that $\vec{z} \subseteq \vec{y}$ and $\nfv{\vec{y}}{\new \vec{z} \hat{Q}}$ and $\vdash \new x. P \cpar \hat{Q}$, and also $\vcenter{\infer[]{ Q }{ \new \vec{z}. \hat{Q} }}$. Furthermore, the size of the proof of $\new x. P \cpar \hat{Q}$ is bounded above by the size of the proof of $\wen \vec{y}. \new x. P \cpar Q$.
By the induction hypothesis, there exist $R$ and $S$ such that $\nfv{x}{R}$, $\vdash P \cpar S$ and either $R = S$ or $R = \wen x. S$, and also
$
\vcenter{ \infer[]{\hat{Q}}{R}}
$.
There are two cases to consider. If $R = S$ then let $T = \new \vec{z}. S$; and if $R = \wen x. S$ then let $T = \new x. \new \vec{z}. S$, in which case, since $\new \vec{z}. \new x. S \equiv \new x. \new \vec{z}. S$ we have 
$\vcenter{\infer[]{\new \vec{z}. R}{T}}$. 
In either case $\nfv{x}{T}$. Thereby we can construct the derivation
$
\vcenter{
\infer[]{
Q
}{
\infer[]{
\new \vec{z}. \hat{Q}
}{
\infer[]{
\new \vec{z}. R
}{
T
}}}}
$.
Furthermore, appealing to Lemma~\ref{lemma:close}, the proof
$
\vcenter{
\infer[]{
\wen \vec{y}. P \cpar \new \vec{z}. S
}{
\infer[]{
\mathopen{\new \vec{y}.} \left( P \cpar S \right)
}{
\infer[]{
\new \vec{y}. \cunit
}{
\cunit
}}}}
$
can be constructed and, furthermore, $\size{ \wen \vec{y}. P \cpar \new \vec{z}. S } \prec \size{ \wen x. \wen \vec{y}. P \cpar Q }$, since by Lemma~\ref{lemma:bound} $\size{ \new \vec{z}. S } \preceq \size{ Q }$ and the \textit{wen} count strictly decreases.
\end{enumerate}

\item \textbf{Principal cases for new.}
The principal cases for \textit{new} are where the rules \textit{close}, \textit{extrude new}, \textit{medial new} and \textit{new wen} rules interfere directly with the $\textit{new}$ quantifier at the root of the principal formula.
Three cases are presented.

\begin{enumerate}[label*=\textbf{.\arabic*}]
\item The first principal case for $\textit{new}$ is when the bottommost rule of a proof is an instance of the \textit{close} rules of the form
$
\vcenter{
\infer[]{
\new x. P \cpar \wen x. Q \cpar R
}{
\mathopen{\new x.} \left( P \cpar Q \right) \cpar R
}}
$,
where $\vdash \mathopen{\new x.} \left( P \cpar Q \right) \cpar R$.
By the induction hypothesis, there exist formulae $U$ and $V$ such that $\vdash P \cpar Q \cpar V$ and $\nfv{x}{U}$ and either $U = V$ or $U = \wen x. V$, and also we have derivation
$\vcenter{\infer[]{R}{U}}$.
Furthermore, the size of the proof of $P \cpar Q \cpar V$ is no larger than the size of the proof of $\mathopen{\new x.} \left( P \cpar Q \right) \cpar R$; hence strictly bounded by the size of the proof of $\mathopen{\new x.} P \cpar \wen x. Q \cpar R$.
In the case $U = V$, we have $
\vcenter{\infer[]{
\wen x. Q \cpar V
}{
\mathopen{\wen x.} \left( Q \cpar V \right)
}}
$, since $\nfv{x}{U}$.
In the case $U = \wen x. V$, we have
$
\vcenter{\infer[]{
\wen x. Q \cpar \wen x. V
}{
\mathopen{\wen x.} \left( Q \cpar V \right)
}}
$.
Hence, by applying one of the above cases the following derivation
$
\vcenter{
\infer[]{
\wen x. Q \cpar R
}{
\infer[]{
\wen x. Q \cpar U
}{
\mathopen{\wen x.} \left( Q \cpar V \right)
}}}
$ can be constructed as required.
The principal case where the bottommost rule in a proof is the \textit{extrude new} rule follows a similar pattern.

\fullproof{
The second principal case for \textit{new} is when the bottommost rule of a proof is an instance of the \textit{extrude new} rule as follows, where $\nfv{x}{Q}$.
$
\vcenter{
\infer[]{
\new x. P \cpar Q \cpar R
}{
\mathopen{\new x.} \left( P \cpar Q \right) \cpar R
}}
$
where $\vdash \mathopen{\new x.} \left( P \cpar Q \right) \cpar R$ is provable.
By, the induction hypothesis, there exist formulae $U$ and $V$ where $\vdash P \cpar Q \cpar V$ and either $U = V$ or $U = \wen x. V$, and also
$\vcenter{\infer[]{R}{U}}$.
Furthermore, the size of the proof of $P \cpar Q \cpar V$ is bounded above by the size of the proof of $\mathopen{\new x.} \left( P \cpar Q \right) \cpar R$; thereby strictly bounded by the size of the proof of $\new x. P \cpar Q \cpar R$.
If $U = \wen x. V$, define $W = \mathopen{\wen x.} \left( Q \cpar V \right)$, and if $U = V$ define $W = Q \cpar V$.
In the cases where $U = \wen x. V$, since $\nfv{x}{Q}$, 
$\vcenter{\infer[]{Q \cpar \wen x. V}{\mathopen{\wen x.} \left( Q \cpar V \right)}}$, where the premise equals $W$.
In the cases where $U = V$, $Q \cpar U = W$.
Hence the derivation 
$
\vcenter{
\infer[]{
Q \cpar R
}{
\infer[]{
Q \cpar U
}{
W
}}}
$
can be constructed, as required.
}

\item Consider the second principal case for \textit{new} where the \textit{medial new} rule is the bottommost rule of a proof of the form
\[
\infer[\mbox{such that $\vdash \mathopen{\new \vec{y}.} \left(\new x. P \cseq \new x. Q\right) \cpar R$.}]{
\mathopen{\new x. \new \vec{y}.}\left( P \cseq Q \right) \cpar R
}{
\mathopen{\new \vec{y}.} \left( \new x. P \cseq \new x. Q \right) \cpar R
}
\]
The $\vec{y}$ is required to handle cases induced by equivariance.
By applying the induction hypothesis repeatedly, 
there exists $\vec{z}$ and $\hat{R}$ such that $\vec{z} \subseteq \vec{y}$ and
$\nfv{\vec{y}}{\wen \vec{z}. \hat{R} }$ and $\vdash \left( \new x. P \cseq \new x. Q \right) \cpar \hat{R}$, and also $\vcenter{\infer[]{R}{\hat{R}}}$.
Furthermore, the size of the proof of $\left( \new x. P \cseq \new x. Q \right) \cpar \hat{R}$ is bounded above by the size of the proof of $\mathopen{\new \vec{y}.} \left( \new x. P \cseq \new x. Q \right) \cpar R$.
By the induction hypothesis, there exist $S_i$ and $T_i$ such that $\vdash \new x. P \cpar S_i$ and $\vdash \new x. Q \cpar T_i$, for $1 \leq i \leq n$, and $n$-ary killing context such that 
$\vcenter{
\infer[]{\hat{R}}{\tcontext{ S_1 \cseq T_1, S_2 \cseq T_2, \hdots, S_n \cseq T_n }}}
$.
Furthermore, the size of the proofs of $\new x. P \cpar S_i$ and $\new x. Q \cpar T_i$ are bounded above by the size of the proof of $\left(\new x. P \cseq \new x. Q\right) \cpar R$.
%
By the induction hypothesis again, there exist $U^i$ and $\hat{U}^i$ such that $\vdash P \cpar \hat{U}^i$ and $\nfv{x}{U^i}$ and either $U^i = \hat{U}^i$ or $U^i = \wen x. \hat{U}^i$,
and also
$
\vcenter{\infer[]{S_i}{U^i}}
$.
Also by the induction hypothesis, there exist $V^i$ and $\hat{V}^i$ such that $\vdash Q \cpar \hat{V}^i$ and $\nfv{x}{V^i}$ and either $V^i = \hat{V}^i$ or $V^i = \wen x. \hat{V}^i$, 
and also
$
\vcenter{\infer[]{T_i}{V^i}}
$.
Now define $W$ and $\hat{W}$ such that
$
\hat{W} = \wen \vec{z}. \tcontext{ \hat{U}^i \cseq \hat{V}^i \colon 1 \leq i \leq n }
$
and, 
if for all $1 \leq i \leq n$, $U^i = \hat{U}^i$ and $V^i = \hat{V}^i$, then $W = \hat{W}$;
otherwise $W = \wen x. \hat{W}$.
%
Hence for each $i$, one of the following derivations holds.
\begin{itemize}
\item 
$U^i = \hat{U}^i$ and $V^i = \hat{V}^i$ hence
$U^i \cseq V^i = \hat{U}^i \cseq \hat{V}^i$.

\item
If $U^i = \wen x. \hat{U}^i$ and $V^i = \hat{V}^i$, hence $\nfv{x}{V^i}$, by the \textit{left wen} rule
$
\vcenter{
\infer[]{
\wen x. \hat{U}^i \cseq \hat{V}^i
}{
\mathopen{\wen x.} \left( \hat{U}^i \cseq \hat{V}^i \right)
}}
$.

\item
If $U^i = \hat{U}^i$, hence $\nfv{x}{\hat{U}^i}$, and $V^i = \wen x. \hat{V}^i$, by the \textit{right wen} rule
$
\vcenter{
\infer[]{
U^i \cseq \wen x. \hat{V}^i
}{
\mathopen{\wen x.} \left( \hat{U}^i \cseq \hat{V}^i \right)
}}
$.

\item
Otherwise by the \textit{suspend} rule $
\vcenter{\infer[]{
\wen x. \hat{U}^i \cseq \wen x. \hat{V}^i
}{
\mathopen{\wen x.} \left( \hat{U}^i \cseq \hat{V}^i \right)
}}
$
\end{itemize}
If for all $i$ such that $1 \leq i \leq n$, $U^i = \hat{U}^i$ and $V^i = \hat{V}^i$ then $W = \hat{W}$. Otherwise, by Lemma~\ref{lemma:commute},
$
\vcenter{
\infer[]{
\wen \vec{z}. \tcontext{
 U^i \cseq V^i \colon 1 \leq i \leq n
}
}{
\wen \vec{z}. \wen x. \tcontext{
 \hat{U}^i \cseq \hat{V}^i \colon 1 \leq i \leq n
}
}}$, where the premise is equialent to $W$.
Thereby the  derivation below left can be constructed, and furthermore, using Lemma~\ref{lemma:close}, 
the proof below right can also be constructed.
\[
\infer[]{
R
}{
\infer[]{
\wen \vec{z}. \hat{R}
}{
\infer[]{
\wen \vec{z}. \tcontext{ S_i \cseq T_i \colon 1 \leq i \leq n }
}{
\infer[]{
\wen \vec{z}. \tcontext{ U^i \cseq V^i \colon 1 \leq i \leq n }
}{
W
}}}}
\qquad
\infer[]{
\mathopen{\new \vec{y}.} \left( P \cseq Q \right) \cpar \hat{W}
}{
\infer[]{
\mathopen{\new \vec{y}.} \left(\left( P \cseq Q \right) \cpar
\tcontext{
 \hat{U}^i \cseq \hat{V}^i \colon 1 \leq i \leq n
}
\right)
}{
\infer[]{
\mathopen{\new \vec{y}.} 
\tcontext{
 \left( P \cseq Q \right) \cpar \left(\hat{U}^i \cseq \hat{V}^i\right) \colon 1 \leq i \leq n
}
}{
\infer[]{
\mathopen{\new \vec{y}.} 
\tcontext{
 \left( P \cpar \hat{U}^i \right) \cseq \left( Q \cpar \hat{V}^i\right) \colon 1 \leq i \leq n
}
}{
\infer[]{
\mathopen{\new \vec{y}.} 
\tcontext{
 \cunit \colon 1 \leq i \leq n
}
}{
\cunit
}}}}}
\]

By Lemma~\ref{lemma:bound},
$\size{\hat{W}} \preceq \size{R}$; hence $\size{ \mathopen{\new \vec{y}.} \left( P \cseq Q \right) \cpar \hat{W} } \prec \size{ \mathopen{\new x. \new \vec{y}.} \left( P \cseq Q \right) \cpar R }$ since the \textit{new count} strictly decreases, as required.
\smallskip

\item Consider the third principal  case for \textit{new} where the bottommost rule of a proof is the \textit{new wen} rule of the form
\[
\infer[\mbox{, where $\vdash \new \vec{z}. \wen y. \new x. P \cpar Q$.} ]{
\new x. \new \vec{z}. \wen y. P \cpar Q
}{
\new \vec{z}. \wen y. \new x. P \cpar Q
}
\]
By applying the induction hypothesis repeatedly, 
there exist $\vec{w}$ and $\hat{Q}$ such that $\vec{w} \subseteq \vec{z}$
and $\nfv{\vec{z}}{\wen \vec{w} .\hat{Q}}$ and $\vdash \wen y. \new x. P \cpar \hat{Q}$, and also $\vcenter{\infer[]{Q}{\wen \vec{w}. \hat{Q}}}$. Furthermore, the size of the proof of $\wen y. \new x. P \cpar \hat{Q}$ is bounded above by the size of the proof of $\new \vec{z}. \wen y.\new x. P \cpar Q$.
By the induction hypothesis, there exist $R$ and $S$ such that $\nfv{x}{R}$ and $\vdash \new x. P \cpar S$ and either $R = S$ or $R = \new y. S$, and also 
$
\vcenter{
\infer[]{\hat{Q}}{R}}
$.
Furthermore, the size of the proof of $\new x. P \cpar S$ is bounded above by the size of the proof of $\wen y. \new x. P \cpar Q$, hence strictly bounded above by the size of the proof of $\new x. \wen y. P \cpar Q$ enabling the induction hypothesis.
By the induction hypothesis again, there exist $U$ and $V$ such that $\nfv{x}{U}$ and $\vdash P \cpar V$ and either $U = V$ or $U = \wen x. V$, and also
$
\vcenter{
\infer[]{S}{U}}
$.

Let $W$ and $\hat{W}$ be defined such that, if $R = \new y. S$, then $\hat{W} = \new y. V$;
or, if $R = S$, then $\hat{W} = V$.
If $V = U$ then define $W = \wen \vec{w}. \hat{W}$.
If $U = \wen x. V$, then define $W = \wen x. \wen \vec{w}. \hat{W}$.
There are four scenarios for constructing a derivation with premise $W$ and conclusion $\wen \vec{w}. R$.
\begin{itemize}
\item In the case $V = U$ and $R = \new y. S$ then $\wen \vec{w}. \new y. U = W$.

\item If $V = U$ and $R = S$ then $\wen \vec{w}. U = W$.

\item If both $U = \wen x. V$ and $R = \new y. S$ hold, then 
we have 
\[
\infer[\mbox{, where the premise is $W$.}]{
\wen \vec{w}.R
}{
\infer[]{
 \wen \vec{w}. \new y. \wen x. V 
}{
\wen x. \wen \vec{w}. \new y. V
}}
\]

\item If both $U = \wen x. V$ and $R = S$ then 
$
\vcenter{
\infer[]{
\wen \vec{w}.R
}{
 \wen \vec{w}. U
}}
$,
where the premise is equivalent to 
$W$.
\end{itemize}
Thereby, by applying one of the above cases,  we have 
$\vcenter{
\infer[.]{
\hat{Q}
}{
\infer[]{
\wen \vec{w}. Q
}{
\infer[]{
\wen \vec{w}. R
}{
W
}}}}
$

In the case that $\hat{W} = \new y. V$,  the left most derivation below holds. 
In the case, $\hat{W} = V$ and $\nfv{y}{V}$ the middle derivation below holds.  
Hence in either case, appealing to Lemma~\ref{lemma:close}, the proof below right can be constructed:
\[
\infer[]{
\wen y. P \cpar \new y. V
}{
\mathopen{\new y.} \left( P \cpar V \right)
}
\qquad
\infer[]{
\wen y. P \cpar \hat{W}
}{
\infer[]{
\mathopen{\wen y.} \left( P \cpar V \right)
}{
\mathopen{\new y.} \left( P \cpar V \right)
}}
\qquad
\infer[]{
\new \vec{z}. \wen y. P \cpar \wen \vec{w}. \hat{W}
}{
\infer[]{
\mathopen{\new \vec{z}.} \left( \wen y. P \cpar \hat{W} \right)
}{
\infer[]{
\mathopen{\new \vec{z}. \new y.} \left( P \cpar V \right)
}{
\infer[]{
\new \vec{z}. \new y. \cunit
}{
\cunit
}}}}
\]

\noindent Furthermore, by Lemma~\ref{lemma:bound}, $\size{ \wen \vec{w}. \hat{W} } \preceq \size{ Q }$.
Hence $\size{\wen y. P \cpar \wen \vec{w}. \hat{W}} \prec \size{ \new x. \new \vec{z}. \wen y. P \cpar Q }$ since the \textit{new} count strictly decreases.

\end{enumerate}

\item \textbf{Principal cases for seq.} There are two forms of principal cases for \textit{seq}.
The first case, induced by the \textit{sequence} rule, is the case that forces the \textit{medial}, \textit{medial1} and \textit{medial new} rules.
The other cases are induced by the \textit{suspend}, \textit{left wen} and \textit{right wen} rules (which are forced as a knock on effect of the \textit{medial new} rule).

\begin{enumerate}[label*=\textbf{.\arabic*}]
\item Consider the first principal case for \textit{seq}.
The difficulty in this case is that, due to associativity of \textit{seq}, the \rseq rule may be applied in several ways when there are multiple occurrences of \textit{seq}.
Consider a principal formula of the form $\left(T_0 \cseq T_1\right) \cseq T_2$, where we aim to split the formula around the second \textit{seq} operator. The difficulty is that the bottommost rule may be an instance of the \rseq rule applied between $T_0$ and $T_1 \cseq T_2$. Symmetrically, the principal formula may be of the form $T_0 \cseq \left(T_1 \cseq T_2\right)$ but the bottommost rule may be an instance of the \rseq rule applied between $T_0 \cseq T_1$ and $T_2$.
In the following analysis, only the former case is considered; the symmetric case follows a similar pattern.
The principal formula is $\left(T_0 \cseq T_1\right) \cseq T_2$ and the bottommost rule is an instance of the \textit{sequence} rule of the form
\[
\infer[]{
\left( T_0 \cseq T_1 \cseq T_2 \right) \cpar \left( U \cseq V \right) \cpar W
}{
\left(\left( T_0 \cpar U \right) \cseq \left( \left(T_1 \cseq T_2\right) \cpar V \right)\right) \cpar W
}
\] 
where $T_0 \not\equiv \cunit$, $T_2 \not\equiv \cunit$ (otherwise splitting is trivial), and either $U \not\equiv \cunit$ or $V \not\equiv \cunit$ (otherwise the \rseq rule cannot be applied);
and also $\vdash \left(\left( T_0 \cpar U \right) \cseq \left( \left(T_1 \cseq T_2\right) \cpar V \right)\right) \cpar W$.
By the induction hypothesis, 
there exist $P_i$ and $Q_i$ such that $\vdash T_0 \cpar U \cpar P_i$ and $\vdash \left(T_1 \cseq T_2\right) \cpar V \cpar Q_i$ hold, for $1 \leq i \leq n$, and an $n$-ary killing context $\tcontext{}$ such that 
\[
\infer[.]{W}{
\tcontext{ P_1 \cseq Q_1, \hdots, P_n \cseq Q_n }
}
\]
Furthermore,
the size of the proof of formula
$
\left(T_1 \cseq T_2\right) \cpar V \cpar Q_i
$
is bounded above by the size of the proof of
$\left(\left( T_0 \cpar U \right) \cseq \left( \left(T_1 \cseq T_2\right) \cpar V \right)\right) \cpar W$, hence the induction hypothesis is enabled.
By the induction hypothesis, 
there exists $R^i_j$ and $S^i_j$ such that $\vdash T_1 \cpar R^i_j$ and $\vdash T_2 \cpar S^i_j$, for $1 \leq j \leq m_i$, and $m_i$-ary killing context $\tcontextn{i}{}$ such that
\[
\infer[.]{
V \cpar Q_i
}{
\tcontextn{i}{R^i_1 \cseq S^i_1, \hdots, R^i_{m_i} \cseq S^i_{m_i}}
}
\]
Furthermore, by Lemma~\ref{lemma:medial} there exist killing contexts $\tcontextmn{i}{0}{}$ and $\tcontextmn{i}{1}{}$ and sets of integers $J^i \subseteq \left\{1, \hdots, n\right\}$, $K^i \subseteq \left\{1, \hdots, n\right\}$ such that
\[
\infer[.]{
\tcontextn{i}{R^i_1 \cseq S^i_1, \hdots, R^i_{m_i} \cseq S^i_{m_i}}
}{
\tcontextmn{i}{0}{R^i_j \colon j \in J^i}
 \cseq
\tcontextmn{i}{1}{S^i_k \colon k \in K^i}
}
\]
Thereby, the following derivation can be constructed.
\[
\infer[]{
\left( U \cseq V \right) \cpar W
}{
\infer[]{
\left( U \cseq V \right) \cpar \tcontext{ P_1 \cseq Q_1, \hdots, P_n \cseq Q_n } 
}{
\infer[]{
 \tcontext{ \left( U \cseq V \right) \cpar \left( P_1 \cseq Q_1 \right), \hdots, \left( U \cseq V \right) \cpar \left( P_n \cseq Q_n \right) } 
}{
\infer[]{
\tcontext{ \left( U \cpar P_1 \right) \cseq \left( V \cpar Q_1 \right), \hdots, \left( U \cpar P_n \right) \cseq \left( V \cpar Q_n \right) } 
}{
\infer[]{
 \tcontext{ \left( U \cpar P_i \right) \cseq \tcontextn{i}{R^i_j \cseq S^i_j \colon 1 \leq j \leq m_i } \colon 1 \leq i \leq n } 
}{
\tcontext{
 \left( U \cpar P_i \right) \cseq \tcontextmn{i}{0}{R^i_j \colon j \in J^i} \cseq \tcontextmn{i}{1}{S^i_k \colon k \in K^i}
 \colon 1 \leq i \leq n 
} 
}}}}}
\]
Furthermore, the following two proofs can be constructed.
\[
\infer[]{
T_2 \cpar \tcontextn{i}{ S^i_j \colon 1 \leq j \leq m_i }
}{
\infer[]{
\tcontextn{i}{ T_2 \cpar S^i_j \colon 1 \leq j \leq m_i }
}{
\infer[]{
\tcontextn{i}{ \cunit \colon 1 \leq j \leq m_i }
}{
\cunit
}}}
\qquad\qquad
\infer[]{
\left(T_0 \cseq T_1\right) \cpar \left(\left( U \cpar P_i \right) \cseq \tcontextn{i}{R^i_j \colon 1 \leq j \leq m_i }\right)
}{
\infer[]{
\left(T_0 \cpar U \cpar P_i\right) \cseq \left(T_1 \cpar \tcontextn{i}{R^i_j \colon 1 \leq j \leq m_i }\right)
}{
\infer[]{
T_1 \cpar \tcontextn{i}{R^i_j \colon 1 \leq j \leq m_i }
}{
\infer[]{
\tcontextn{i}{T_1 \cpar R^i_j \colon 1 \leq j \leq m_i }
}{
\infer[]{
\tcontextn{i}{\cunit \colon 1 \leq j \leq m_i }
}{
\cunit
}}}}}
\]
By Lemma~\ref{lemma:bound}, 
\[\size{ \tcontext{
 \left( U \cpar P_1 \right) \cseq \tcontextmn{i}{0}{R^i_j \colon j \in J^i} \cseq \tcontextmn{i}{1}{S^i_k \colon k \in K^i}
 \colon 1 \leq i \leq n 
} }  \preceq \size{ \left( U \cseq V \right) \cpar W }
\]
which are also upper bounds for
$
\size{ \tcontextmn{i}{0}{R^i_j \colon j \in J^i} 
}$
and
$\size{
\tcontextmn{i}{1}{S^i_k \colon k \in K^i}
}$.
Furthermore, $T_0 \not\equiv \cunit$ and $T_2 \not\equiv \cunit$ both $\occ{ T_0 } \mstrict \occ{ T_0 \cseq T_1 \cseq T_2 }$ and  $\occ{ T_2 } \mstrict \occ {T_0 \cseq T_1 \cseq T_2 }$
Hence the sizes of the above proofs of 
$T_2 \cpar \tcontextn{i}{ S^i_j \colon 1 \leq j \leq m_i }$ 
and 
\[
\left(T_0 \cseq T_1\right) \cpar \left(\left( U \cpar P_i \right) \cseq \tcontextn{i}{R^i_j \colon 1 \leq j \leq m_i }\right)
\]
 are strictly less than the size of the proof of $\left( T_0 \cseq T_1 \cseq T_2 \right) \cpar \left( U \cseq V \right) \cpar W$.

\input{seq-medial}

\end{enumerate}

\item \textbf{Principal case for times.}
There is only one principal case for \textit{times}, which does not differ significantly from the corresponding case in \textsf{BV} and its extensions. A proof may begin with an instance of the \textit{switch} rule of the form
\[
\infer[\mbox{where $\vdash \left(T_0 \tensor U_0 \tensor \left( \left(T_1 \tensor U_1\right) \cpar V \right)\right) \cpar W$,}]{
\left( T_0 \tensor T_1 \tensor U_0 \tensor U_1 \right) \cpar V \cpar W
}{
\left(T_0 \tensor U_0 \tensor \left( \left(T_1 \tensor U_1\right) \cpar V \right)\right) \cpar W
}
\]
such that $T_0 \tensor U_0 \not\equiv \cunit$ and $V \not\equiv \cunit$ (otherwise the \textit{switch} rule cannot be applied), and also $T_0 \tensor T_1 \not\equiv \cunit$ and $U_0 \tensor U_1 \not\equiv \cunit$ (otherwise splitting holds trivially).
By the induction hypothesis, there exist $R_i$ and $S_i$ such that $\vdash \left(T_0 \tensor U_0\right) \cpar R_i$ and $\vdash \left(T_1 \tensor U_1\right) \cpar V \cpar S_i$ hold, for $1 \leq i \leq n$, and an $n$-ary killing context $\tcontext{}$ such that 
derivation $\vcenter{\infer{W}{\tcontext{ R_1 \cpar S_1, \hdots, R_n \cpar S_n }}}$ holds.
Furthermore $\size{\left(T_0 \tensor U_0\right) \cpar R_i}$ and $\size{\left(T_1 \tensor U_1\right) \cpar V \cpar S_i}$ are bounded above by $\size{\left(T_0 \tensor U_0 \tensor \left( \left(T_1 \tensor U_1\right) \cpar V \right)\right) \cpar W}$.
Hence, by the induction hypothesis twice there exist formulae $P^{i,0}_j$, $Q^{i,0}_j$, $P^{i,1}_k$ and $Q^{i,1}_k$ such that $\vdash T_0 \cpar P^{i,0}_j$, $\vdash U_0 \cpar Q^{i,0}_j$, $\vdash T_1 \cpar P^{i,1}_k$ and $\vdash U_1 \cpar Q^{i,1}_k$, for $1 \leq j \leq m^0_i$ and $1 \leq k \leq m^1_i$, and $m^0_i$-ary killing context $\tcontextmn{0}{i}{}$ and $m^1_i$-ary killing context $\tcontextmn{1}{i}{}$ such that derivations
\[
\vcenter{
\infer[]{
R_i}{
\tcontextmn{0}{i}{P^{i,0}_j \cpar Q^{i,0}_j \colon 1 \leq j \leq m^0_i}
}
}
\quad
\mbox{ and }
\quad
\vcenter{
\infer[]{
V \cpar S_i
}{
\tcontextmn{1}{i}{P^{i,1}_k \cpar Q^{i,1}_k \colon 1 \leq k \leq m^1_i}
}}
\] 
can be constructed.
Thereby the following derivation can be constructed.
\[
\infer[]{
V \cpar W
}{
\infer[]{
V \cpar \tcontext{ R_i \cpar S_i \colon 1 \leq i \leq n } 
}{
\infer[]{
\tcontext{ R_i \cpar V \cpar S_i \colon 1 \leq i \leq n }
}{
\infer[]{
\tcontext{
  \tcontextmn{0}{i}{P^{i,0}_j \cpar Q^{i,0}_j \colon 1 \leq j \leq m^0_i}
   \cpar
  \tcontextmn{1}{i}{P^{i,1}_k \cpar Q^{i,1}_k \colon 1 \leq k \leq m^1_i}
\colon 1 \leq i \leq n }
}{
\infer[]{
 \tcontext{ \tcontextmn{1}{i}{
  \tcontextmn{0}{i}{P^{i,0}_j \cpar Q^{i,0}_j \colon 1 \leq j \leq m^0_i}
  \cpar
  P^{i,1}_k \cpar Q^{i,1}_k \colon 1 \leq k \leq m^1_i
} \colon 1 \leq i \leq n }
}{
 \tcontext{ \tcontextmn{1}{i}{
  \tcontextmn{0}{i}{
  P^{i,0}_j \cpar P^{i,1}_k \cpar Q^{i,0}_j \cpar Q^{i,1}_k
  \colon 1 \leq j \leq m^0_i
  }
 \colon 1 \leq k \leq m^1_i
 } \colon 1 \leq i \leq n }
}}}}}
\]
Now observe that the following two proofs can be constructed.
\[
\infer[]{
\left(T_0 \tensor T_1\right)
\cpar
P^{i,0}_j \cpar P^{i,1}_k
}{
\infer[]{
\left(T_0 \cpar P^{i,0}_j\right) \tensor \left(T_1 \cpar P^{i,1}_k\right)
}{
\cunit
}}
\qquad\qquad\qquad\qquad
\infer[]{
\left(U_0 \tensor U_1\right)
\cpar
Q^{i,0}_j \cpar Q^{i,1}_k
}{
\infer[]{
\left(U_0 \cpar Q^{i,0}_j\right) \tensor \left(U_1 \cpar Q^{i,1}_k\right)
}{
 \cunit
}}
\]
Furthermore,
$\occ{ T_0 \tensor T_1 } \mstrict \occ{ T_0 \tensor T_1 \tensor U_0 \tensor U_1 }$
and 
$\occ{ U_0 \tensor U_1 } \mstrict \occ{ T_0 \tensor T_1 \tensor U_0 \tensor U_1 }$,
since $T_0 \tensor T_1 \not\equiv \cunit$ and $U_0 \tensor U_1 \not\equiv \cunit$.
Also, by Lemma~\ref{lemma:bound}, the following inequality holds.
\[
\size{
 \tcontext{ \tcontextmn{1}{i}{
  \tcontextmn{0}{i}{
  P^{i,0}_j \cpar P^{i,1}_k \cpar Q^{i,0}_j \cpar Q^{i,1}_k
  \colon 1 \leq j \leq m^0_i
  }
 \colon 1 \leq k \leq m^1_i
 } \colon 1 \leq i \leq n }
} \preceq \size{ V \cpar W }
\]
Hence both
$
\size{ P^{i,0}_j \cpar P^{i,1}_k }
\preceq \size{ V \cpar W }$
and $
\size{ Q^{i,0}_j \cpar Q^{i,1}_k }
\preceq \size{ V \cpar W }$ hold.
Thereby the size of each of the above proofs is strictly bounded above by the size of the proof of $\left(T_0 \tensor T_1 \tensor U_0 \tensor U_1 \right) \cpar V \cpar W$.

\item \textbf{Principal cases for with.}
There are three forms of principal case where the \textit{with} operator is directly involved in the bottommost rules.
Note that in \textsf{MAV} the \textit{with} operator is separated from the core splitting lemma, much like universal quantification in this paper.
However, in the case of \textsf{MAV1} the \textit{left name} and \textit{right name} rules introduce inter-dependencies between nominals and \textit{with}, forcing cases for \textit{with} to be checked in this lemma.

\begin{enumerate}[label*=\textbf{.\arabic*}]

\item Consider the principal case involving the \textit{extrude} rule.
In this case, the bottommost rule is of the form 
\[
\infer[\mbox{where $\vdash \left(P  \cpar R\right)\wwith\left( Q \cpar R\right) \cpar S$ holds.}]{
\left(P \wwith Q\right) \cpar R \cpar S
}{
\left(P  \cpar R\right)\wwith\left( Q \cpar R\right) \cpar S
}
\]
Now, by the induction hypothesis, 
since $\vdash \left(P  \cpar R\right)\wwith\left( Q \cpar R\right) \cpar S$ holds, we have that 
$\vdash P  \cpar R \cpar S$ and $\vdash Q \cpar R \cpar S$ hold, as required.

\item Consider the principal case involving the \textit{left name} rule.
In this case, the bottommost rule is of the form 
\[
\infer[\mbox{, where $\nfv{x}{Q}$, such that $\vdash \mathopen{\wen x.}\left( P \wwith Q \right) \cpar R$.}]{
\left( \wen x.P \wwith Q \right) \cpar R
}{
\mathopen{\wen x.}\left( P \wwith Q \right) \cpar R
}
\] 
By the induction hypothesis, there exist $S$ and $\hat{S}$ such that
$\vcenter{\infer[]{R}{S}}$ and $\nfv{x}{S}$ and $\vdash \left(P \wwith Q\right) \cpar \hat{S}$ 
and either $S = \hat{S}$ or $S = \mathopen{\new x.}\hat{S}$. Furthermore, the size of the proof of $\left(P \wwith Q\right) \cpar \hat{S}$ is strictly less than the size of the proof of $\left( \wen x.P \wwith Q \right) \cpar R$, since the \textit{wen} count strictly decreases, and by Lemma~\ref{lemma:bound}, $\size{\hat{S}} \leq \size{R}$.
By the induction hypothesis again, $\vdash P \cpar \hat{S}$ and $\vdash Q \cpar \hat{S}$ hold.

Now if $S = \hat{S}$ then $\nfv{x}{\hat{S}}$ and $\vdash Q \cpar S$ holds immediately, whereas 
$\vdash \wen x.P \cpar R$ is proved as below left. Otherwise, $S = \new x.\hat{S}$ and
$\vdash \wen x.P \cpar R$ is proved in the middle derivation below, whereas $\vdash  Q \cpar S$ is proved in the right derivation below.
\[
\infer[]{
\wen x.P \cpar R 
}{
\infer[]{
\wen x.P \cpar \hat S
}{
\infer[]
{\wen x.\left(P \cpar \hat S \right)}
{
\infer[]{
\mathopen{\new x.}\left(P \cpar \hat{S} \right)
}{
\infer[]{
\mathopen{\new x.}\cunit
}{
\cunit
}}}}}
\qquad \qquad
\infer[]{
\wen x.P \cpar R 
}{
\infer[]{
\wen x.P \cpar \new x. \hat S
}{
\infer[]{
\mathopen{\new x.}\left(P \cpar \hat{S} \right)
}{
\infer[]{
\mathopen{\new x.}\cunit
}{
\cunit
}}}}
\qquad
\qquad
\infer[.]{
Q \cpar \wen x.\hat{S}
}{
\infer[]{
\mathopen{\wen x.}\left( Q \cpar \hat{S} \right)
}{
\infer[]{
 \mathopen{\new x.}\left( Q \cpar \hat{S} \right)
}{
\infer[]{
 \new x.\cunit 
}{
 \cunit
}}}}
\]


Hence, in either case, $\vdash Q \cpar S$ 
and since $
\vcenter{
\infer[]{
Q \cpar R
}{
Q \cpar S
}}$,
we have that $\vdash Q \cpar R$ holds.
Thereby $\vdash \wen x.P \cpar R$ and $\vdash Q \cpar R$ hold, as required.
The case for the \textit{left name} rule, where $\new$ replaces $\wen$ is similar; as are the cases for the \textit{right name} and \textit{with name} rules.

\item Consider the principal case involving the \textit{medial} rule.
In this case, the bottommost rule of a proof is of the form 
\[
\infer[\mbox{such that $\vdash \left(\left( P \wwith R \right) \cseq \left( Q \wwith S \right)\right) \cpar W$ holds.}]{
\left(\left( P \cseq Q \right) \wwith \left( R \cseq S \right)\right) \cpar W
}{
\left(\left( P \wwith R \right) \cseq \left( Q \wwith S \right)\right) \cpar W
}
\]
By the induction hypothesis, for $1 \leq i \leq n$ there exists $U_i$ and $V_i$ such that $\vdash \left( P \wwith R \right) \cpar U_i$ and $\vdash \left( Q \wwith S \right) \cpar V_i$ hold, and $n$-ary killing context $\tcontext{}$ such that 
$
\vcenter{\infer[]{W}{\tcontext{U_i \cseq V_i \colon 1 \leq i \leq n}}}$.
Furthermore, the size of the proofs of $\left( P \wwith R \right) \cpar U_i$ and $\left( Q \wwith S \right) \cpar V_i$ are strictly less than the size of the proof of $\left(\left( P \wwith R \right) \cseq \left( Q \wwith S \right)\right) \cpar W$. Hence by the induction hypothesis again, $\vdash P \cpar U_i$, $\vdash R \cpar U_i$, $\vdash Q \cpar V_i$ and $\vdash S \cpar V_i$.
Hence we can construct the following two proofs, as required.
\[
\infer[]{
\left( P \cseq Q \right) \cpar W
}{
\infer[]{
\left( P \cseq Q \right) \cpar \tcontext{U_i \cseq V_i \colon 1 \leq i \leq n}
}{
\infer[]{
\tcontext{\left( P \cseq Q \right) \cpar \left(U_i \cseq V_i\right) \colon 1 \leq i \leq n}
}{
\infer[]{
\tcontext{\left( P \cpar U_i \right) \cseq \left(Q \cpar V_i\right) \colon 1 \leq i \leq n}
}{
\infer[]{
\tcontext{\cunit  \colon 1 \leq i \leq n}
}{
\cunit
}}}}}
\qquad
\infer[]{
\left( R \cseq S \right) \cpar W
}{
\infer[]{
\left( R \cseq S \right) \cpar \tcontext{U_i \cseq V_i \colon 1 \leq i \leq n}
}{
\infer[]{
\tcontext{\left( R \cseq S \right) \cpar \left(U_i \cseq V_i\right) \colon 1 \leq i \leq n}
}{
\infer[]{
\tcontext{\left( R \cpar U_i \right) \cseq \left(S \cpar V_i\right) \colon 1 \leq i \leq n}
}{
\infer[]{
\tcontext{\cunit \colon 1 \leq i \leq n}
}{
\cunit
}}}}}
\]

\end{enumerate}

\input{equivariance}

\input{commutative}

\end{enumerate}

This covers all scenarios for the bottommost rule, hence splitting follows by induction over the size of the proof.
\end{proof}

%% file: seq-medial.tex
\item Consider the principal case for \textit{seq} where the bottommost rule of a proof is an instance of the \textit{suspend} rule of the form
\[
\infer[\mbox{, where $\vdash \left(P_0\cseq \mathopen{\wen x.} \left( P_1\cseq P_2  \right) \cseq P_3 \right) \cpar Q$ holds.}]{
\left(P_0\cseq \wen x. P_1\cseq \wen x. P_2  \cseq P_3 \right) \cpar Q
}{
\left(P_0\cseq \mathopen{\wen x.} \left( P_1\cseq P_2  \right) \cseq P_3 \right) \cpar Q
}
\] 
By induction, there exist $U^0_i$ and $U^1_i$ such that $\vdash P_0\cpar U^0_i$ and $\vdash \left( \mathopen{\wen x.} \left( P_1\cseq P_2  \right) \cseq P_3 \right) \cpar U^1_i$ hold, for $1 \leq i \leq n$,
and $n$-ary killing context $\tcontext{}$ such that 
$
\vcenter{
\infer[]{Q}{
\tcontext{ U^0_i \cseq U^1_i \colon 1 \leq i \leq n }
}}
$.
Furthermore the size of the proof of $\left( \mathopen{\wen x.} \left( P_1\cseq P_2  \right) \cseq P_3 \right) \cpar U^1_i$ is bounded above by the size of the proof of $\left(P_0\cseq \wen x. P_1\cseq \wen x. P_2  \cseq P_3 \right) \cpar Q$.
By induction again, there exist $V^i_j$ and $W^i_j$ such that $\vdash \mathopen{\wen x.} \left( P_1\cseq P_2  \right) \cpar V^i_j$ and $\vdash P_3 \cpar W^i_j$, for $1 \leq j \leq m_i$, and $m_i$-ary killing context $\tcontextn{i}{}$ such that the following derivation holds.
$
\vcenter{
\infer[]{
U^1_i
}{
\tcontextn{i}{ V^i_j \cseq W^i_j \colon 1 \leq j \leq m_i }
}}
$.
Furthermore, the size of the proof of $\mathopen{\wen x.} \left( P_1\cseq P_2  \right) \cpar V^i_j$ is bounded by the size of the proof of $\left( \mathopen{\wen x.} \left( P_1\cseq P_2  \right) \cseq P_3 \right) \cpar U^1_i$. 
By applying the induction hypothesis again, there exist $R^i_j$ and $\hat{R}^i_j$ such that $\nfv{x}{R^i_j}$ and $\vdash \left(P_1\cseq P_2 \right) \cpar \hat{R}^i_j$ and either $R^i_j = \hat{R}^i_j$ or $R^i_j = \new x. \hat{R}^i_j$, and also
$
\vcenter{
\infer[]{
V^i_j
}{
 R^i_j
}}
$.
Furthermore, the size of the proof of $\left(P_1\cseq P_2 \right) \cpar \hat{R}^i_j$ is bounded above by the size of the proof of $\left( \mathopen{\wen x.} \left( P_1\cseq P_2  \right) \cseq P_3 \right) \cpar U^1_i$.
By a fourth induction, there exist $S^{i,j}_k$ and $T^{i,j}_k$ such that both $\vdash P_1\cpar S^{i,j}_k$ and $\vdash P_2  \cpar T^{i,j}_k$ hold, for $1 \leq k \leq \ell^{i,j}$, and $\ell^{i,j}$-ary killing context $\tcontextn{i,j}{}$ such that the following derivation holds:
\[
\infer[]{
\hat{R}^i_j
}{
\tcontextn{i,j}{ S^{i,j}_1 \cseq T^{i,j}_1, S^{i,j}_2 \cseq T^{i,j}_2, \hdots, S^{i,j}_{\ell^{i,j}} \cseq T^{i,j}_{\ell^{i,j}} }
}.
\]
By Lemma~\ref{lemma:medial}, there exists some $I^i_j \subseteq \left\{ 1 \hdots \ell^{i,j} \right\}$ and $J^i_j \subseteq \left\{ 1 \hdots \ell^{i,j} \right\}$ and killing contexts $\tcontextmn{i,j}{0}{}$ and $\tcontextmn{i,j}{1}{}$ such that
\[
\infer[]{
\hat{R}^i_j
}{
\infer[]{
\tcontextn{i,j}{ S^{i,j}_k \cseq T^{i,j}_k \colon 1 \leq k \leq \ell^{i,j} }
}{
\tcontextmn{i,j}{0}{S^{i,j}_k \colon k \in I^i_j } \cseq
\tcontextmn{i,j}{1}{T^{i,j}_k \colon k \in J^i_j }
}}.
\]
Define $\hat{S}^i_j$ and $\hat{T}^i_j$ as follows.
If $R^i_j = \hat{R}^i_j$, then 
\[
\hat{S}^i_j = \tcontextmn{i,j}{0}{S^{i,j}_k \colon k \in I^i_j }
\mbox{ and }
\hat{T}^i_j = \tcontextmn{i,j}{1}{T^{i,j}_k \colon k \in J^i_j };
\]
and hence, we can construct the derivation 
\[
\infer[]{
R^i_j
}{
\tcontextmn{i,j}{0}{S^{i,j}_k \colon k \in I^i_j } \cseq
\tcontextmn{i,j}{1}{T^{i,j}_k \colon k \in J^i_j }
}
\] where the premise equals $\hat{S}^i_j  \cseq \hat{T}^i_j$.
If however $R^i_j = \new x. \hat{R}^i_j$, then define
\[
\hat{S}^i_j = \new x. \tcontextmn{i,j}{0}{S^{i,j}_k \colon k \in I^i_j }
\mbox{ and }
\hat{T}^i_j = \new x. \tcontextmn{i,j}{1}{T^{i,j}_k \colon k \in J^i_j };
\]
and hence, the derivation 
\[
\infer[]{
R^i_j
}{
\infer[]{
\mathopen{\new x.} \left(
 \tcontextmn{i,j}{0}{S^{i,j}_k \colon k \in I^i_j } \cseq
 \tcontextmn{i,j}{1}{T^{i,j}_k \colon k \in J^i_j }
\right)
}{
\hat{S}^i_j \cseq
\hat{T}^i_j
}}
\]
can be constructed.
By Lemma~\ref{lemma:medial}, for some $K^i \subseteq \left\{ 1\hdots m_i \right\}$, $L^i \subseteq \left\{ 1\hdots m_i \right\}$ and killing contexts $\tcontextmn{i}{0}{}$ and $\tcontextmn{i}{1}{}$, we obtain the following derivation:
\[
\infer[]{
\tcontextn{i}{
    \hat{S}^i_j \cseq 
    \hat{T}^i_j \cseq 
    W^i_j \colon 1 \leq j \leq m_i
}
}{
\tcontextmn{i}{0}{
    \hat{S}^i_j
    \colon j \in K^i
} \cseq
\tcontextmn{i}{1}{
    \hat{T}^i_j \cseq 
    W^i_j \colon j \in L^i
}
}
\]
By using the above derivations we can construct the following derivation:
\[
\infer[]{
Q
}{
\infer[]{
\tcontext{ U^0_i \cseq U^1_i \colon 1 \leq i \leq n }
}{
\infer[]{
\tcontext{ U^0_i \cseq \tcontextn{i}{ V^i_j \cseq W^i_j \colon 1 \leq j \leq m_i } \colon 1 \leq i \leq n }
}{
\infer[]{
\tcontext{ U^0_i \cseq \tcontextn{i}{ 
                                    R^i_j 
                                                     \cseq W^i_j \colon 1 \leq j \leq m_i
                                    } \colon 1 \leq i \leq n }
}{
\infer[]{
\tcontext{
           U^0_i \cseq \tcontextn{i}{ 
                                                                        \hat{S}^i_j \cseq 
                                                                        \hat{T}^i_j 
                                                     \cseq W^i_j \colon 1 \leq j \leq m_i
                                    }  \colon 1 \leq i \leq n }
}{
\tcontext{ 
           U^0_i \cseq
             \tcontextmn{i}{0}{
    \hat{S}^i_j 
    \colon j \in K^i
} \cseq
\tcontextmn{i}{1}{
    \hat{T}^i_j \cseq 
    W^i_j \colon j \in L^i
}
           \colon 1 \leq i \leq n  }
}}}}}
\]

\noindent Consider whether the judgement $\vdash \wen x. P_1 \cpar \hat{S}^i_j$ holds.
We have two cases: in the first, $\hat{S}^i_j = \tcontextmn{i,j}{0}{S^{i,j}_k \colon k \in I^i_j }$ and $\nfv{x}{\hat{S}^i_j}$; in the second $\hat{S}^i_j = \new x. \tcontextmn{i,j}{0}{S^{i,j}_k \colon k \in I^i_j }$. In each case, one of the following derivations can be respectively constructed.
\[
\infer[]{
\wen x. P_1\cpar \tcontextmn{i,j}{0}{S^{i,j}_k \colon k \in I^i_j }
}{
\infer[]{
\new x. P_1\cpar \tcontextmn{i,j}{0}{S^{i,j}_k \colon k \in I^i_j }
}{
\mathopen{\new x.}\left( P_1\cpar \tcontextmn{i,j}{0}{S^{i,j}_k \colon k \in I^i_j } \right)
}}
\qquad\qquad
\infer[]{
\wen x. P_1\cpar \new x. \tcontextmn{i,j}{0}{S^{i,j}_k \colon k \in I^i_j }
}{
\mathopen{\new x.} \left( P_1\cpar \tcontextmn{i,j}{0}{S^{i,j}_k \colon k \in I^i_j } \right)
}
\]
Similarly, consider whether judgement $\vdash \wen x. P_2 \cpar \hat{T}^i_j$ holds.
Either we have 
\[
\hat{T}^i_j = \tcontextmn{i,j}{1}{T^{i,j}_k \colon k \in J^i_j } \mbox{ and } \nfv{x}{\hat{T}^i_j}; 
\]
or we have $\hat{T}^i_j = \new x. \tcontextmn{i,j}{1}{T^{i,j}_k \colon k \in J^i_j }$. In each case, one of the following derivations holds, respectively.
\[
\infer[]{
\wen x. P_2\cpar \tcontextmn{i,j}{1}{T^{i,j}_k \colon k \in J^i_j }
}{
\infer[]{
\new x. P_2\cpar \tcontextmn{i,j}{1}{T^{i,j}_k \colon k \in J^i_j }
}{
\mathopen{\new x.} \left( P_2\cpar \tcontextmn{i,j}{1}{T^{i,j}_k \colon k \in J^i_j } \right)
}}
\qquad\qquad
\infer[]{
\wen x. P_2\cpar 
\hat{T}^i_j
}{
\mathopen{\new x.} \left( P_2\cpar \tcontextmn{i,j}{1}{T^{i,j}_k \colon k \in J^i_j } \right)
}
\]
Thereby, by applying one of the above cases for each $i$ and $j$, 
the following two proofs exist.
\[
\infer[]{
\left( P_0 \cseq \wen x. P_1\right) \cpar
\left(
    U^0_i \cseq
    \tcontextmn{i}{0}{
         \hat{S}^i_j 
         \colon j \in K^i
    }
\right)
}{
\infer[]{
\left( P_0\cpar U^0_i \right) \cseq
\left(
    \wen x. P_1 \cpar
    \tcontextmn{i}{0}{
             \hat{S}^i_j 
      \colon j \in K^i
    }
\right)
}{
\infer[]{
    \wen x. P_1 \cpar
    \tcontextmn{i}{0}{
             \hat{S}^i_j 
      \colon j \in K^i
    }
}{
\infer[]{
    \tcontextmn{i}{0}{
         \wen x. P_1 \cpar
         \hat{S}^i_j 
      \colon j \in K^i
    }
}{
\infer[]{
    \tcontextmn{i}{0}{
         \mathopen{\new x.} \left( P_1\cpar \tcontextmn{i,j}{0}{S^{i,j}_k \colon k \in I^i_j } \right)
      \colon j \in K^i
    }
}{
\infer[]{
    \tcontextmn{i}{0}{
         \mathopen{\new x.} \tcontextmn{i,j}{0}{ P_1\cpar S^{i,j}_k  \colon k \in I^i_j }
      \colon j \in K^i
    }
}{
\infer[]{
    \tcontextmn{i}{0}{
         \new x. \tcontextmn{i,j}{0}{ \cunit \colon k \in I^i_j }
      \colon j \in K^i
    }
}{
\cunit
}}}}}}}
\qquad
\infer[]{
\left( \wen x. P_2  \cseq P_3 \right) \cpar
\left(
\tcontextmn{i}{1}{
    \hat{T}^i_j \cseq 
    W^i_j \colon j \in L^i
}\right)
}{
\infer[]{
\tcontextmn{i}{1}{
    \left( \wen x. P_2  \cseq P_3 \right) \cpar
    \left( \hat{T}^i_j \cseq W^i_j \right)
    \colon j \in L^i
}
}{
\infer[]{
\tcontextmn{i}{1}{
    \left( \wen x. P_2  \cpar \hat{T}^i_j\right) \cseq
    \left( P_3 \cpar W^i_j \right)
    \colon j \in L^i
}
}{
\infer[]{
\tcontextmn{i}{1}{
    \wen x. P_2  \cpar \hat{T}^i_j
    \colon j \in L^i
}
}{
\infer[]{
\tcontextmn{i}{1}{
    \mathopen{\new x.} \left( P_2\cpar \tcontextmn{i,j}{1}{T^{i,j}_k \colon k \in J^i_j } \right)
    \colon j \in L^i
}
}{
\infer[]{
\tcontextmn{i}{1}{
    \new x. \tcontextmn{i,j}{1}{ P_2\cpar T^{i,j}_k  \colon k \in J^i_j }
    \colon j \in L^i
}
}{
\infer[]{
\tcontextmn{i}{1}{
    \new x. \tcontextmn{i,j}{1}{ \cunit \colon k \in J^i_j }
    \colon j \in L^i
}
}{
\cunit
}}}}}}}
\]
Furthermore, by Lemma~\ref{lemma:bound}, 
\[
\size{ 
    U^0_i \cseq
    \tcontextmn{i}{0}{
         \hat{S}^i_j 
         \colon j \in K^i
    }
} \preceq \size{ Q }
\mbox{ and }
\size{ 
 \tcontextmn{i}{1}{
    \hat{T}^i_j \cseq 
    W^i_j \colon j \in L^i
} } \preceq \size{Q}.
\]
Hence, sizes
\[
\size{
\left( P_0 \cseq \wen x. P_1\right) \cpar
\left(
    U^0_i \cseq
    \tcontextmn{i}{0}{
         \hat{S}^i_j 
         \colon j \in K^i
    }
\right)
}
\mbox{ and }
\size {
\left( \wen x. P_2  \cseq P_3 \right) \cpar
\left(
\tcontextmn{i}{1}{
    \hat{T}^i_j \cseq 
    W^i_j \colon j \in L^i
}\right)
}
\]
are strictly bounded above by 
$\size{ \left(P_0\cseq \wen x. P_1\cseq \wen x. P_2  \cseq P_3 \right) \cpar Q }$, as required.
Cases for \textit{left wen} and \textit{right wen} rules are similar.

%% file: equivariance.tex
\item \textbf{Commutative cases induced by equivariance.}
%
There are certain commutative cases induced by the \textit{equivariance} rule for nominal quantifiers.
These are the cases that force the rules \textit{all name}, \textit{with name}, \textit{left name} and \textit{right name} to be included.
Notice also that \textit{equivariance} for \textit{new} is required when handling the case induced by \textit{equivariance} for \textit{wen}; hence \textit{equivariance} for both nominal quantifiers must be explicit structural rules rather than properties derived from each other.

\begin{enumerate}[label*=\textbf{.\arabic*}]
\item Consider the commutative case for \textit{wen} where the bottommost rule of a proof is an instance of the \textit{close} rule of following form
\[
\infer[\mbox{,where $\vdash \mathopen{\new y.}\left( \wen x. P \cpar Q \right) \cpar R$, $\nfv{y}{R}$ and $\nfv{x}{R}$.} ]{
\wen x. \wen y. P \cpar \new y. Q \cpar R
}{
\mathopen{\new y.}\left( \wen x. P \cpar Q \right) \cpar R
}
\]
Notice that $\wen x$ is the principal connective but the \textit{close} rule is applied to $\wen y$ behind the principal connective.
Thus we desire some formula $R'$ such that $
\vcenter{
\infer[]{
\new y. Q \cpar R
}{R'}}$
and $\nfv{x}{R'}$ and either $\vdash \wen y. P \cpar R'$ or there exists $Q'$ such that $R' = \new x.Q'$ and $\vdash \wen y. P \cpar Q'$, and the size of $ \wen y. P \cpar R'$ is strictly smaller than $\wen x. \wen y. P \cpar \new y. Q \cpar R$.
By the induction hypothesis, there exist $S$ and $T$ such that $\nfv{y}{S}$ and $\vdash \wen x. P \cpar Q \cpar T$ and either $S = T$ or $S = \wen y. T$ and the derivation
$
\vcenter{
\infer{R}{S}}
$ holds.
Furthermore the size of the proof of $\wen x. P \cpar Q \cpar T$ is bounded above by the size of the proof of $\mathopen{\new y.}\left( \wen x.P \cpar Q \right) \cpar R$; hence strictly bounded by the size of the proof of $\wen x. \wen y. P \cpar \new y.Q \cpar R$.
Hence, by induction, there exist $U$ and $V$ such that $\vdash P \cpar V$ and $\nfv{x}{U}$ and either $U = V$ or $U = \new x. V$ the derivation 
$
\vcenter{
\infer[]{
Q \cpar T
}{
U
}}
$ holds.
Observe that if $S = T$, then 
$
\vcenter{
\infer[]{
\new y. Q \cpar S
}{
\mathopen{\new y.} \left( Q \cpar T \right)
}}$, since $\nfv{y}{S}$.
If $S = \wen y. T$ then 
$
\vcenter{
\infer[]{
\new y. Q \cpar \wen y. T
}{
\mathopen{\new y.} \left( Q \cpar T \right)
}}$.
Thereby the following derivation can be constructed, where if $U = V$ then $W = \new y. V$ and if $U = \new x. V$ then $W = \new x. \new y. V$, and also the premise is equivalent to $W$ by \textit{equivariance} for \textit{new}:
$
\vcenter{
\infer[]{
\new y. Q \cpar R
}{
\infer[]{
\new y. Q \cpar S
}{
\infer[]{
\mathopen{\new y.} \left( Q \cpar T \right)
}{
\new y. U
}}}}
$.
Furthermore, the following proof can be constructed 
$
\vcenter{
\infer[]{
\wen y. P \cpar \new y. V
}{
\infer[]{
\mathopen{\new y.} \left(P \cpar V\right)
}{
\infer[]{
 \new y. \cunit
}{
\cunit
}}}}$
and, by Lemma~\ref{lemma:bound}, $\size{\new y. V} \preceq \size{\new y. Q \cpar R}$ hence 
$\size{ \wen y. P \cpar \new y. V } \prec \size{ \wen x. \wen y. P \cpar \new y. Q \cpar R }$, as required.

\item Consider a commutative case for \textit{new} induced by \textit{equivariance} for \textit{new}, where the bottommost rule is an instance of \textit{extrude new} of the form
\[
\infer[, \mbox{where $\nfv{y}{Q}$ and $\vdash \mathopen{\new y.} \left( \new x. P \cpar Q\right) \cpar R$.}]{
\new x. \new y. P \cpar Q \cpar R
}{
\mathopen{\new y.} \left( \new x. P \cpar Q\right) \cpar R
}
\]
By the induction hypothesis, there exist $S$ and $T$ such that $\nfv{y}{S}$ and $\vdash \new x. P \cpar Q \cpar T$ and either $S = T$ or $S = \wen y. T$, where
$\vcenter{\infer[]{R}{S}}$.
Furthermore, the size of the proof of $\new x. P \cpar Q \cpar T$ is bound above by the size of the proof of $\mathopen{\new y.} \left( \new x. P \cpar Q\right) \cpar R$, hence strictly bound above by the size of the proof of $\new x. \new y. P \cpar Q \cpar R$.
Hence, by induction again, there exist $U$ and $V$ such that $\nfv{x}{U}$ and $\vdash P \cpar V$ and either $U = V$ or $U = \wen x. V$, and also $\infer[]{Q \cpar T}{U}$.
Now define $\hat{W}$ and $W$ as follows.
If $S = T$ then let $\hat{W} = V$.
If $S = \wen y. T$ then let $\hat{W} = \wen y. V$.
If $U = V$ then let $W = \hat{W}$.
If $U = \wen x. V$ then let $W = \wen x. \hat{W}$.
Now observe if $S = T$ then
$
\vcenter{
\infer[]{
Q \cpar R
}{
\infer[]{
Q \cpar T
}{
U
}}}$ and $U = W$.
For $S = \wen y. T$ observe
$
\vcenter{
\infer[]{
Q \cpar R
}{
\infer[]{
Q \cpar \wen y. T
}{
\infer[]{
\mathopen{\wen y.} \left( Q \cpar T \right)
}{
\mathopen{\wen y.} U
}}}}
$, since $\nfv{y}{Q}$,
and if $U = V$ then $\mathopen{\wen y.} U = \hat{W}$,
while if $U = \wen x. V$ then $\wen y.U \equiv \wen x.\hat{W}$, by \textit{equivariance} for \textit{wen}.
Hence in all cases $\vcenter{\infer[]{Q \cpar R}{W}}$
and, since $\nfv{y}{Q}$ and $\nfv{y}{T}$, we can arrange that $\nfv{y}{W}$.
Now, for the cases where $\hat{W} = V$, we have $\nfv{y}{V}$,
and hence $\vcenter{\infer[]{\new y. P \cpar V}{\mathopen{\new y.}\left( P \cpar V \right)}}$.
Also if $\hat{W} = \wen y.V$, then $
\vcenter{
\infer[]{
\new y. P \cpar \wen y.V
}{
\mathopen{\new y.}\left( P \cpar V \right)
}}$. Hence in either case we can construct the proof $
\vcenter{
\infer[]{
\new y. P \cpar \hat{W}
}{
\infer[]{
\mathopen{\new y.}\left( P \cpar V \right)
}{
\infer[]{
 \new y.\cunit
}{ \cunit
}}}}$.
Furthermore, $\size{ \new y. P \cpar \hat{W} } \prec \size{ \new x. \new y. P \cpar Q \cpar R }$,
since by Lemma~\ref{lemma:bound} $\size{ \hat{W} } \preceq \size{ Q \cpar R }$.

\item Similar commutative cases for \textit{wen} and \textit{new} as principal formulae are induced by \textit{equivariance} where the bottommost rule in a proof is an instance of the \textit{close}, \textit{right wen} or \textit{suspend} rules.
In each case, the quantifier involved in the bottommost rule appears behind the principal connective and is propagated in front of the principal connective using \textit{equivariance}.
\end{enumerate}

\fullproof{
A similar commutative case for \textit{wen} is induced where the bottommost rule in a proof is an instance of the \textit{left wen} rule of the form 
$
\infer[]{
\wen x. \wen y. P \cpar Q \cpar R
}{
\mathopen{\wen y.} \left( \wen x. P \cpar Q \right) \cpar R
}
$
where $\nfv{y}{Q}$ and $\vdash \mathopen{\wen y.} \left( \wen x. P \cpar Q \right) \cpar R$.
By the induction hypothesis, there exist $S$ and $T$ such that $\nfv{y}{S}$ and $\vdash \wen x. P \cpar Q \cpar T$ and either $S = T$ or $S = \new y. T$ such that $
\infer[]{R}{S}$.
Furthermore, the size of the proof of $\wen x. P \cpar Q \cpar T$ is bounded by the size of the proof of $\mathopen{\wen y.} \left( \wen x. P \cpar Q \right) \cpar R$ hence strictly bounded by the size of the proof of $\wen x. \wen y. P \cpar Q \cpar R$.
Hence, by the induction hypothesis again, there exist $U$ and $V$ such that $\nfv{x}{U}$ and $\vdash P \cpar V$ and either $U = V$ or $U = \new x. V$, and derivation $
\infer[]{Q \cpar T}{U}$ holds.
If $S = \new y. T$ let $\hat{W} = \new y. V$ otherwise $\hat{W} = V$. If $U = \new x. V$ let $W = \new x. \hat{W}$, otherwise $W = \hat{W}$.
Thereby the following derivations hold:
\begin{itemize}
\item If $S = \new y. T$ and $U = \new x. V$ then
$
\infer[]{
Q \cpar S
}{
\infer[]{
\mathopen{\new y.} \left( Q \cpar T \right)
}{
 \new y. U \equiv W
}}$.

\item If $S = \new y. T$ and $U = V$ then $
\infer[]{Q \cpar S
}{
\infer[]{
\mathopen{\new y.} \left( Q \cpar T \right) 
}{
 \new y. U
}}$, where the premise is equivalent to $W$.

\item If $S = T$ then $Q \cpar S = Q \cpar T$
and $
\infer[]{
Q \cpar T
}{
 U \equiv W
}$.
\end{itemize}
Hence in any of the above cases, $
\infer[]{
Q \cpar R
}{
\infer[]{
 Q \cpar S
}{ W
}}$.
Now, if $\hat{W} = \new y. V$, then $
\infer[]{
\wen y. P \cpar \hat{W}
}{
 \new y. \left(P \cpar V\right)
}$;
and if $\hat{W} = V$ then $S = T$, hence $
\infer[]{
\wen y. P \cpar \hat{W}
}{
\infer[]{
 \new y. P \cpar \hat{W}
}{
 \mathopen{\new y.} \left( P \cpar V\right)
}}$, since $\nfv{y}{V}$.
Clearly $
\infer[]{
\mathopen{\new y.} \left( P \cpar V\right)
}{
\infer[]{
\new y. \cunit
}{
 \cunit
}}
$, hence $\vdash \wen y.P \cpar \hat{W}$ holds.  Furthermore, $\size{\wen y. P \cpar \hat{W}} \prec \size{\wen x. \wen y. P \cpar Q \cpar R}$, since by Lemma~\ref{lemma:bound} $\size{\hat{W}} \preceq \size{ Q \cpar R }$.

The third commutative case for \textit{wen} induced by \textit{equivariance} is where the bottommost rule is an instance of the \textit{suspend} rule of the form 
$
\infer[]{
\wen x. \wen y. P \cpar \wen y. Q \cpar R
}{
\mathopen{\wen y.} \left( \wen x. P \cpar Q \right) \cpar R
}
$, where $\vdash \mathopen{\wen y.} \left( \wen x. P \cpar Q \right) \cpar R$.
By the induction hypothesis, there exist $S$ and $T$ where $\nfv{y}{S}$ and $\vdash \wen x. P \cpar Q \cpar T$ and either $S = T$ or $S = \new y. T$ such that $\infer{R}{S}$.
Furthermore, the size of the proof of $\wen x. P \cpar Q \cpar T$ is bounded above by the size of the proof of $\mathopen{\wen y.} \left( \wen x. P \cpar Q \right) \cpar R$, hence strictly bounded above by the size of the proof of $\wen x. \wen y. P \cpar \wen y. Q \cpar R$.
Hence by the induction hypothesis again, there exist $U$ and $V$ where $\nfv{x}{U}$ and $\vdash P \cpar V$ and either $U = V$ or $U = \new x. V$, such that $\infer{Q \cpar T}{U}$.
If $U = V$ then let $W = \new y. V$, and if $U = \new x. V$ then  let $W = \new x. \new y. V$.
Also observe that whether we have $S = T$ or $S = \new y. T$, we have $
\infer[]{
\wen y. Q \cpar S
}{
\mathopen{\wen y.} \left( Q \cpar T \right)
}$.
Thereby the following derivation can be constructed:
$\infer[]{
\wen y. Q \cpar R
}{
\infer[]{
\mathopen{\wen y.} \left( Q \cpar T \right)
}{
\wen y. U \equiv W
}}$.
Furthermore, the following proof can be constructed $
\infer[]{
\wen y. P \cpar \new y. V
}{
\infer[]{
\mathopen{\new y.} \left( P \cpar V \right)
}{
\infer[]{
\new y. \cunit
}{
\cunit
}}}$;
and the size of the proof of $\size{ \wen y. P \cpar \new y. V } \prec \size{ \wen x. \wen y. P \cpar \wen y. Q \cpar R}$, since $\size{ \new y. V } \preceq \size{ \wen y. Q \cpar R }$ by Lemma~\ref{lemma:bound}.

Consider the first commutative case for \textit{new} induced by equivariance, where the bottommost rule is of the form
$
\infer[]{
\new x. \new y. P \cpar \wen y. Q \cpar R
}{
\mathopen{\new y.} \left( \new x. P \cpar Q\right) \cpar R
}
$,
where $\vdash \mathopen{\new y.} \left( \new x. P \cpar Q\right) \cpar R$.
By the induction hypothesis, there exist $S$ and $T$ such that $\nfv{y}{S}$ and $\vdash \wen x. P \cpar Q \cpar T$ and either $S = T$ or $S = \wen y .T$, where
$\infer[]{R}{S}$.
Furthermore, the size of the proof of $\wen x. P \cpar Q \cpar T$ is bound above by the size of the proof of $\mathopen{\new y.} \left( \new x. P \cpar Q\right) \cpar R$, hence strictly bound above by the size of the proof of $\new x. \new y. P \cpar \wen y. Q \cpar R$.
Hence, by induction again, there exist $U$ and $V$ such that $\nfv{x}{U}$ and $\vdash P \cpar V$ and either $U = V$ or $U = \wen x. V$, and also $\infer[]{Q \cpar T}{U}$.
If $U = V$ then let $W = \wen y. V$, and if $U = \wen x. V$ then let $W = \wen x. \wen y. V$. 
Also, regardless of whether $S = T$ or $S = \wen y. T$, $\infer[]{\wen y. Q \cpar S}{\mathopen{\wen y.} \left( Q \cpar S \right)}$. Hence derivation 
$
\infer[]{
\wen y. Q \cpar R
}{
\infer[]{
\wen y. Q \cpar S
}{
\infer[]{
\mathopen{\wen y.} \left( Q \cpar T \right)
}{
\wen y. U
}}}$ can be constructed, where the premise is equivalent to $W$.
Furthermore, $
\infer[]{
\new y. P \cpar \wen y. V
}{
\infer[]{
\mathopen{\new y.} \left( P \cpar V \right)
}{
\infer[]{
\new y. \cunit 
}{ \cunit
}}}$,
and $\size{ \new y. P \cpar \wen y. V } \prec \size{ \new x. \new y. P \cpar \wen y. Q \cpar R }$,
since by Lemma~\ref{lemma:bound} $\size{ \wen y. V } \preceq \size{ \wen y. Q \cpar R }$.

}

%% file: commutative.tex
\item \textbf{Regular commutative cases.} 
As in every splitting lemma, there are numerous \textit{commutative} cases where the bottommost rule in a proof does not directly involve the principal connective.
For each principal formula handled by this splitting lemma (\textit{new}, \textit{wen}, \textit{with}, \textit{seq} and \textit{times}) there are commutative cases induced by \textit{new}, \textit{wen}, \textit{all}, \textit{with} and \textit{times} and also two commutative cases induced by \textit{seq}.
Thus there are 35 similar commutative cases to check, that all follow a pattern, hence only a representative selection of four cases are presented that make special use of $\alpha$-conversion and the rules \textit{new wen},
\textit{all name}, \textit{with name}, \textit{left name} and \textit{right name}.
Further, representative cases appear in the proof for existential quantifiers.

\begin{enumerate}[label*=\textbf{.\arabic*}]

\item Consider the commutative case where the principal formula is $\new x. P$ and the bottommost rule is an instance of \textit{extrude new} but applied to a distinct \textit{new} quantifier $\new y. Q$, as in the following rule instance 
\[
\infer[\mbox{, where $\nfv{y}{\new x. P \cpar R}$.}]{
\new x. P \cpar \new y. Q \cpar R \cpar S
}{
\mathopen{\new y.} \left( \new x. P \cpar Q \cpar R\right) \cpar S
}
\] 
Also assume, by $\alpha$-conversion, that $x \not= y$.
By induction, there exist $T$ and $U$ such that $\vdash \new x. P \cpar Q \cpar R \cpar U$, $\nfv{y}{T}$ and either $T = U$ or $T = \wen y. U$, and also
$\vcenter{\infer{S}{T}}$.
Furthermore, the size of the proof of $\new x. P \cpar Q \cpar R \cpar U$ is bounded above by the size of the proof of $\mathopen{\new y.} \left( \new x. P \cpar Q \cpar R\right) \cpar S$ and hence strictly bounded above by the size of the proof of $\new x. P \cpar \new y. Q \cpar R \cpar S$, enabling the induction hypothesis.
Hence, by the induction hypothesis, there exist formulae $V$ and $\hat{V}$ such that $\vdash P \cpar \hat{V}$ and $\nfv{x}{V}$ and either $V = \hat{V}$ or $V = \wen x. \hat{V}$, and also $
\vcenter{
\infer[]{
Q \cpar R \cpar U
}{V}}$.
Define $W$ such that if $V = \hat{V}$ then $W = \new y. \hat{V}$ and if $V = \wen x. \hat{V}$ then $W = \wen x. \new y. \hat{V}$.
Hence if $V = \wen x. \hat{V}$ then $
\vcenter{
\infer[]{
\new y. V
}{
 \wen x. \new y. \hat{V}
}}$ by applying the \textit{new wen} rule, where the premise equals $W$.
If $V = \hat{V}$ then $\new y. V = W$. In both cases, $\nfv{x}{W}$.
Now observe that either $T = U$ and $\nfv{y}{U}$, hence the derivation $(a)$ below holds;
or $T = \wen y. U$, hence the derivation $(b)$ below holds. Given these, the derivation $(c)$ can be constructed:
\[
\begin{array}{cccc}
\infer[]{
\new y. Q \cpar R \cpar T
}{
\mathopen{\new y.} \left(Q \cpar R \cpar U\right)
}
&
\infer[]{
\new y. Q \cpar R \cpar \wen y. U
}{
\infer[]{
 \mathopen{\new y.} \left(Q \cpar R\right) \cpar \wen y. U
}{
 \mathopen{\new y.} \left(Q \cpar R \cpar U\right)
}}
& 
\infer[]{
\new y. Q \cpar R \cpar S
}{
\infer[]{
\new y. Q \cpar R \cpar T
}{
\infer[]{
\mathopen{\new y.} \left( Q \cpar R \cpar U \right)
}{
\infer[]{
\new y. V
}{
W
}}}} &
\infer[]{
P \cpar \new y. \hat{V}
}{
\infer[]{
 \mathopen{\new y.} \left( P \cpar \hat{V} \right) 
}{
\infer[]{
 \new y. \cunit
}{
 \cunit
}}}\\
(a) & (b) & (c) & (d)
\end{array}
\]
Since $\nfv{y}{\new x. P \cpar R}$ and $x \not= y$, we have $\nfv{y}{P}$;
thereby the proof $(d)$ above can be constructed.
Furthermore, $\size{ P \cpar \new y. \hat{V} } \prec \size{ \wen x. P \cpar \new y. Q \cpar R \cpar S }$ since by Lemma~\ref{lemma:bound} $\size{ \new y. \hat{V} } \preceq \size{\new y. Q \cpar R \cpar S}$
and the \textit{wen count} strictly decreases.

\item Consider the commutative case for principal formula $\wen x. T$ where the bottommost rule is \textit{external}: 
\[
\infer[]{
\wen x. T \cpar \left( U \wwith V \right) \cpar W \cpar P
}{
\left(\left(\wen x. T \cpar U \cpar W\right) \wwith
\left(\wen x. T \cpar V \cpar W\right)\right)
\cpar P
}
\]
where $\vdash \left(\left(\wen x. T \cpar U \cpar W\right) \wwith \left(\wen x. T \cpar V \cpar W\right)\right) \cpar P$ holds.
By the induction hypothesis, we have that both $\vdash \wen x. T \cpar U \cpar W \cpar P$ and $\vdash \wen x. T \cpar V \cpar W \cpar P$ hold;
and furthermore the multiset inequalities 
\[
\begin{array}{rcl}
\occ{\wen x. T \cpar U \cpar W \cpar P} & \mstrict & \occ{ \wen x. T \cpar \left( U \wwith V \right) \cpar W \cpar P } \mbox{ and }\\
\occ{\wen x. T \cpar V \cpar W \cpar P} & \mstrict & \occ{ \wen x. T \cpar \left( U \wwith V \right) \cpar W \cpar P }
\end{array}
\]
hold.
Hence, by the induction hypothesis, there exist $Q$ and $\hat{Q}$ such that $\vdash T \cpar \hat{Q}$, $\nfv{x}{Q}$ and either $Q = \hat{Q}$ or $Q = \new x. \hat{Q}$.
Also, by the induction hypothesis, there exist $R$ and $\hat{R}$ such that
$\vdash T \cpar \hat{R}$, $\nfv{x}{R}$ and either $R = \hat{R}$ or $R = \new x. \hat{R}$.
Furthermore the two derivations 
$
\vcenter{
\infer[]{
U \cpar W \cpar P
}{ Q }}
$
and
$
\vcenter{
\infer[]{
V \cpar W \cpar P
}{ R }
}
$ hold.
Now define $S$ such that if $Q = \hat{Q}$ and $R = \hat{R}$ then $S = \hat{Q} \wwith \hat{R}$, and 
$S = \mathopen{\wen x.} \left( \hat{Q} \wwith \hat{R} \right)$ otherwise, observing that in either case $\nfv{x}{S}$.
In the case $Q = \wen x. \hat{Q}$ and $R = \wen x. \hat{R}$, by the \textit{with name} rule,
$
\vcenter{
\infer[]{
\wen x. \hat{Q} \wwith \wen x. \hat{R}
}{
\mathopen{\wen x.} \left( \hat{Q} \wwith \hat{R} \right)
}}$.
In the case $Q = \wen x. \hat{Q}$ and $R = \hat{R}$, by the \textit{left name} rule,
$
\vcenter{
\infer[]{
\wen x. \hat{Q} \wwith \hat{R}
}{
\mathopen{\wen x.} \left( \hat{Q} \wwith \hat{R} \right)
}}$.
In the case that $Q = \hat{Q}$ and $R = \wen x. \hat{R}$, by the \textit{right name} rule,
$
\vcenter{
\infer[]{
\hat{Q} \wwith \wen x. \hat{R}
}{
\mathopen{\wen x.} \left( \hat{Q} \wwith \hat{R} \right)
}}$.
Thereby the following derivation and proof can be constructed:
\[
\infer[]{
\left( U \wwith V \right) \cpar W \cpar P
}{
\infer[]{
\left(U \cpar W \cpar P\right) \wwith \left(V \cpar W \cpar P\right) 
}{
\infer[]{
Q \wwith R
}{
S
}}}
\qquad
\infer[.]{
T \cpar \left(\hat{Q} \wwith \hat{R}\right)
}{
\infer[]{
\left(T \cpar \hat{Q}\right) \wwith \left( T \cpar \hat{R}\right)
}{
\infer[]{
\cunit \wwith \cunit
}{
\cunit
}}}
\]
Furthermore, by Lemma~\ref{lemma:bound}, $\size{ S } \preceq \size{ \left( U \wwith V \right) \cpar W \cpar P }$; and,
since the \textit{wen} count strictly decreases,
$\size{ T \cpar \hat{Q} \wwith \hat{R} } \prec \size{\wen x. T \cpar \left( U \wwith V \right) \cpar W \cpar P}$.

\item Consider the commutative case where the principal formula is $\wen x. T$ and the bottommost rule is an instance of the
\textit{extrude1} rule of the form
\[
\infer[]{
\wen x. T \cpar \forall y. U \cpar V \cpar W
}{
\mathopen{\forall y.}\left( \wen x. T \cpar U \cpar V \right) \cpar W
}\]
 assuming $\nfv{y}{(\wen x. T \cpar V)}$
and $\vdash \mathopen{\forall y.}\left( \wen x. T \cpar U \cpar V \right) \cpar W$ holds.
By Lemma~\ref{lemma:universal}, for every variable $z$, $\vdash \left( \wen x. T \cpar U \cpar V \right)\sub{y}{z} \cpar W$ holds.
Furthermore, since $\nfv{y}{(\wen x. T \cpar V)}$, we have equivalence $\left( \wen x. T \cpar U \cpar V \right)\sub{y}{z} \cpar W \equiv \wen x. T \cpar U\sub{y}{z} \cpar V \cpar W$.
The strict multiset inequality $
\occ{ \wen x. T \cpar U\sub{y}{z} \cpar V \cpar W } \mstrict \occ{ \wen x. T \cpar \forall y. U \cpar V \cpar W }
$ holds.
Hence, by the induction hypothesis, for every variable $z$, there exist formulae $P^z$ and $Q^z$ such that $\vdash T \cpar Q^z$ and $\nfv{x}{P^z}$ and either $P^z = Q^z$ or $P^z = \new x. Q^z$, and also
$
\vcenter{
\infer[]{
U\sub{y}{z} \cpar V \cpar W
}{
P^z
}}
$.
Define $W^z$ such that if $P^z = Q^z$ then $W^z = \forall z. Q^z$ and if $P^z = \new x. Q^z$ then $W^z = \new x. \forall z. Q^z$. Hence if $P^z = \new x. Q^z$ then, since $\forall$ permutes with any quantifier using the \textit{all name} rule, $
\vcenter{
\infer[]{
\forall z. \new x. Q^z
}{
 \new x. \forall z. Q^z
}}$.
Hence, for a fresh $z$ such that $\nfv{z}{(\forall y. U \cpar V \cpar W)}$ and $\nfv{z}{T}$, 
the following derivations can be constructed:
\[
\infer[]{
\forall y. U \cpar V \cpar W
}{
\infer[]{
\mathopen{\forall z.} \left( U\sub{y}{z} \cpar V \cpar W \right) 
}{
\infer[]{
\forall z. P^z
}{
W^z
}}}
\qquad
\infer[]{
T \cpar \forall z. Q^z
}{
\infer[]{
 \mathopen{\forall z.} \left( T \cpar  Q^z\right) 
}{
\infer[]{
 \mathopen{\forall z.} \cunit 
}{
 \cunit
}}}
\]
Furthermore, $\size{ W^z } \preceq \size{ \forall y. U \cpar V \cpar W }$ by Lemma~\ref{lemma:bound}; hence 
\[
\size{ T \cpar \forall z. Q^z } \prec \size{ \wen x. T \cpar \forall y. U \cpar V \cpar W }
\] 
since the wen count strictly decreases.

\item Consider the commutative case where the principal connective is $\textit{wen}$ and the bottommost rule is an instance of the extrude new rule of the form 
\[
\infer[,]{
\wen x. P \cpar \new y. Q \cpar R \cpar S
}{
\mathopen{\new y.} \left( \wen x. P \cpar Q \cpar R\right) \cpar S
}
\]
 where $\nfv{y}{\wen x. P \cpar R}$ and also $x \not= y$, where the second condition can be achieved by $\alpha$-conversion.
By the induction hypothesis, there exist $T$ and $U$ such that $\vdash \wen x. P \cpar Q \cpar R \cpar U$, $\nfv{y}{T}$ and either $T = U$ or $T = \wen y. U$, and also
$\vcenter{\infer{S}{T}}$.
Furthermore, the size of the proof of $\wen x. P \cpar Q \cpar R \cpar U$ is bounded above by the size of the proof of $\mathopen{\new y.} \left( \wen x. P \cpar Q \cpar R\right) \cpar S$ and hence strictly bounded above by the size of the proof of $\wen x. P \cpar \new y. Q \cpar R \cpar S$, enabling the induction hypothesis.
Hence, by the induction hypothesis, there exist formulae $V$ and $\hat{V}$ such that $\vdash P \cpar \hat{V}$ and $\nfv{x}{V}$ and either $V = \hat{V}$ or $V = \new x. \hat{V}$, and also $
\vcenter{
\infer[]{
Q \cpar R \cpar U}{
 V
}}$.
Define $W$ such that if $V = \hat{V}$ then $W = \new y. \hat{V}$ and if $V = \new x. \hat{V}$ then $W = \new x. \new y. \hat{V}$.
Now observe that either we have that $T = U$ and $\nfv{y}{U}$ and hence
the derivation $(a)$ below left holds;  or we have that $T = \wen y. U$ and hence the derivation $(b)$ belw  holds. 
Hence,
by applying one of these cases, we have the derivation $(c)$ below, where the premise is equivalent to $W$.
\[
\begin{array}{cccc}
\infer[]{
\new y. Q \cpar R \cpar T
}{
 \mathopen{\new y.} \left(Q \cpar R \cpar U\right)
}
&
\infer[]{
\new y. Q \cpar R \cpar \wen y. U 
}{
\infer[]{
\mathopen{\new y.} \left(Q \cpar R\right) \cpar \wen y. U
}{
 \mathopen{\new y.} \left(Q \cpar R \cpar U\right)
}}
&
\infer[]{
\new y. Q \cpar R \cpar S
}{
\infer[]{
\new y. Q \cpar R \cpar T
}{
\infer[]{
\mathopen{\new y.} \left( Q \cpar R \cpar U \right)
}{
\new y. V
}}} &
\infer[.]{
P \cpar \new y. \hat{V}
}{
\infer[]{
 \mathopen{\new y.} \left( P \cpar \hat{V} \right)
}{
\infer[]{
 \new y. \cunit
}{
 \cunit
}}}
\\
(a)  & (b) & (c) & (d)
\end{array}
\]
Since $\nfv{y}{\wen x. P}$ and $x \not= y$, we have $\nfv{y}{P}$;
thereby the proof $(d)$ above can be constructed.
Furthermore, $\size{ P \cpar \new y.\hat{V} } \prec \size{ \wen x. P \cpar \new y. Q \cpar R \cpar S }$ since by Lemma~\ref{lemma:bound} 
\[
\size{ \new y. \hat{V} } \preceq \size{\new y. Q \cpar R \cpar S}
\]
and the \textit{wen count} strictly decreases.
\end{enumerate}

\fullproof{
Consider the commutative case when the \textit{wen} quantifier commutes with another \textit{wen} quantifier at the root of the principal formula.
In this case, the bottommost rule of a proof is of the form $\wen x. P \cpar \wen y. Q \cpar R \cpar \infer[]{
S}{
 \mathopen{\wen y.} \left( \wen x. P \cpar Q \cpar R\right) \cpar S
}$, where $\nfv{y}{\wen x. P \cpar R}$ and $x \not= y$, where the second condition can be achieved by $\alpha$-conversion.
By the induction hypothesis, there exist $T$ and $U$ such that $\vdash \wen x. P \cpar Q \cpar R \cpar U$, $\nfv{y}{T}$ and either $T = U$ or $T = \new y. U$, and also
$\infer{S}{T}$.
Furthermore, the size of the proof of $\wen x. P \cpar Q \cpar R \cpar U$ is bounded above by the size of the proof of $\mathopen{\wen y.} \left( \wen x. P \cpar Q \cpar R\right) \cpar S$ and hence strictly bounded above by the size of the proof of $\wen x. P \cpar \wen y. Q \cpar R \cpar S$, enabling the induction hypothesis.
Hence, by the induction hypothesis, there exist formulae $V$ and $\hat{V}$ such that $\vdash P \cpar \hat{V}$ and $\nfv{x}{V}$ and either $V = \hat{V}$ or $V = \new x. \hat{V}$, and also $
\infer[]{
Q \cpar R \cpar U 
}{ V }$.

Define $W$ such that if $V = \hat{V}$ then $W = \new y. \hat{V}$ and if $V = \new x. \hat{V}$ then $W = \new x. \new y. \hat{V}$. In either case, $\nfv{x}{W}$.
Now observe that either $T = U$ and $\nfv{y}{U}$ hence
$\infer[]{
\wen y. Q \cpar R \cpar T
}{
\infer[]{
 \mathopen{\wen y.} \left(Q \cpar R \cpar U\right) 
}{
 \mathopen{\new y.} \left(Q \cpar R \cpar U\right)
}}$;
or $T = \new y. U$ hence $
\infer[]{
\wen y. Q \cpar R \cpar \new y. U
}{
\infer[]{
 \mathopen{\wen y.} \left(Q \cpar R\right) \cpar \mathopen{\new y.} U 
}{
\mathopen{\new y.} \left(Q \cpar R \cpar U\right)
}}$.
Hence the following derivation can be constructed, by applying one of these cases:
$
\infer[]{
\wen y. Q \cpar R \cpar S
}{
\infer[]{
 \wen y. Q \cpar R \cpar T
}{
\infer[]{
\mathopen{\new y.} \left( Q \cpar R \cpar U \right)
}{
\new y. V
}}}
$, where the premise is equivalent to $W$.
Since $\nfv{y}{\wen x. P}$ and $x \not= y$ we have $\nfv{y}{P}$,
hence the following proof can be constructed:
$
\infer[]{
P \cpar \new y. \hat{V}
}{
\infer[]{
\mathopen{\new y.}\left( P \cpar \hat{V} \right)
}{
\infer[]{
\new y. \cunit 
}{
\cunit
}}}$.
Furthermore, $\size{ P \cpar \new y. \hat{V} } \prec \size{ \wen x. P \cpar \new y. Q \cpar R \cpar S }$ since by Lemma~\ref{lemma:bound} $\size{ \new y. \hat{V} } \preceq \size{\new y. Q \cpar R \cpar S}$
and the \textit{wen count} strictly decreases.
\smallskip

Consider the commutative case when the \textit{new} quantifier commutes with $P_0 \tensor P_1$ as the principal formula.
In this case the bottommost rule of a proof is of the form $
\infer[]{
\left( P_0 \tensor P_1 \right) \cpar \new y. Q \cpar R \cpar S
}{
\mathopen{\new y.} \left( \left( P_0 \tensor P_1 \right) \cpar Q \cpar R\right) \cpar S
}$,
assuming that $\nfv{y}{\left( P_0 \tensor P_1 \right) \cpar R}$.
By induction, there exist $T$ and $U$ such that $\vdash \left( P_0 \tensor P_1 \right) \cpar Q \cpar R \cpar U$, $\nfv{y}{T}$ and either $T = U$ or $T = \wen y. U$, and also
$\infer{S}{T}$.
Furthermore, the size of the proof of $\left( P_0 \tensor P_1 \right) \cpar Q \cpar R \cpar U$ is bounded above by the size of the proof of $\mathopen{\new y.} \left( \left( P_0 \tensor P_1 \right) \cpar Q \cpar R\right) \cpar S$ and hence strictly bounded above by the size of the proof of $\left( P_0 \tensor P_1 \right) \cpar \new y. Q \cpar R \cpar S$, enabling the induction hypothesis.
Hence, by the induction hypothesis, there exist formulae $V_i$ and $W_i$ such that $\vdash P_0 \cpar V_i$ and $\vdash P_1 \cpar W_i$, for $1 \leq i \leq n$, and also $n$-ary killing context $\tcontext{}$ such that
$
\infer[]{
Q \cpar R \cpar U 
}{
\tcontext{ V_i \cpar W_i \colon 1 \leq i \leq n }
}$.
Furthermore, the size of the proofs of $P_0 \cpar V_i$ and $P_1 \cpar W_i$ are bounded above by the size of the proof of $\left( P_0 \tensor P_1 \right) \cpar \new y. Q \cpar R \cpar S$.
Now observe that either $T = U$ and $\nfv{y}{U}$ hence
$\infer[]{
\new y. Q \cpar R \cpar T
}{
 \mathopen{\new y.} \left(Q \cpar R \cpar U\right)
}$;
or $T = \wen y. U$ hence $
\infer[]{
\new y. Q \cpar R \cpar \wen y. U
}{
\infer[]{
 \mathopen{\new y.} \left(Q \cpar R\right) \cpar \wen y. U
}{
 \mathopen{\new y.} \left(Q \cpar R \cpar U\right)
}}$.
Hence the following derivation can be constructed, by the above observations:
$
\infer[]{
\new y. Q \cpar R \cpar S
}{
\infer[]{
\new y. Q \cpar R \cpar T
}{
\infer[]{
\mathopen{\new y.} \left( Q \cpar R \cpar U \right)
}{
\new y. \tcontext{ V_i \cpar W_i \colon 1 \leq i \leq n }
}}}
$.
Observe that $\new y. \tcontext{ }$ is a $n$-ary killing context as required.
\smallskip

Consider the commutative case when the \textit{wen} quantifier commutes with $P_0 \tensor P_1$ as the principal formula.
In this case the bottommost rule of a proof is of the form $
\infer[]{
\left( P_0 \tensor P_1 \right) \cpar \wen y. Q \cpar R \cpar S
}{
 \mathopen{\wen y.} \left( \left( P_0 \tensor P_1 \right) \cpar Q \cpar R\right) \cpar S
}$, where $\nfv{y}{\left( P_0 \tensor P_1 \right) \cpar R}$.
By induction, there exist $T$ and $U$ such that $\vdash \left( P_0 \tensor P_1 \right) \cpar Q \cpar R \cpar U$, $\nfv{y}{T}$ and either $T = U$ or $T = \new y. U$, and also
$\infer{S}{T}$.
Furthermore, the size of the proof of $\left( P_0 \tensor P_1 \right) \cpar Q \cpar R \cpar U$ is bounded above by the size of the proof of $\mathopen{\wen y.} \left( \left( P_0 \tensor P_1 \right) \cpar Q \cpar R\right) \cpar S$ and hence strictly bounded above by the size of the proof of $\left( P_0 \tensor P_1 \right) \cpar \wen y. Q \cpar R \cpar S$, enabling the induction hypothesis.
Hence, by the induction hypothesis, there exist formulae $V_i$ and $W_i$ such that $\vdash P_0 \cpar V_i$ and $\vdash P_1 \cpar W_i$, for $1 \leq i \leq n$, and also $n$-ary killing context $\tcontext{}$ such that
$\infer[]{
Q \cpar R \cpar U
}{
\tcontext{ V_i \cpar W_i \colon 1 \leq i \leq n }
}$.
Furthermore, the size of the proofs of $P_0 \cpar V_i$ and $P_1 \cpar W_i$ are bounded above by the size of the proof of $\left( P_0 \tensor P_1 \right) \cpar \wen y. Q \cpar R \cpar S$.
Now observe that either $T = U$ and $\nfv{y}{U}$ hence
$\infer[]{
\wen y. Q \cpar R \cpar T
}{
\infer[]{
\mathopen{\wen y.} \left(Q \cpar R \cpar U\right)
}{
 \mathopen{\new y.} \left(Q \cpar R \cpar U\right)
}}$;
or $T = \new y. U$ hence $
\infer[]{
\wen y. Q \cpar R \cpar \new y. U
}{
\infer[]{
\mathopen{\wen y.} \left(Q \cpar R\right) \cpar \new y. U
}{
 \mathopen{\new y.} \left(Q \cpar R \cpar U\right)
}}$.
Hence the following derivation can be constructed, by the above observations:
$
\infer[]{
\wen y. Q \cpar R \cpar S
}{
\infer[]{
\wen y. Q \cpar R \cpar T
}{
\infer[]{
\mathopen{\new y.} \left( Q \cpar R \cpar U \right)
}{
\mathopen{\new y.} \tcontext{ V_i \cpar W_i \colon 1 \leq i \leq n }
}}}
$.
Observe that $\new y. \tcontext{ }$ is a $n$-ary killing context as required.
\smallskip

Consider the commutative case where the \textit{new} quantifier commutes with $P_0 \cseq P_1$ as the principal formula.
In this case the bottommost rule of a proof is of the form $
\infer[]{
\left( P_0 \cseq P_1 \right) \cpar \new y. Q \cpar R \cpar S 
}{
\mathopen{\new y.} \left( \left( P_0 \cseq P_1 \right) \cpar Q \cpar R\right) \cpar S
}$,
where $\nfv{y}{\left( P_0 \cseq P_1 \right) \cpar R}$.
By induction, there exist $T$ and $U$ such that $\vdash \left( P_0 \cseq P_1 \right) \cpar Q \cpar R \cpar U$, $\nfv{y}{T}$ and either $T = U$ or $T = \wen y. U$, and also
$\infer{S}{T}$.
Furthermore, the size of the proof of $\left( P_0 \cseq P_1 \right) \cpar Q \cpar R \cpar U$ is bounded above by the size of the proof of $\mathopen{\new y.} \left( \left( P_0 \cseq P_1 \right) \cpar Q \cpar R\right) \cpar S$ and hence strictly bounded above by the size of the proof of $\left( P_0 \cseq P_1 \right) \cpar \new y. Q \cpar R \cpar S$, enabling the induction hypothesis.
Hence, by the induction hypothesis, there exist formulae $V_i$ and $W_i$ such that $\vdash P_0 \cpar V_i$ and $\vdash P_1 \cpar W_i$, for $1 \leq i \leq n$, and also $n$-ary killing context $\tcontext{}$ such that
$\infer[]{
Q \cpar R \cpar U
}{
\tcontext{ V_i \cseq W_i \colon 1 \leq i \leq n }
}$.
Furthermore, the size of the proofs of $P_0 \cpar V_i$ and $P_1 \cpar W_i$ are bounded above by the size of the proof of $\left( P_0 \cseq P_1 \right) \cpar \new y. Q \cpar R \cpar S$.
Now observe that either $T = U$ and $\nfv{y}{U}$ hence
$
\infer[]{
\new y. Q \cpar R \cpar T
}{
\mathopen{\new y.} \left(Q \cpar R \cpar U\right)
}$;
or $T = \wen y. U$ hence $
\infer[]{
\new y. Q \cpar R \cpar \wen y. U
}{
\infer[]{
 \mathopen{\new y.} \left(Q \cpar R\right) \cpar \wen y. U 
}{ \mathopen{\new y.} \left(Q \cpar R \cpar U\right)
}}$.
Hence the following derivation can be constructed:
$
\infer[]{
\new y. Q \cpar R \cpar S
}{
\infer[]{
\new y. Q \cpar R \cpar T
}{
\infer[]{
\mathopen{\new y.} \left( Q \cpar R \cpar U \right)
}{
\new y. \tcontext{ V_i \cseq W_i \colon 1 \leq i \leq n }
}}}
$.
Observe that $\new y. \tcontext{ }$ is a $n$-ary killing context as required.
\smallskip

Consider the commutative case when the \textit{wen} quantifier commutes with $P_0 \cseq P_1$ as the principal formula.
In this case the bottommost rule of a proof is of the form $
\infer[]{
\left( P_0 \cseq P_1 \right) \cpar \wen y. Q \cpar R \cpar S
}{
\mathopen{\wen y.} \left( \left( P_0 \cseq P_1 \right) \cpar Q \cpar R\right) \cpar S
}$. where $\nfv{y}{\left( P_0 \cseq P_1 \right) \cpar R}$.
By induction, there exist $T$ and $U$ such that $\vdash \left( P_0 \cseq P_1 \right) \cpar Q \cpar R \cpar U$, $\nfv{y}{T}$ and either $T = U$ or $T = \new y. U$, and also
$\infer{S}{T}$.
Furthermore, the size of the proof of $\left( P_0 \cseq P_1 \right) \cpar Q \cpar R \cpar U$ is bounded above by the size of the proof of $\mathopen{\wen y.} \left( \left( P_0 \cseq P_1 \right) \cpar Q \cpar R\right) \cpar S$ and hence strictly bounded above by the size of the proof of $\left( P_0 \cseq P_1 \right) \cpar \wen y. Q \cpar R \cpar S$, enabling the induction hypothesis.
Hence, by the induction hypothesis, there exist formulae $V_i$ and $W_i$ such that $\vdash P_0 \cpar V_i$ and $\vdash P_1 \cpar W_i$, for $1 \leq i \leq n$, and also $n$-ary killing context $\tcontext{}$ such that
$
\infer[]{
Q \cpar R \cpar U 
}{
 \tcontext{ V_i \cseq W_i \colon 1 \leq i \leq n }
}$.
Furthermore, the size of the proofs of $P_0 \cpar V_i$ and $P_1 \cpar W_i$ are bounded above by the size of the proof of $\left( P_0 \cseq P_1 \right) \cpar \wen y. Q \cpar R \cpar S$.
Now observe that either $T = U$ and $\nfv{y}{U}$ hence
$
\infer[]{
\wen y. Q \cpar R \cpar T
}{
\infer[]{
 \wen y. \left(Q \cpar R \cpar U\right) 
}{ \mathopen{\new y.} \left(Q \cpar R \cpar U\right)
}}$;
or $T = \new y. U$ hence $
\infer[]{
\wen y. Q \cpar R \cpar \new y. U
}{
\infer[]{
 \wen y. \left(Q \cpar R\right) \cpar \new y. U 
}{
 \mathopen{\new y.} \left(Q \cpar R \cpar U\right)
}}$.
Hence the following derivation can be constructed:
$
\infer[]{
\wen y. Q \cpar R \cpar S
}{
\infer[]{
\wen y. Q \cpar R \cpar T
}{
\infer[]{
\mathopen{\new y.} \left( Q \cpar R \cpar U \right)
}{
\mathopen{\new y.} \tcontext{ V_i \cseq W_i \colon 1 \leq i \leq n }
}}}
$.
Observe that $\new y. \tcontext{ }$ is a $n$-ary killing context as required.
\smallskip

Consider commutative cases where the principal formula moves entirely to the left hand side of a \textit{seq} operator. 
For principal formula $\wen x. U$, the bottommost rule in a proof is of the form
$
\infer[]{
\wen x. U \cpar \left( V \cseq P \right) \cpar W \cpar Q
}{
\left( \left( \wen x. U \cpar V \cpar W\right) \cseq P \right) \cpar Q
}
$
such that $\vdash \left( \left( \wen x. U \cpar V \cpar W\right) \cseq P \right) \cpar Q$ holds.
By the induction hypothesis, there exist $R_i$ and $S_i$ such that $\vdash \wen x. U \cpar V \cpar W \cpar R_i$ and $\vdash P \cpar S_i$, for $1 \leq i \leq n$, and $n$-ary killing context $\tcontext{}$ such that the derivation
$
\infer[]{
 Q
}{
 \tcontext{ R_1 \cseq S_1, \hdots, R_n \cseq S_n }
}
$ holds, 
and furthermore the size of the proof of $\wen x. U \cpar V \cpar W \cpar R_i$ is bounded above by the size of the proof of $\left( \left( \wen x. U \cpar V \cpar W\right) \cseq P \right) \cpar Q$ hence strictly bounded above by the size of the proof of $\wen x. U \cpar \left( V \cseq P \right) \cpar W \cpar Q$ enabling the induction hypothesis.
By the induction hypothesis again, there exist formulae $T_i$ and $\hat{T}_i$ such that $\vdash U \cpar \hat{T}_i$ and $\nfv{x}{T_i}$ and either $T_i = U_i$ or $T_i = \new x. U_i$, and also the derivation $
\infer[]{
V \cpar W \cpar R_i
}{
 T_i
}$ holds.
Hence the following derivation can be constructed.
\[
\infer[]{
\left( V \cseq P \right) \cpar W \cpar Q
}{
\infer[]{
\left( V \cseq P \right) \cpar W \cpar \tcontext{ R_1 \cseq S_1, \hdots, R_n \cseq S_n } \\
}{
\infer[]{
\tcontext{ \left( V \cseq P \right) \cpar W \cpar R_i \cseq S_i \colon 1 \leq i \leq n } \\
}{
\infer[]{
\tcontext{ \left( V \cpar W \cpar R_i \right) \cseq \left( P \cpar S_i \right) \colon 1 \leq i \leq n } 
}{
\infer[]{
\tcontext{ V \cpar W \cpar R_i \colon 1 \leq i \leq n } 
}{
\tcontext{ T_i \colon 1 \leq i \leq n }
}}}}}
\]
Furthermore, by Lemma~\ref{lemma:bound},
$\size{T_i} \preceq \size{\left( V \cseq P \right) \cpar W \cpar Q}$
and hence $\size{ U \cpar \hat{T}_i } \prec \size{\wen x. U \cpar \left( V \cseq P \right) \cpar W \cpar Q}$, since the \textit{wen} count strictly decreases.
The cases where the principal formula moves entirely to the right hand side of the \textit{seq} operator, and the analogous case for \textit{times}, are similar to the above case.

\item \textbf{Commutative cases involving all and with.}
Consider the commutative case for \textit{with} where $T_0 \tensor T_1$ is the principal formula. In this case the bottommost rule is the following form, such that $\vdash \left(\left(\left(T_0 \tensor T_1\right) \cpar U \cpar W\right) \wwith \left(\left(T_0 \tensor T_1\right) \cpar V \cpar W\right)\right) \cpar P$ holds:
$
\infer[]{
\left(T_0 \tensor T_1\right) \cpar \left( U \wwith V \right) \cpar W \cpar P
}{
\left(\left(T_0 \tensor T_1\right) \cpar U \cpar W \wwith
\left(T_0 \tensor T_1\right) \cpar V \cpar W\right)
\cpar P
}
$.
By the induction hypothesis, $\vdash \left(T_0 \tensor T_1\right) \cpar U \cpar W \cpar P$ and $\vdash \left(T_0 \tensor T_1\right) \cpar V \cpar W \cpar P$; and furthermore strict multiset inequalities 
$\occ{\left(T_0 \tensor T_1\right) \cpar U \cpar W \cpar P} \mstrict \occ{ \left(T_0 \tensor T_1\right) \cpar \left( U \wwith V \right) \cpar W \cpar P }$
and
$\occ{\left(T_0 \tensor T_1\right) \cpar V \cpar W \cpar P} \mstrict \occ{ \left(T_0 \tensor T_1\right) \cpar \left( U \wwith V \right) \cpar W \cpar P }$
hold.
Hence, by the induction hypothesis, there exist $Q^0_i$ and $Q^1_i$ such that $\vdash T_0 \cpar Q^0_i$ and $\vdash T_1 \cpar Q^1_i$, for $1 \leq i \leq m$; and $R_j^0$ and $R_j^1$ such that
$\vdash T_0 \cpar R^0_j$ and $\vdash T_1 \cpar R^1_j$, for $1 \leq j \leq n$;
and also $m$-ary killing context $\tcontextn{0}{}$ and $n$-ary killing context $\tcontextn{1}{}$
such that the two derivations 
$
\infer[]{
U \cpar W \cpar P 
}{
 \tcontextn{0}{Q^0_i \cpar Q^1_i \colon 1 \leq i \leq m }
}
$
and
$
\infer[]{
V \cpar W \cpar P
}{
 \tcontextn{1}{R^0_j \cpar R^1_j \colon 1 \leq j \leq n }
}
$ hold.
Furthermore, the size of the proofs of $T_0 \cpar R^0_j$, $T_1 \cpar R^1_j$, $T_0 \cpar Q^0_i$ and $T_1 \cpar Q^1_i$ are bounded above by the size of the proof of $\left(T_0 \tensor T_1\right) \cpar \left( U \wwith V \right) \cpar W \cpar P$.
Thereby the following derivation can be constructed.
\[
\infer[]{
\left( U \wwith V \right) \cpar W \cpar P
}{
\infer[]{
\left(U \cpar W \cpar P\right) \wwith \left(V \cpar W \cpar P\right) 
}{
\tcontextn{0}{Q^0_i \cpar Q^1_i \colon 1 \leq i \leq m }
\wwith
\tcontextn{1}{R^0_j \cpar R^1_j \colon 1 \leq j \leq n }
}}
\]

Consider the commutative case where universal quantification commutes with $T_0 \tensor T_1$  as the principal formula. Suppose the bottommost rule is of the form
$
\infer[]{
\left(T_0 \tensor T_1\right) \cpar \forall y. U \cpar V \cpar W
}{
\mathopen{\forall y.}\left( \left(T_0 \tensor T_1\right) \cpar U \cpar V \right) \cpar W
}$, 
assuming $\nfv{y}{(\left(T_0 \tensor T_1\right) \cpar V)}$
where $\vdash \mathopen{\forall y.}\left( \left(T_0 \tensor T_1\right) \cpar U \cpar V \right) \cpar W$ holds.
By Lemma~\ref{lemma:universal}, for every variable $z$, $\vdash \left( \left(T_0 \tensor T_1\right) \cpar U \cpar V \right)\sub{y}{z} \cpar W$ holds.
Furthermore, since $\nfv{y}{(\left(T_0 \tensor T_1\right) \cpar V)}$, $\left( \left(T_0 \tensor T_1\right) \cpar U \cpar V \right)\sub{y}{z} \cpar W \equiv \left(T_0 \tensor T_1\right) \cpar U\sub{y}{z} \cpar V \cpar W$.
Since $\forall$ is removed, $
\occ{ \left(T_0 \tensor T_1\right) \cpar U\sub{y}{z} \cpar V \cpar W } \mstrict \occ{ \left(T_0 \tensor T_1\right) \cpar \forall y. U \cpar V \cpar W }
$ holds.
Pick a fresh $z$ such that $\nfv{z}{(\forall y. U \cpar V \cpar W)}$.
Hence, by the induction hypothesis, there exist formulae $P_i$ and $Q_i$ such that $\vdash T_0 \cpar P_i$ and $\vdash T_1 \cpar Q_i$, for $1 \leq i \leq n$; and also $n$-ary killing context $\tcontext{}$ such that
$
\infer[]{
U\sub{y}{z} \cpar V \cpar W
}{
\tcontext{ P_i \cpar Q_i \colon 1 \leq i \leq n }
}
$.
Furthermore, the size of the proof of $T_0 \cpar P_i$ and $T_1 \cpar Q_i$ are bounded above by the size of the proof of $\left(T_0 \tensor T_1\right) \cpar \forall y. U \cpar V \cpar W$.
Since $z$ was chosen such that $\nfv{z}{ \forall y. U \cpar V \cpar W }$ the following derivation can be constructed, as required.
$
\infer[]{
\forall y. U \cpar V \cpar W
}{
\infer[]{
\mathopen{\forall z.} \left( U\sub{y}{z} \cpar V \cpar W \right) 
}{
\mathopen{\forall z.} \tcontext{ P_i \cpar Q_i \colon 1 \leq i \leq n }
}}
$.
\smallskip

Consider the commutative case for \textit{with} where $T_0 \cseq T_1$ is the principal formula. The bottommost rule is the form
$
\infer[]{
\left(T_0 \cseq T_1\right) \cpar \left( U \wwith V \right) \cpar W \cpar P
}{
\left(\left(T_0 \cseq T_1\right) \cpar U \cpar W \wwith
\left(T_0 \cseq T_1\right) \cpar V \cpar W\right)
\cpar P
}
$
where $\vdash \left(\left(\left(T_0 \cseq T_1\right) \cpar U \cpar W\right) \wwith \left(\left(T_0 \cseq T_1\right) \cpar V \cpar W\right)\right) \cpar P$ holds.
By the induction hypothesis, $\vdash \left(T_0 \cseq T_1\right) \cpar U \cpar W \cpar P$ and $\vdash \left(T_0 \cseq T_1\right) \cpar V \cpar W \cpar P$; and furthermore the strict multiset inequalities
$\occ{\left(T_0 \cseq T_1\right) \cpar U \cpar W \cpar P} \mstrict \occ{ \wen x. T \cpar \left( U \wwith V \right) \cpar W \cpar P }$
and
$\occ{\left(T_0 \cseq T_1\right) \cpar V \cpar W \cpar P} \mstrict \occ{ \wen x. T \cpar \left( U \wwith V \right) \cpar W \cpar P }$ hold.

Hence, by the induction hypothesis, there exist $Q^0_i$ and $Q^1_i$ such that $\vdash T_0 \cpar Q^0_i$ and $\vdash T_1 \cpar Q^1_i$, for $1 \leq i \leq m$; and $R_j^0$ and $R_j^1$ such that
$\vdash T_0 \cpar R^0_j$ and $\vdash T_1 \cpar R^1_j$, for $1 \leq j \leq n$;
and also $m$-ary killing context $\tcontextn{0}{}$ and $n$-ary killing context $\tcontextn{1}{}$
such that the two derivations 
$
\infer[]{
U \cpar W \cpar P
}{
 \tcontextn{0}{Q^0_i \cseq Q^1_i \colon 1 \leq i \leq m }
}
$
and
$
\infer[]{
V \cpar W \cpar P 
}{
  \tcontextn{1}{R^0_j \cseq R^1_j \colon 1 \leq j \leq n }
}
$ hold.
Furthermore, the size of the proofs of $T_0 \cpar R^0_j$, $T_1 \cpar R^1_j$, $T_0 \cpar Q^0_i$ and $T_1 \cpar Q^1_i$ are bounded above by the size of the proof of $\left(T_0 \cseq T_1\right) \cpar \left( U \wwith V \right) \cpar W \cpar P$.
Thereby the following derivation can be constructed.
\[
\infer[]{
\left( U \wwith V \right) \cpar W \cpar P
}{
}{
\infer[]{
\left(U \cpar W \cpar P\right) \wwith \left(V \cpar W \cpar P\right) 
}{
\tcontextn{0}{Q^0_i \cseq Q^1_i \colon 1 \leq i \leq m }
\wwith
\tcontextn{1}{R^0_j \cseq R^1_j \colon 1 \leq j \leq n }
}}
\]
\smallskip

Consider the commutative case where universal quantification commutes with $T_0 \cseq T_1$  as the principal formula. Suppose the bottommost rule is of the form
$
\infer[]{
\left(T_0 \cseq T_1\right) \cpar \forall y. U \cpar V \cpar W
}{
\mathopen{\forall y.}\left( \left(T_0 \cseq T_1\right) \cpar U \cpar V \right) \cpar W
}$,
 assuming $\nfv{y}{(\wen x. T \cpar V)}$
where $\vdash \mathopen{\forall y.}\left( \left(T_0 \cseq T_1\right) \cpar U \cpar V \right) \cpar W$ holds.
By Lemma~\ref{lemma:universal}, for every variable $z$, $\vdash \left( \left(T_0 \cseq T_1\right) \cpar U \cpar V \right)\sub{y}{z} \cpar W$ holds.
Furthermore, since $\nfv{y}{(\left(T_0 \cseq T_1\right) \cpar V)}$, $\left( \left(T_0 \cseq T_1\right) \cpar U \cpar V \right)\sub{y}{z} \cpar W \equiv \left(T_0 \cseq T_1\right) \cpar U\sub{y}{z} \cpar V \cpar W$.

The strict multiset inequality $
\occ{ \left(T_0 \cseq T_1\right) \cpar U\sub{y}{z} \cpar V \cpar W } \mstrict \occ{ \left(T_0 \tensor T_1\right) \cpar \forall y. U \cpar V \cpar W }
$ holds.
Pick a fresh $z$ such that $\nfv{z}{(\forall y. U \cpar V \cpar W)}$.
Hence, by the induction hypothesis, there exist formulae $P_i$ and $Q_i$ such that $\vdash T_0 \cpar P_i$ and $\vdash T_1 \cpar Q_i$, for $1 \leq i \leq n$; and also $n$-ary killing context $\tcontext{}$ such that
$
\infer[]{
U\sub{y}{z} \cpar V \cpar W
}{
\tcontext{ P_i \cseq Q_i \colon 1 \leq i \leq n }
}
$.
Furthermore, the size of the proof of $T_0 \cpar P_i$ and $T_1 \cpar Q_i$ are bounded above by the size of the proof of $\left(T_0 \cseq T_1\right) \cpar \forall y. U \cpar V \cpar W$.

Since $z$ was chosen such that $\nfv{z}{ \forall y. U \cpar V \cpar W }$ the following derivation can be constructed, as required.
$
\infer[]{
\forall y. U \cpar V \cpar W
}{
\infer[]{
\mathopen{\forall z.} \left( U\sub{y}{z} \cpar V \cpar W \right) 
}{
\forall z. \tcontext{ P_i \cseq Q_i \colon 1 \leq i \leq n }
}}
$.
\smallskip

Consider the commutative case for \textit{sequence} rule in the presence of principal formula $T \cseq U$, where the \textit{seq} connective in the principal formula is not active on the \textit{sequence} rule. In this case, the bottommost rule in a proof is an instance of the \textit{sequence} rule of the form
$
\infer[]{
\left( T \cseq U \right) \cpar \left( V \cseq P \right) \cpar W \cpar Q
}{
\left( \left( \left( T \cseq U\right) \cpar V \cpar W\right) \cseq P \right) \cpar Q
}
$,
where $T \cseq U \not\equiv \cunit$ and $P \not\equiv \cunit$ and $\vdash \left( \left( \left( T \cseq U\right) \cpar V \cpar W\right) \cseq P \right) \cpar Q$ holds.
By the induction hypothesis, there exists $R_i$, $S_i$ such that $\vdash  \left( T \cseq U\right) \cpar V \cpar W \cpar R_i$ and $\vdash P \cpar S_i$, for $1 \leq i \leq n$, and $n$-ary killing context $\tcontext{}$ such that the following derivation holds:
$
\infer[]{
 Q
}{
 \tcontext{ R_1 \cseq S_1, \hdots, R_n \cseq S_n }
}
$.
Furthermore, $\size{\left( T \cseq U\right) \cpar V \cpar W \cpar R_i} \mleq \size{\left( \left( \left( T \cseq U\right) \cpar V \cpar W\right) \cseq P \right) \cpar Q}$ hence the induction hypothesis is enabled again.
By the induction hypothesis, for $1 \leq i \leq n$, there exist formulae $P^i_j$ and $Q^i_j$ such that $\vdash T \cpar P^i_j$ and $\vdash U \cpar Q^i_j$ hold, for $1 \leq j \leq m_i$, and killing contexts $\tcontextn{i}{}$ such that the following derivation holds.
\[
\infer[]{
V \cpar W \cpar R_i
}{
 \tcontextn{i}{P^i_1 \cseq Q^i_1, \hdots, P^i_{m_i} \cseq Q^i_{m_i}}
}
\]
Furthermore the following strict multiset inequalities hold.
\[
\occ{T \cpar P^i_j} \mstrict \occ{\left( T \cseq U \right) \cpar \left( V \cseq P \right) \cpar W \cpar Q
}
\qquad
\mbox{and}
\qquad
\occ{U \cpar Q^i_j} \mstrict \occ{\left( T \cseq U \right) \cpar \left( V \cseq P \right) \cpar W \cpar Q
}
\]
Hence the following derivation can be constructed, as required.
\[
\infer[]{
\left( V \cseq P \right) \cpar W \cpar Q
}{
\infer[]{
\left( V \cseq P \right) \cpar W \cpar \tcontext{ R_1 \cseq S_1, \hdots, R_n \cseq S_n }
}{
\infer[]{
\tcontext{ \left( V \cseq P \right) \cpar W \cpar \left(R_1 \cseq S_1\right), \hdots, \left( V \cseq P \right) \cpar W \cpar \left(R_n \cseq S_n\right) }
}{
\infer[]{
\tcontext{ \left( V \cpar W \cpar R_1 \right) \cseq \left( P \cpar S_1 \right), \hdots, \left( V \cpar W \cpar R_n \right) \cseq \left( P \cpar S_n \right) }
}{
\infer[]{
\tcontext{ V \cpar W \cpar R_1, \hdots, V \cpar W \cpar R_n }
}{
\tcontext{ \tcontextn{i}{P^i_j \cseq Q^i_j \colon 1 \leq j \leq m_i } \colon 1 \leq i \leq n }
}}}}}
\]
The case for the \rseq rule commuting with the principal formula $T \tensor U$ is similar to the above case. Also the cases for the \textit{switch} rule commuting with seq and times as the principal formula, follow a similar pattern.
\smallskip

}  


\item \textbf{Commutative cases deep in contexts.}
In many commutative cases, the bottommost rule does not interfere with the principal formula either directly or indirectly. Two such cases are presented for \textit{wen} as the principal connective. Other such cases use almost identical reasoning.

\begin{enumerate}[label*=\textbf{.\arabic*}]

\item Consider when a rule is applied outside the scope of the principal formula.
In this case, the bottommost rule in a proof is of the form
\[
\infer[\mbox{, such that $\vdash \wen x. U \cpar \context{ W }$.}]{
\wen x. U \cpar \context{ V }
}{
 \wen x. U \cpar \context{ W }
}
\]
By the induction hypothesis, there exist formulae $P$ and $Q$ such that $\vdash U \cpar Q$ and $\nfv{x}{P}$ and either $P = Q$ or $P = \new x. Q$, and also $
\vcenter{\infer[]{
\context{ W }
}{
 P
}}
$.
Hence clearly derivation $
\vcenter{\infer[]{
\context{ V }
}{
\infer[]{
 \context{ W } 
}{
  P
}}}
$
 holds.
Furthermore, by Lemma~\ref{lemma:bound}, $\size{ \wen x. U \cpar \context{ W }} \prec \size{ U \cpar \context{ W }}$ 
and $\size{ U \cpar \context{ W }} \preceq \size{\wen x. U \cpar \context{ V }}$.

\fullproof{
Assume that the following application of any rule
$
\infer[]{
\left( T \cseq U \right) \cpar \context{ V }
}{
 \left( T \cseq U \right) \cpar \context{ W }
}$
is the bottommost rule in a proof,
such that $\vdash \left( T \cseq U \right) \cpar \context{ W }$.
By the induction hypothesis, there exist $n$-ary killing context $\tcontext{}$ and formulae $Q_i$ and $R_i$ such that $\vdash T \cpar Q_i$ and $\vdash U \cpar R_i$, for $1 \leq i \leq n$, such that $
\infer[]{
\context{ W } 
}{
\tcontext{ Q_1 \cseq R_1, \ldots, Q_n \cseq R_n }
}
$.
Hence, the derivation 
$
\infer[]{
\context{ V }
}{
\infer[]{
 \context{ W } 
 \tcontext{ Q_1 \cseq R_1, \ldots, Q_n \cseq R_n }
}}
$
holds, satisfying the induction invariant.

Alternatively, the bottommost rule may appear inside the context of principal formula without affecting the root connective of the principal formula. 
%
Consider the case where \textit{seq} is the principal formula.
Assume that the following application of any rule is the bottommost rule in a proof
$
\infer[]{
\left( \context{T} \cseq V \right) \cpar W 
}{
 \left( \context{U} \cseq V \right) \cpar W
}
$
such that $\vdash \left( \context{U} \cseq V \right) \cpar W$ has a proof of length $n$. 
Hence by induction, there exist $n$-ary killing context $\tcontext{}$ and formulae $P_i$ and $Q_i$ such that $\vdash \context{U} \cpar P_i$ and $\vdash V \cpar Q_i$ hold and have a proof no longer than $n$, for $1 \leq i \leq n$, and also 
$
\infer[]{
W
}{
 \tcontext{ P_1 \cseq Q_1, \ldots, P_n \cseq Q_n }
}
$.
Hence we can construct the following proof of length no longer than $n+1$, for all $i$, as required:
$
\infer[]{
\context{T} \cpar P_i
}{
\infer[]{
   \context{U} \cpar P_i 
}{
 \cunit
}}
$.
}

\item Consider the case where the following application of any rule in a derivation of the form 
\[
\infer[]{
\wen x. \context{T} \cpar W
}{
 \wen x. \context{U} \cpar W
}
\]
is the bottommost rule is a proof of length $k+1$, where $\vdash \wen x. \context{U} \cpar W$ has a proof of length $k$.
Hence, by induction, there exist formulae $P$ and $Q$ such that $\vdash \context{U} \cpar Q$ and $\nfv{x}{P}$ and either $P = Q$ or $P = \new x. Q$, and also
$
\vcenter{
\infer[]{
W
}{
P
}}
$.
Furthermore, the size of the proof of $\context{U} \cpar Q$ is bounded above by the size of the proof of $\wen x. \context{U} \cpar W$; hence either $\size{\context{U} \cpar Q} \prec \size{\wen x. \context{U} \cpar W}$ or $\size{\context{U} \cpar Q} = \size{\wen x. \context{U} \cpar W }$ and the length of the proof of $U \cpar Q$ is bound by $k$.
The proof 
$
\vcenter{
\infer[]{
\context{T} \cpar Q
}{
\infer[]{
 \context{U} \cpar Q 
}{
   \cunit
}}}
$ can be constructed as required.
Furthermore, if $\size{\context{U} \cpar Q} \prec \wen x. \size{\context{U} \cpar W}$ then $\size{\context{U} \cpar Q} \prec \size{\wen x. \context{U} \cpar \context{ V }}$, by Lemma~\ref{lemma:bound}. Otherwise, $\size{\context{U} \cpar Q} = \size{\wen x. \context{U} \cpar W }$
hence $\size{U \cpar Q} \preceq \size{\wen x. U \cpar \context{ V }}$ by Lemma~\ref{lemma:bound} and the length of the proof of $\vdash \context{T} \cpar Q$ is $k+1$.
Thereby in either case, the size of the proof of $\context{T} \cpar Q$ is bounded above by the size of the proof of $\wen x. \context{T} \cpar W$.
\end{enumerate}

%% file: context.tex
\section{Context Reduction and the Admissibility of Co-rules}
\label{section:context}

The splitting lemmas  in the previous section are formulated for sequent-like \textit{shallow contexts}.
By applying splitting repeatedly, \textit{context reduction} (Lemma~\ref{lemma:context}) is established, which can be used to extends normalisation properties to an arbitrary (deep) context.
In particular, we extend a series of proof normalisation properties called \textit{co-rule elimination} properties to any context, by first establishing the normalisation property in a shallow context, then applying context reduction to extend to any context.
Together, these \textit{co-rule elimination} properties establish cut elimination, by eliminating each connective directly involved in a cut one-by-one.

%

\subsection{Extending from a sequent-like context to a deep context}

Context reduction extends rules simulated by splitting to any context.
This appears to be the first context reduction lemma in the literature to handle first-order quantifiers.
Of particular note is the use of substitutions to account for the effect of existential quantifiers in the context.
The trick is to first establish the following stronger invariant.
\begin{lemma}\label{lemma:invariant}
If $\vdash \context{T}$, then there exist formulae $U_i$ and substitutions $\sigma_i$, for $1 \leq i \leq n$, and $n$-ary killing context $\tcontext{}$ such that $\vdash T\sigma_i \cpar U_i$; and, for any formula $V$ there exist $W_i$ such that either $W_i = V\sigma_i \cpar U_i$ or $W_i = \cunit$ and the following holds:
$
\vcenter{
\infer[]{
 \context{V}
}{
 \tcontext{ W_1, W_2, \hdots, W_n }
}
}$.
\end{lemma}

Having established the above stronger invariant, the context lemma follows directly.
\begin{lemma}[Context reduction]
\label{lemma:context}
If $\vdash P\mathclose\sigma \cpar R$ yields that $\vdash Q\mathclose\sigma \cpar R$, for any formula $R$ and substitution of terms for variables $\sigma$, then $\vdash \context{ P }$ yields $\vdash \context{ Q }$, for any context $\context{}$.
\end{lemma}
\tohide{
\input{contextproof}
Note that the case for existential quantifiers will not work for second-order quantifiers, since termination of the induction is reliant on the size of the term-free part of the formula being reduced. Thus the techniques in the above proof apply to first-order quantifiers only.
}{
The subtelty of context reduction is to accommodate \textit{plus} and \textit{some} by the following stronger induction invariant:
If $\vdash \context{T}$, then there exist formulae $U_i$ and substitutions $\sigma_i$ such that $\vdash T\mathclose{\sigma_i} \cpar U_i$, for $1 \leq i \leq n$; and $n$-ary killing context $\tcontext{}$  such that for any formula $V$ there exist $W_i$ such that either $W_i = V\mathclose{\sigma_i} \cpar U_i$ or $W_i = \cunit$, for $1 \leq i \leq n$, and the following holds:
$
\vcenter{
\infer[]{
 \context{V}
}{
 \tcontext{ W_1, W_2, \hdots, W_n }
}
}$.
}

%% file: contextproof.tex
\begin{proof}
Assume that for any formula $U$, $\vdash S \cpar U$ yields $\vdash T \cpar U$, and fix any context $\context{}$ such that $\vdash \context{S}$ holds. By Lemma~\ref{lemma:invariant}, there exist $n$-ary killing context $\tcontext{}$ and,
for $1 \leq i \leq n$, $P_i$ such that either $P_i = \cunit$ or there exists $W_i$ where $P_i = T \cpar W_i$ and $\vdash S \cpar W_i$, and furthermore
$
\vcenter{
\infer[]{
\context{T}
}{
\tcontext{P_1, \hdots, P_n}
}}$.
Since, by our assumption, also $\vdash T \cpar W_i$ holds for $1 \leq i \leq n$, the proof 
$
\vcenter{
\infer[]{
\context{T}
}{
\infer[]{
 \tcontext{P_1, \hdots, P_n} 
}{
\infer[]{
 \tcontext{\cunit, \hdots, \cunit} 
}{
 \cunit
}}}}
$ can be constructed.
Therefore $\vdash \context{T}$ holds.
\end{proof}

%% file: corule.tex
\subsection{Cut elimination as co-rule elimination}

\begin{figure}
\begin{gather*}
\infer[\mbox{(atomic co-interaction)}]{
\context{ \cunit }
}{
\context{ \alpha \tensor \co{\alpha} }
}
\qquad
\infer[\mbox{(co-select1)}]{
\context{ P \sub{x}{v} }
}{
\context{ \forall x. P }
}
\\[9pt]
\infer[\mbox{(co-sequence)}]{
\context{(P \tensor U) \cseq (Q \tensor V)}
}{
\context{(P \cseq Q) \tensor (U \cseq V)}
}
\qquad\qquad
\infer[\mbox{(co-external)}]{
\context{ (P \cpar R) \ooplus (Q \cpar S) }
}{
\context{ (P \ooplus Q) \cpar S }
}
\\[9pt]
\infer[\mbox{(co-tidy)}]{
\context{ \cunit }
}{
\context{ \cunit \ooplus \cunit }
}
\qquad\qquad
\infer[\mbox{(co-left)}]{
\context{ P}
}{
\context{ P \wwith Q }
}
\qquad\qquad
\infer[\mbox{(co-right)}]{
\context{ Q}
}{
\context{ P \wwith Q }
}
\\[9pt]
\infer[\mbox{(co-extrude1)}]{
\context{ \mathopen{\exists x.}\left( P \tensor R \right) }
}{
\context{ \exists x. P \tensor R }
}
\qquad\qquad\qquad
\infer[\mbox{(co-tidy1)}]{
\context{ \cunit }
}{
\context{ \exists x. \cunit }
}
\\[9pt]
\infer[\mbox{(co-close)}]{
\context{ \mathopen{\wen x .}\left( P \tensor Q \right) }
}{
\context{ \new x. P \tensor \wen x. Q }
}
\qquad\qquad\qquad
\infer[\mbox{(co-tidy name)}]{
\context{ \cunit }
}{
\context{ \wen x. \cunit }
}
\end{gather*}
\caption[Co-rules extending the system \textsf{MAV1} to \textsf{SMAV1}.]
{Co-rules extending the system \textsf{MAV1} to \textsf{SMAV1}, where 
 $\nfv{x}{R}$.}
\label{figure:co-rules}
\end{figure}
For a rule of the form $\vcenter{\infer{P}{Q}}$, there is a  corresponding \textit{co-rule} of the form $\vcenter{\infer{\co{Q}}{\co{P}}}$, where premise and conclusion are interchanged and each formula is dualised using negation. 
The rules \textit{switch}, \textit{fresh} and \textit{new wen} are their own co-rules. 
Also the co-rule of the \textit{medial new} rule is an instance of the \textit{suspend} rule.
All other rules give rise to distinct co-rules, presented in Figure~\ref{figure:co-rules}.
Note co-rules with no role in cut elimination are ommitted from the figure.


The following nine lemmas each establish that a co-rule is admissible in \textsf{MAV1}.
Only the following co-rules need be handled directly in order to establish cut elimination:
\textit{co-close}, \textit{co-tidy name}, \textit{co-extrude1}, \textit{co-select1}, \textit{co-tidy1}, \textit{co-left}, \textit{co-right}, \textit{co-external}, \textit{co-tidy}, \textit{co-sequence} and \textit{atomic co-interaction}.
In each case, the proof proceeds by applying splitting in a shallow context, forming a new proof, and finally applying Lemma~\ref{lemma:context}. 
Each co-rule can be treated independently, hence are established as separate lemmas.
\begin{lemma}[co-close]
\label{lemma:co-close}
If $\vdash \context{ \wen x. P \tensor \new x. Q }$ holds
then $\vdash \context{ \mathopen{\wen x.} \left( P \tensor Q \right) }$ holds.
\end{lemma}
\input{coclose}

\begin{lemma}[co-tidy name]
\label{lemma:co-tidynew}
If $\vdash \context{ \wen x. \cunit }$ holds
then $\vdash \context{ \cunit }$ holds.
\end{lemma}
\begin{proof}
Assume that $\vdash \wen x. \cunit \cpar P$ holds.
By Lemma~\ref{lemma:split-times}, there exists $Q$ such that $\vdash Q$ and $\vcenter{\infer{P}{Q}}$.
Hence the following proof of $P$ can be constructed:
$
\vcenter{
\infer{
P
}{
\infer[]{
Q
}{
\cunit
}}}
$.
Therefore, by Lemma~\ref{lemma:context}, for any context $\context{}$, if $\vdash \context{\wen{x}. \cunit}$ then $\vdash \context{ \cunit }$, as required.
\end{proof}

\begin{lemma}[co-extrude1]
\label{lemma:co-extrude}
If  $\nfv{x}{Q}$ and $\vdash \context{ \exists x. P \tensor Q }$ holds then $\vdash \context{ \mathopen{\exists x.} \left( P \tensor Q \right) }$ holds.
\end{lemma}
\input{coextrude}

\begin{lemma}[co-tidy1]
\label{lemma:co-tidy1}
If $\vdash \context{ \exists x. \cunit }$ holds
then $\vdash \context{ \cunit }$ holds.
\end{lemma}
\begin{proof}
Assume that $\vdash \exists x. \cunit \cpar T$ holds.
By Lemma~\ref{lemma:split-exists}, there exists $U_i$ such that $\vdash U_i$, for $1 \leq i \leq n$, and $n$-ary killing context $\tcontext{}$ such that 
$\vcenter{\infer{T}{\tcontext{U_1, \hdots, U_n}}}$.
Hence the following proof of $T$ can be constructed:
\[
\infer[.]{
\cunit \cpar T 
}{
\infer[]{
\tcontext{U_1, \hdots, U_n}
}{
\infer[]{
\tcontext{\cunit, \hdots, \cunit}
}{
\cunit
}}}
\]
Therefore, by Lemma~\ref{lemma:context}, if $\vdash \context{\exists{x} \cunit}$ then $\vdash \context{ \cunit }$, as required.
\end{proof}

The above four lemmas are particular to \textsf{MAV1}.
The following lemma is proven directly for \textsf{MAV}, similarly to Lemma~\ref{lemma:universal}; however, for \textsf{MAV1} the proof is more indirect due to interdependencies between $\wwith$ and nominals.
\begin{lemma}[co-left and co-right]\label{lemma:co-branching}
If $\vdash \context{ P \wwith Q }$ holds
then both $\vdash \context{ P }$ and $\vdash \context{ Q }$ hold.
\end{lemma}

The proofs for the four co-rule elimination lemmas below are similar to the corresponding cases in \textsf{MAV}~\cite{Horne2015}.
\begin{lemma}[co-external]
\label{lemma:co-external}
If $\vdash \context{ P \tensor  \left( Q \ooplus R \right) }$ holds
 then $\vdash \context{ \left( P \tensor Q \right) \ooplus \left( P \tensor R \right) }$ holds.
\end{lemma}

\begin{lemma}[co-sequence]
\label{lemma:co-sequence}
If $\vdash \context{ \left( P \cseq Q \right) \tensor \left( R \andthen S \right) }$ holds
 then $\vdash \context{ \left( P \tensor R \right) \cseq \left( Q \tensor S \right) }$ holds.
\end{lemma}

\begin{lemma}[co-tidy]
\label{lemma:co-tidy}
If $\vdash \context{ \cunit \ooplus \cunit }$ holds, then $\vdash \context{ \cunit }$ holds.
\end{lemma}

\begin{lemma}[atomic co-interaction]
\label{lemma:co-atoms}
If $\vdash \context{ \alpha \tensor \co{\alpha} }$ holds
then $\vdash \context{ \cunit }$ holds.
\end{lemma}

\subsection{The proof of cut elimination}
The main result of this paper, Theorem~\ref{theorem:cut}, follows by induction on the structure of $P$ in a formula of the form $\vdash \context{ P \tensor \co{P} }$, by applying the above eight co-rule elimination lemmas and also Lemma~\ref{lemma:universal} in the cases for \textit{all} and \textit{some}.

\tohide{
\begin{proof}
The base cases for any atom $\alpha$ follows since if $\vdash \context{\co{\alpha} \tensor \alpha}$ then $\vdash \context{\cunit}$ by Lemma~\ref{lemma:co-atoms}. The base case for the unit is immediate. 
As the induction hypothesis in the following cases assume for any context $\context{}$, $\vdash \context{P \tensor \co{P}}$ yields $\context{ \cunit }$ and $\vdash \contextb{Q \tensor \co{Q}}$ yields $\contextb{ \cunit }$.

Consider the case for \textit{times}. 
Assume that $\vdash \context{ P \tensor Q \tensor \left(\co{P} \cpar \co{Q}\right) }$ holds.
By the \textit{switch} rule, $\vdash \context{ \left( P \tensor \co{P} \right) \cpar \left( Q \tensor \co{Q} \right) }$ holds.
Hence, by the induction hypothesis twice, $\vdash \context{ \cunit }$ holds.
The case for \textit{par} is symmetric to the case for \textit{times}.

Consider the case for \textit{seq}. 
Assuming that $\vdash \context{ \left( P \cseq Q \right) \tensor \left(\co{P} \cseq \co{Q}\right) }$ holds,
by Lemma~\ref{lemma:co-sequence}, it holds that $\vdash \context{ \left( P \tensor \co{P} \right) \cseq \left( Q \tensor \co{Q} \right) }$.
Hence, by the induction hypothesis twice, $\vdash \context{ \cunit }$ holds.

Consider the case for \textit{with}.
Assume that $\vdash \context{ \left(P \wwith Q\right) \tensor \left(\co{P} \ooplus \co{Q}\right) }$ holds.
By Lemma~\ref{lemma:co-external},
$\vdash \context{ \left(\left(P \wwith Q\right) \tensor \co{P}\right) \ooplus \left(\left(P \wwith Q\right) \tensor \co{Q}\right) }$
holds.
By Lemma~\ref{lemma:co-branching} twice, $\vdash \context{ \left(P \tensor \co{P}\right) \ooplus \left(Q \tensor \co{Q}\right) }$ holds.
Hence by the induction hypothesis twice, $\vdash \context{ \cunit \ooplus \cunit }$ holds.
Hence by Lemma~\ref{lemma:co-tidy}, $\vdash \context{ \cunit }$ holds, as required.
The case for \textit{plus} is symmetric to the case for \textit{with}.

Consider the case for universal quantification.
Assume that $\vdash \context{ \forall x. P \tensor \exists x. \co{P} }$ holds.
By Lemma~\ref{lemma:co-extrude}, it holds that $\vdash \context{ \mathopen{\exists x.} \left( \forall x. P \tensor \co{P} \right) }$, since $\nfv{x}{\exists x. P}$.
By Lemma~\ref{lemma:universal}, $\vdash \context{ \mathopen{\exists x.} \left( P \tensor \co{P} \right) }$ holds.
Hence by the induction hypothesis, $\vdash \context{ \exists x. \cunit }$ holds.
Hence by Lemma~\ref{lemma:co-tidy1}, $\vdash \context{ \cunit }$ holds, as required.
The case for existential quantification is symmetric to the case for universal quantification.

Consider the case for \textit{new}.
Assume that $\vdash \context{ \new x. P \tensor \wen x. \co{P} }$ holds.
By Lemma~\ref{lemma:co-close}, it holds that $\vdash \context{ \mathopen{\wen x.} \left( P \tensor \co{P} \right) }$.
Hence by the induction hypothesis, $\vdash \context{ \wen x. \cunit }$ holds.
Hence by Lemma~\ref{lemma:co-tidynew}, $\vdash \context{ \cunit }$ holds, as required.
The case for \textit{wen} is symmetric to the case for \textit{new}.

Therefore, by induction on the structure of $P$, if $\vdash \context{ P \tensor \co{P} }$ holds, then $\vdash \context{ \cunit }$ holds.
\end{proof}
}{}

Notice that the structure of the above argument is similar to the structure of the argument for Proposition~\ref{proposition:reflexivity}.
The only difference is that the formulae are dualised and co-rule lemmas are applied instead of rules.

%% file: coclose.tex
\begin{proof}
Assume that $\vdash \left(\wen x. P \tensor \new x. Q \right)\mathclose\sigma \cpar R$ for some substitution of terms for variables $\sigma$.
By Lemma~\ref{lemma:split-times}, there exist $S_i$ and $T_i$ such that $\vdash \left(\wen x. P\right)\mathclose\sigma \cpar S_i$ and $\vdash \left(\new x. Q\right)\mathclose\sigma \cpar T_i$, for $1 \leq i \leq n$,
and $n$-ary killing context such that the derivation
\[
\infer{
R
}{
 \tcontext{ S_i \cpar T_i \colon 1 \leq i \leq n }
}
\] holds.
Also for some $y$ such that $\nfv{y}{\wen x. P}$, $\nfv{y}{\new x. Q}$ and $\nfv{y}{\sigma}$, $\left(\wen x. P\right)\mathclose\sigma \equiv \mathopen{\wen y.} \left( P\sub{x}{y}\mathclose\sigma \right)$ and $\left(\new x. Q\right)\mathclose\sigma \equiv \mathopen{\new y.}\left(Q\sub{x}{y}\mathclose\sigma\right)$, where $\nfv{y}{\sigma}$ is defined such that $y$ does not appear in the domain of $\sigma$ nor free in any term in the range of $\sigma$.
Hence both $\vdash \mathopen{\wen y.}\left( P\sub{x}{y}\mathclose\sigma \right) \cpar S_i$ and $\vdash \mathopen{\new y.} \left(Q\sub{x}{y}\mathclose\sigma\right) \cpar T_i$ hold.

Hence, by Lemma~\ref{lemma:split-times}, there exist $U_i$ and $\hat{U}_i$ such that $\vdash P\sub{x}{y}\mathclose\sigma \cpar \hat{U}_i$ and either $U_i = \hat{U}_i$ or $U_i = \new y. \hat{U}_i$, and also the derivation
$
\vcenter{
\infer{
S_i
}{
 U_i
}}
$ holds.

Similarly, by Lemma~\ref{lemma:split-times}, there exist $W_i$ and $\hat{W}_i$ such that $\vdash Q\sub{x}{y}\mathclose\sigma \cpar \hat{W}_i$ and either $W_i = \hat{W}_i$ or $W_i = \wen y. \hat{W}_i$, and also the derivation 
$
\vcenter{
\infer{
T_i
}{
 W_i
}}
$ holds.

There are four cases to consider for each $i$. Three of the cases are as follows.
\begin{itemize}
\item
If $U_i = \new y. \hat{U}_i$ and $W_i = \wen y. \hat{W}_i$ then 
\[
\infer[.]{
\new y. \hat{U}_i \cpar \wen y. \hat{W}_i
}{
\new y. \left( \hat{U}_i \cpar \hat{W}_i\right)
}
\]

\item 
If $U_i = \hat{U}_i$, $\nfv{y}{\hat{U}_i}$, and $W_i = \wen y. \hat{W}_i$, 
then 
\[
\infer[.]{
U_i \cpar \wen y. \hat{W}_i
}{
\infer[]{
\mathopen{\wen y.}\left( \hat{U}_i \cpar \hat{W}_i\right)
}{
\mathopen{\new y.} \left( \hat{U}_i \cpar \hat{W}_i\right)
}}
\]

\item
If $U_i = \new y. \hat{U}_i$ and $W_i = \hat{W}_i$, such that $\nfv{y}{\hat{W}_i}$ then 
\[
\infer[.]{
\new x. U_i \cpar \hat{W}_i
}{
\mathopen{\new y.} \left( \hat{U}_i \cpar \hat{W}_i\right)
}
\]
\end{itemize}
Thereby in any of the above three cases the following derivation can be constructed.
\[
\infer[]{
\left(\mathopen{\wen x.} \left( P \tensor Q \right)\right)\mathclose{\sigma} \cpar U_i \cpar W_i 
}{
\infer[]{
\left(\mathopen{\wen x.} \left( P \tensor Q \right)\right)\mathclose{\sigma} \cpar \new y. \left( \hat{U}_i \cpar \hat{W}_i \right)
}{
\mathopen{\new y.} \left(\left( P \tensor Q \right)\sub{x}{y}\mathclose\sigma \cpar \hat{U}_i \cpar
   \hat{W}_i \right)
}}
\]
In the fourth case $U_i = \hat{U}_i$ and $W_i = \hat{W}_i$, such that $\nfv{y}{\hat{W}_i}$ and  $\nfv{y}{\hat{U}_i}$ yielding the following.
\[
\infer[]{
\left(\mathopen{\wen x.} \left( P \tensor Q \right)\right)\mathclose\sigma \cpar \hat{U}_i \cpar
   \hat{W}_i
}{
\infer[]{
\mathopen{\new y.} \left(\left( P \tensor Q\right)\sub{x}{y}\mathclose\sigma\right) \cpar \hat{U}_i \cpar
   \hat{W}_i
}{
\mathopen{\new y.} \left(\left( P \tensor Q \right)\sub{x}{y}\mathclose\sigma \cpar \hat{U}_i \cpar
   \hat{W}_i \right)
}}
\]
By applying one of the above possible derivations for every $i$, the following proof can be constructed.
\[
\infer[]{
\left(\mathopen{\wen x.} \left( P \tensor Q \right)\right)\mathclose\sigma \cpar R
}{
\infer[]{
\left(\mathopen{\wen x.} \left( P \tensor Q \right)\right)\mathclose\sigma \cpar \tcontext{S_i \cpar T_i \colon 1 \leq i \leq n }
}{
\infer[]{
\left(\mathopen{\wen x.} \left( P \tensor Q \right)\right)\mathclose\sigma \cpar
\tcontext{
 U_i
 \cpar 
 W_i
 \colon 1 \leq i \leq n
}
}{
\infer[]{
\tcontext{
 \left(\mathopen{\wen x.} \left( P \tensor Q \right)\right)\mathclose\sigma
 \cpar
 U_i
 \cpar 
 W_i
 \colon 1 \leq i \leq n
}
}{
\infer[]{
\tcontext{
 \mathopen{\new y.} \left(\left( P \tensor Q \right)\sub{x}{y}\mathclose\sigma \cpar \hat{U}_i \cpar
   \hat{W}_i \right)
 \colon 1 \leq i \leq n
}
}{
\infer[]{
\tcontext{
   \mathopen{\new y.} \left(\left( P\sub{x}{y}\mathclose\sigma \cpar \hat{U}_i \right) \tensor
   \left( Q\sub{x}{y}\mathclose\sigma \cpar \hat{W}_i \right)\right)
 \colon 1 \leq i \leq n
}
}{
\infer[]{
\tcontext{
   \new y. \cunit
 \colon 1 \leq i \leq n
}
}{
 \cunit
}}}}}}}
\]
Therefore, by Lemma~\ref{lemma:context}, for all contexts $\context{}$,
if $\vdash \context{ \wen x. P \tensor \new x. Q }$ then
$\vdash \context{ \mathopen{\new x.} \left( P \tensor Q \right) }$.
\end{proof}

%% file: coextrude.tex
\begin{proof}
Assume that $\vdash \left( \exists x. P \tensor Q \right)\mathclose{\sigma} \cpar V$ holds, where $\nfv{x}{Q}$. Now, since
$\nfv{y}{(\exists x. P \tensor Q)}$ and $\nfv{y}{\sigma}$, we have
$\left( \exists x. P \tensor Q \right)\mathclose{\sigma} \cpar V
\equiv
\left( \mathopen{\exists y.} \left(P\sub{x}{y} \sigma\right) \tensor Q\sigma \right)  \cpar V$.
So, by Lemma~\ref{lemma:split-times}, there exist $T_i$ and $U_i$ such that $\vdash \mathopen{\exists y.} \left(P\sub{x}{y} \sigma\right) \cpar T_i$ and $\vdash Q\sigma \cpar U_i$, for $1 \leq i \leq n$, and $n$-ary killing context such that the derivation
\[
\infer[]{
V
}{ \tcontext{T_1 \cpar U_1, \hdots, T_n \cpar U_n}
}
\] holds.
By Lemma~\ref{lemma:split-exists}, there exist $R^i_j$ and $v^i_j$ such that $\vdash P\sub{x}{y} \sigma \sub{y}{v^i_j} \cpar R^i_j$, for $1 \leq j \leq m_i$, and $m_i$-ary killing context $\tcontextn{i}{}$ such that the derivation
\[
\infer[]{
T_i
}{
 \tcontextn{i}{ R^i_1, R^i_2, \hdots, R^i_{m_i} }
}
\] holds.
Hence the following proof can be constructed, where we appeal to $\alpha$-conversion in the conclusion.
\[
\infer[]{
\mathopen{\exists y.} \left( P\sub{x}{y}\sigma \tensor Q\sigma \right) \cpar V
}{
\infer[]{
\mathopen{\exists y.} \left( P\sub{x}{y}\sigma \tensor Q\sigma \right)
\cpar \tcontext{T_i \cpar U_i \colon 1 \leq i \leq n }
}{
\infer[]{
\mathopen{\exists y.} \left( P\sub{x}{y}\sigma \tensor Q\sigma \right)
\cpar \tcontext{
            \tcontextn{i}{ R^i_j \colon 1 \leq j \leq m_i } \cpar U_i
            \colon 1 \leq i \leq n 
           }
}{
\infer[]{
\tcontext{
 \mathopen{\exists y.} \left( P\sub{x}{y}\sigma \tensor Q\sigma \right)
 \cpar \tcontextn{i}{ R^i_j \colon 1 \leq j \leq m_i } \cpar U_i
 \colon 1 \leq i \leq n 
 }
}{
\infer[]{
\tcontext{
 \tcontextn{i}{ \mathopen{\exists y.} \left( P\sub{x}{y}\sigma \tensor Q\sigma \right)
 \cpar R^i_j \cpar U_i \colon 1 \leq j \leq m_i }
 \colon 1 \leq i \leq n 
 }
}{
\infer[]{
\tcontext{
 \tcontextn{i}{
  \left( P\sub{x}{y}\sigma\sub{y}{v^i_j} \tensor Q\sigma \right)
  \cpar R^i_j \cpar U_i \colon 1 \leq j \leq m_i }
 \colon 1 \leq i \leq n
 }
}{
\infer[]{
\tcontext{
 \tcontextn{i}{
  \left( P\sub{x}{y}\sigma\sub{y}{v^i_j} \cpar R^i_j \right) 
  \tensor
  \left( Q\sigma \cpar U_i \right)
  \colon 1 \leq j \leq m_i }
 \colon 1 \leq i \leq n
 }
}{
\infer[]{
\tcontext{
 \tcontextn{i}{
  \cunit
  \colon 1 \leq j \leq m_i }
 \colon 1 \leq i \leq n } 
}{
 \cunit
}}}}}}}}
\]
\noindent Hence, by Lemma~\ref{lemma:context}, if $\vdash \context{ \exists x. P \tensor Q }$, where $\nfv{x}{Q}$, then $\vdash \context{ \mathopen{\exists x.} \left( P \tensor Q \right) }$.
\end{proof}

%% file: symmetric.tex
\subsection{Discussion on alternative presentations of rules for \textsf{MAV1}}

Having established cut elimination (Theorem~\ref{theorem:cut}), an immediate corollary is that all co-rules in Fig.~\ref{figure:co-rules} are admissible. This can be formulated by demonstrating that linear implication coincides with the inverse of a derivation in the symmetric system \textit{SMAV1}.
\begin{corollary}
$\vdash P \multimap Q$ in \textsf{MAV1} if and only if $\infer{Q}{P}$ in \textsf{SMAV1}.
\end{corollary}
\begin{proof}
Firstly, assume $\vdash P \multimap Q$ in \textsf{MAV1},
in which case the following can be constructed in \textsf{SMAV1}: 
\[
\infer{
 Q
}{
\infer[.]{
 \left(P \tensor \co{P}\right) \cpar Q
}{
\infer[]{
 P \tensor \left(\co{P} \cpar Q\right)
}{
 P
}}}
\]
For the converse, assume $\vcenter{\infer{Q}{P}}$ in \textsf{SMAV1}; hence 
\[
\infer{
 \co{P} \cpar Q
}{\infer[]{
 \co{P} \cpar P
}{
 \cunit
}}
\] 
can be constructed.
Thereby by Lemma~\ref{lemma:universal} and Lemmas~\ref{lemma:co-close}~to~\ref{lemma:co-sequence},
the above derivation in \textsf{SMAV1} can be transformed into a proof in \textsf{MAV1}.
\end{proof}
The advantage of the definition of linear implication using provability in \textsf{MAV} rather than derivations in \textsf{SMAV1}, is that \textsf{MAV1} is \textit{analytic}~\cite{Bruscoli2009}; hence, with some care taken for existential quantifiers~\cite{Lincoln1994,Brunnler2003}, each formula gives rise to finitely many derivations up-to congruence. In contrast, in $\textsf{SMAV1}$, many co-rules can be applied indefinitely.
Notice co-rules including \textit{atomic co-interaction}, \textit{co-left} and \textit{co-tidy} can infinitely increase the size of a formula during proof search.

\textbf{A small rule set.}
Alternatively, we could extend the structural congruence with the following. 
\[
\wen x. P \equiv P 
~~\mbox{only if $\nfv{x}{P}$}
\qquad
\new x. P \equiv P 
~~\mbox{only if $\nfv{x}{P}$}
\qquad\mbox{(vacuous)}
\]
Vacuous allows nominals to be defined by the smaller set of rules \textit{close}, \textit{medial new}, \textit{suspend}, \textit{new wen}, \textit{with name}, and \textit{all wen}.
Any formula provable in this smaller system is also provable in \textsf{MAV1}, since all rules of \textsf{MAV1} can be simulated by the rules above.
Perhaps the least obvious case is the \textit{fresh} rule,
where since 
$
\vcenter{
\infer[]{\new x. \wen x. P}{\wen x. \new x. P}}
$, 
by the \textit{new wen} rule
and both 
$\wen x. \new x. P \equiv \new x. P$
and
$\wen x. P \equiv \new x. \wen x. P$
hold
using the \textit{vacuous} rule,
we have 
$
\vcenter{
\infer[]{\wen x. P}{\new x. P}}
$.

Conversely, \textit{vacuous} is a provable equivalence in \textsf{MAV1}; hence, by inductively applying cut elimination to eliminate each \textit{vacuous} rule in a proof using the smaller set of rules, we can obtain a proof with the same conclusion in \textsf{MAV1}.
The disadvantage of the above system is that the \textit{vacuous} rules can introduce an arbitrary number of nominal quantifiers at any stage in the proof leading to infinite paths in proof search,
i.e., the above system is not \textit{analytic}. 
Indeed the multiset-based measure used to guide splitting would not be respected, hence our cut elimination strategy would fail.
None the less, the smaller rule set above offers insight into design decisions.

\textbf{Alternative approaches to cut elimination.}
Further styles of proof system are possible. For example, again as a consequence of cut elimination, we can show the equivalence of \textsf{MAV1} and a system which reduces the implicit contraction in the \textit{external} rule to an atomic form $\vcenter{\infer{\alpha}{\alpha \ooplus \alpha}}$, in which additional medial rules play a central role for propagating contraction~\cite{Brunnler2001,Strassburger2002,Bruscoli2016}.
Similarly, the implicit vacuous existential quantifier introduction can be given an explicit atomic treatment~\cite{Strassburger09}.
The point is that, although the cut elimination result in this work is sufficient to establish the equivalent expressive power of systems mentioned in this subsection, further proof theoretic insight may be gained by attempting direct proofs of cut elimination in such alternative systems.
Indeed a different approach to cut elimination is required for tackling \textsf{MAV2} with second-order quantifiers.

\textbf{Note on probabilistic choice.} 
Insight from investigating the proof theory of \textsf{MAV1} led to the surprising observation that
probabilistic choice has similar proof theoretic properties to \textit{new}.
A proof theory of \textsf{MAV} extended with \textit{sub-additive} operators is explored in related work~\cite{Horne2019}.
The sub-additives, similarly to nominal quantifiers which lie between universal and existential quantifiers, lie between the traditional additives \textit{with} and \textit{plus}.
Sub-additives can either be self-dual, similarly to $\nabla$, or de Morgan dual, similarly to $\new$ and $\wen$ --- controlling distributivity properties which are undesirable when embedding probabilistic processes, much like the quantifiers in this work avoid undesirable distributivity properties when embedding processes with private names.

We remark that adapting recent work on splitting in \textit{subatomic logic}~\cite{Tubella2018} may help explain general patterns emerging, connecting the nominal quantifiers and sub-additives.
Subatomic logic may also be used to provide a more concise proof of splitting by exploiting the evident general patterns in the case analysis. Beside abstractly explaining general patterns, the study of \textsf{MAV1} in terms of subatomic logic would likely expose alternative formulations of the rules of \textsf{MAV1}.

%% file: decidable.tex
\section{Decidability of Proof Search}
\label{section:complexity}

Here we identify complexity classes for proof search in fragments of \textsf{MAV1}. 
The hardness results in this section are consequences of cut elimination (Theorem~\ref{theorem:cut}) and established complexity results for fragments of linear logic and extensions of \textsf{BV}.

NEXPTIME-hardness follows from the NEXPTIME-hardness of \textsf{MALL1}~\cite{Lincoln1994}; while membership in NEXPTIME follows a similar argument as for \textsf{MALL1}~\cite{Lincoln1994b} (in a proof there are at most exponentially many \textit{atomic interaction} rules, each involving quadratically bounded terms).
\begin{proposition}\label{proposition:NEXPTIME}
Deciding provability in \textsf{MAV1} is NEXPTIME-complete.
\end{proposition}
If we restrict terms to a nominal type, i.e.\ \textit{some} can only be instantiated with variables and constants, we obtain a tighter complexity bound.
PSPACE-hardness is a consequence of the PSPACE-hardness of \textsf{MAV}~\cite{Horne2015}, which in turn follows from the PSPACE-hardness of MALL~\cite{Lincoln1992}. Membership in PSPACE follows a similar argument as for \textsf{MALL1} without function symbols~\cite{Lincoln1994}.
\begin{proposition}\label{proposition:PSPACE}
Deciding provability in \textsf{MAV1} without function symbols is PSPACE-complete.
\end{proposition}
If we consider the sub-system without \textit{with} and \textit{plus}, named \textsf{BV1}, we obtain a tighter complexity bound again, even with function symbols in terms.
NP-hardness is a consequence of the NP-hardness of \textsf{BV}~\cite{Ozan2008a}; while membership in NP follows a similar argument as for \textsf{MLL1}~\cite{Lincoln1994b}
\begin{proposition}\label{proposition:NP}
Deciding provability in \textsf{BV1} is NP-complete.
\end{proposition}
For problems in the complexity class NEXPTIME, we can always check a proof in exponential time.
The high worst-case complexity means that proof search in general is considered to be infeasible.
Implementations of NEXPTIME-complete problems that regularly work efficiently, include reasoning in description logic $\mathcal{ALCI(D)}$~\cite{Lutz2004}.

\begin{figure*}[t]
\begin{tabular}{l|l|l}
Complexity class &
Linear logic &
Calculus of structures
\\ \hline
NP-complete &
\textsf{MLL1}~with functions~\cite{Kanovich1994} &
\txt{
\textsf{BV1}~with functions
\\
(Proposition~\ref{proposition:NP})
}
\\ \hline
PSPACE-complete &
\textsf{MALL1}~without functions~\cite{Lincoln1992} &
\txt{
\textsf{MAV1}~without functions
\\
(Proposition~\ref{proposition:PSPACE})
}
\\ \hline
NEXPTIME-complete &
\textsf{MALL1}~with functions~\cite{Lincoln1994,Lincoln1994b} &
\txt{
\textsf{MAV1}~with functions
\\
(Proposition~\ref{proposition:NEXPTIME})
}
\\ \hline
Undecidable &
\textsf{MAELL}~\cite{Lincoln1992}
and 
\textsf{MLL2}~\cite{Lincoln1995}
&
\textsf{NEL}~\cite{Strassburger2003}
\end{tabular}
\caption{Complexity results.}\label{figure:complexity}
\end{figure*}


Figure~\ref{figure:complexity} summarises complexity results for related calculi.
Notice the pattern that each fragment of linear logic has the same complexity as the calculus that is a conservative extension of that fragment of linear logic (with mix), where the extra operator is the self-dual non-commutative operator \textit{seq}. The complexity classes match since the source of the NP-completeness in multiplicative-only linear logic (\textsf{MLL}) lies in the number of ways of partitioning resources (formulae), while the mix rule and sequence rule are also ways of partitioning the same resources.

An exceptional case is that \textsf{BV} extended with exponentials (\textsf{NEL}) is undecidable, whereas the decidability of multiplicative linear logic with exponentials (\textsf{MELL}) is unknown.\footnote{\textsf{MELL} was claimed to be decidable in~\cite{Bimbo2015}, but this was later refuted~\cite{Strassburger19}. } 
However, by including additives to obtain full propositional linear logic (\textsf{MAELL} or simply \textsf{LL}) provability is known to be undecidable.

By the above observations, the complexity of deciding linear implication for embeddings of finite name passing processes, as in $\pi$-calculus, is in PSPACE.
However, extending to finite value passing processes where terms constructed using function symbols can be communicated, e.g.\ capturing tuples in the polyadic $\pi$-calculus~\cite{Milner1993}, the complexity class increases, but only for processes involving choice.
Further extensions to \textsf{MAV1}
introducing second-order quantifiers, exponentials or fixed points would lead to undecidable proof search~\cite{Lafont1996,Lincoln1995,Strassburger2003}.

%% file: conclusion.tex
\section{Conclusion}

This paper makes two significant contributions to proof theory: the first cut elimination result for a novel de Morgan dual pair of nominal quantifiers;
and the first direct cut elimination result for first-order quantifiers in the calculus of structures. 
As a consequence of cut-elimination (Theorem~\ref{theorem:cut}), we obtain the first proof system that features both
non-commutative operator \textit{seq} and first-order quantifiers $\forall$ and $\exists$.
%
A novelty of the nominal quantifiers $\new$ and $\wen$
 compared to established self-dual nominal quantifiers
is in how they distribute over positive and negative operators. 
This greater control of bookkeeping of names enables private names to be modelled in direct embeddings of 
processes as formulae in \textsf{MAV1}.
In Section~\ref{section:syntax}, every rule in \textsf{MAV1} is justified as necessary either: for soundly embedding processes;
or for ensuring cut elimination holds. Of particular note, some rules were introduced for ensuring cut elimination holds in the presence of \textit{equivariance}.

The cut elimination result is an essential prerequisite for recommending the system \textsf{MAV1} as a logical system.
This paper only hints about formal connections between \textsf{MAV1} and models of processes, which receives separate attention in a companion paper~\cite{MSCS}.
In particular, we know that linear implication defines a precongruence over processes embedded as formulae, that is sound with respect to both weak simulation and pomset traces.

Further to connections with process calculi, there are several problems exposed as future work.
Regarding the sequent calculus, in the setting of linear logic (i.e.,\ without \textsf{seq}), it is an open problem to determine whether there is a sequent calculus presentation of \textit{new} and \textit{wen}.
Regarding model theory, a model theory or game semantics may help to explain the nature of the de Morgan dual pair of nominal quantifiers, 
although note that it remains an open problem just to establish a sound and complete denotational model of \textsf{BV}.
Another open question is whether quantifiers \textit{new} and \textit{wen} are relevant in a classical or intuitionistic setting, or whether these operators are uniquely interesting in a linear setting.
Since \textit{new} must distribute over classical disjunction (recall, in contrast, \textit{new} does not distribute over multiplicative disjunction), nominal operators \textit{new} and \textit{wen} likely collapse to an established self-dual nominal operator in the classical setting; hence \textit{wen} is probably unrelated to the `generous' operator proposed in related work on stratifiable languages~\cite{Gabbay}.
Regarding implementation, it is a challenge to reduce non-determinism in proof search~\cite{Chaudhuri2011,Ozan2014,Andreoli1992};
a problem that can perhaps be tackled by restricting to well-behaved fragments of \textsf{MAV1} or by exploiting complexity results to embed rules as constraints for a suitable solver. 
Regarding proof normalisation, systems including classical propositional logic~\cite{Tubella2017}, first-order logic~\cite{Tubella2017}, intuitionistic logic~\cite{Guenot2014} and \textsf{NEL} (\textsf{BV} with exponentials)~\cite{Strassburger2011} 
satisfy a proof normalisation property called \textit{decomposition} related to interpolation;
leading to the question of whether there is an alternative presentation of the rules of \textsf{MAV1}, for which a decomposition result can be established. 
Finally, an expressivity problem, perhaps related to decomposition, is how to establish cut elimination for second-order extensions suitable for modelling infinite processes.

\paragraph*{Acknowledgements.}
We thank the anonymous reviewers, whose thorough reading led to improvements in the presentation of \textsf{MAV1}.


%% file: appendix.tex
\newpage

\appendix

\section{Electronic Appendix}

\begin{proposition}[Reflexivity: Proposition~\ref{proposition:reflexivity}]
For any formula $P$, $\vdash \co{P} \cpar P$ holds, i.e., $\vdash P \multimap P$.
\end{proposition}
\begin{proof}
\input{proof_reflexivity}
\end{proof}

\begin{lemma}[Universal: Lemma~\ref{lemma:universal}] 
If $\vdash \context{ \forall x. P }$ holds
 then, for all terms~$v$, $\vdash \context{ P\sub{x}{v} }$ holds.
\end{lemma}
\input{LemmaUniversal}

\begin{lemma}[Lemma~\ref{lemma:medial}]
Assume that $I$ is a finite subset of natural numbers, $P_i$ and $Q_i$ are formulae, for $i \in I$, and $\tcontext{}$ is a killing context. There exist killing contexts $\tcontextn{0}{}$ and $\tcontextn{1}{}$ and sets of natural numbers $J \subseteq I$ and $K \subseteq I$ such that the following derivation holds:
$
\vcenter{
\infer[]{
\tcontext{ P_i \cseq Q_i \colon i \in I }
}{
\tcontextn{0}{ P_j \colon j \in J } \cseq
\tcontextn{1}{ Q_k \colon k \in K }
}}
$.
\end{lemma}
\input{killing-proof}

\begin{lemma}[Affine: Lemma~\ref{lemma:bound}]
Any derivation $\vcenter{\infer{Q}{P}}$, is bound such that $\size{P} \preceq \size{Q}$.
\end{lemma}
\begin{proof}
The proof proceeds by checking that each rule preserves the bound on the size of the formula, from which the result follows by induction on the length of a derivation.

Consider the case of the \textit{close} rule. $\occ{\new x. P \cpar \wen x. Q} = \occ{ P } \mplus \occ{ Q } = \occ{\mathopen{\new x.} \left( P \cpar Q \right)}$,
since $P \not\equiv \cunit$ and $Q \not\equiv \cunit$,
 and 
$
\wensize{\new x. P \cpar \wen x. Q}
=
\wensize{P} + \left( 1  +  \wensize{Q}\right)
>
\wensize{P} + \wensize{Q}
=
\wensize{\mathopen{\new x.} \left(P \cpar Q\right)}
$.

Consider the case of the \textit{fresh} rule.
For the occurrence count, $\occ{ \wen x. P } = \occ{ \new x. P }$ and the wen count strictly decreases as follows:
$
\wensize{ \wen x. P }
=
 1 + \wensize{ P }
>
\wensize{ P }
=
\wensize{ \new x. P }
$.

Consider the case of the \textit{extrude new} rule, where $Q \not\equiv \cunit$.
If $P \equiv \cunit$, then the occurence count is such that $\occ{ \new x. P \cpar Q } = \left\{\left\{ 0,0 \right\}\right\} \mplus \occ{ Q } > \occ{ Q } = \occ{ \mathopen{\new x.} \left( P \cpar Q \right)}$. If however $P \not\equiv \cunit$, then $\occ{ \new x. P \cpar Q } = \occ{ P } \mplus \occ{ Q } = \occ{ \mathopen{\new x.}\left( P \cpar Q \right) }$. Furthermore, for the new count, the following inequality holds:
$
\newsize{ \new x. P \cpar Q } = \left(1 + \newsize{P}\right)\newsize{Q} \geq 1 + \newsize{P}\newsize{Q} = \newsize{ \mathopen{\new x.} \left( P \cpar Q \right) }
$.

Consider the case of the \textit{external} rule, where $R \not\equiv \cunit$.
For the occurrence count, by distributivity of $\mcup$ over $\mplus$, 
the following multiset equality holds:
\[
\begin{array}{rl}
\occ{\left( P \wwith Q \right) \cpar R }
& = \left( \occ{P} \mcup \occ{Q} \right) \mplus \occ{R} \\
& = \left( \occ{P}  \mplus \occ{R}\right) \mcup \left(\occ{Q}  \mplus \occ{R} \right) \\
& = \occ{\left( P \cpar R\right) \wwith \left(Q  \cpar R \right)}
\end{array}
\]
For the wen count 
\[
\wensize{ \left( P \wwith Q \right) \cpar R }
=
\left(\wensize{ P } + \wensize{Q}\right) \wensize{R}
=
\wensize{ P }\wensize{R} + \wensize{Q}\wensize{R}
=
\wensize{ \left( P \cpar R\right) \wwith \left( Q \cpar R \right)  }
\]
and similarly for the new count.

Consider the case of the \textit{suspend} rule, where $P \not\equiv \cunit$ and $Q \not\equiv \cunit$.
For the occurrence count, 
$\occ{ \wen x. P \cseq \wen x. Q } = \occ{ P } \discup \occ{ Q } = \occ{ \mathopen{\wen x.} \left( P \cseq Q \right) }$
and
$\occ{ \wen x. P \cpar \wen x. Q } = \occ{ P } \mplus \occ{ Q } = \occ{ \mathopen{\wen x.} \left( P \cpar Q \right) }$
for \textit{par} and \textit{seq} respectively.
For the wen count for either operator, $\odot \in \left\{\cpar, \cseq \right\}$, the following strict inequality holds, noting $\wensize{P} \geq 1$ for any formula:
\[
\wensize{ \wen x. P \odot \wen x. Q }
=
\left( 1 + \wensize{ P }\right)\left(1 + \wensize{ Q }\right)
=
\wensize{ P } + \wensize{ P }\wensize{ Q } + \wensize{ Q }
>
1 + \wensize{P}\wensize{Q}
=
\wensize{ \mathopen{\wen x.} \left( P \cseq Q \right) }
\]

Consider the case of the \textit{left wen} rules, where $\nfv{x}{Q}$ and $Q \not\equiv \cunit$.
For the occurrence count, there are four cases covering the operators \textit{seq} and \textit{par}. 
\begin{itemize}
\item
If $P \equiv \cunit$ then, for \textit{seq}:
$
 \occ{ \mathopen{\wen x.} \left( P \cseq Q \right) } = \occ{ Q }  \mstrict  \left\{\left\{ 0,0 \right\}\right\} \discup \occ{ Q } 
 = \occ{ \wen x. P \cseq Q }
$.

\item
If $P \not\equiv \cunit$ then, for \textit{seq}:
$
\occ{ \wen x. P \cseq Q } = \occ{ P } \discup \occ{ Q } = \occ{ \mathopen{\wen x.} \left( P \cseq Q \right) }
$.

\item
If $P \equiv \cunit$ then for \textit{par}:
$
\occ{ \mathopen{\wen x.} \left( P \cpar Q \right) }
= \occ{Q}
\mstrict
\left\{\left\{ 0,0 \right\}\right\} \mplus \occ{ Q } 
=
\occ{ \wen x. P \cpar Q }
$.

\item
If $P \not\equiv \cunit$ then for \textit{par}:
$
\occ{ \wen x. P \cpar Q } = \occ{ P } \mplus \occ{ Q } = \occ{ \mathopen{\wen x.} \left( P \cpar Q \right) }
$.


\end{itemize}
For the wen count 
$
\begin{array}{rl}
\wensize{ \wen x. P \odot Q }
=
\left(1 + \wensize{ P }\right)\wensize{ Q }
=
\wensize{Q} + \wensize{P}\wensize{ Q }
\geq
1 + \wensize{ P }\wensize{ Q }
=
\wensize{ \mathopen{\wen x.} \left( P \odot Q \right) }
\end{array}
$ holds, for $\odot \in \left\{\cpar, \cseq \right\}$.
Also, for the new count 
$\newsize{ \wen x. P \cseq Q } = \mmax{\newsize{ P }, \newsize{ Q}} = \newsize{ \mathopen{\wen x.} \left( P \cseq Q \right) }$ and
$\newsize{ \wen x. P \cpar Q } = \newsize{ P }\newsize{ Q } = \newsize{ \mathopen{\wen x.} \left( P \cpar Q \right) }$.
The case \textit{right wen} follows a symmetric argument.

Consider the case for the \textit{extrude} rule, where $Q \not\equiv \cunit$.
$\occ{ \mathopen{\forall x.} \left( P \cpar Q \right) } \mstrict \occ{ \forall x. P \cpar Q }$ by the following:
$
 \left\{\left\{0\right\}\right\} \mcup \left( \occ{ P } \mplus \occ{ Q } \right)
\mstrict
 \left(\left\{\left\{0\right\}\right\} \mplus \occ{ Q } \right) \mcup \left( \occ{ P } \mplus \occ{ Q } \right)
=
 \left( \left\{\left\{0\right\}\right\} \mcup \occ{ P } \right) \mplus \occ{ Q }
$.

Consider the case for the \textit{medial1} rule, where $P \not\equiv \cunit$ and $Q \not\equiv\cunit$.
By distributivity of $\discup$ over $\mcup$,
$
\occ{ \mathopen{\forall x.} \left( P \cseq Q \right) }
=
\left\{\left\{ 0 \right\}\right\} \mcup \left( \occ{ P } \discup \occ{ Q } \right)
=
\left( \left\{\left\{ 0 \right\}\right\} \mcup \occ{ P } \right) \discup \left( \left\{\left\{ 0 \right\}\right\} \mcup \occ{ Q } \right)
=
\occ{ \forall x. P \cseq \forall x. Q }
$.
Also $\wensize{ \mathopen{\forall x.} \left( P \cseq Q \right) } = \wensize{ \forall x. P \cseq \forall x. Q }$
and $\newsize{ \mathopen{\forall x.} \left( P \cseq Q \right) } = \newsize{ \forall x. P \cseq \forall x. Q }$.

For the \textit{select} rule,
$
 \occ{ \exists x. P }
   = \left\{\left\{ 0 \right\}\right\} \mcup \occ{ P }
   \mstrict \occ{ P }
   = \occ{ P \sub{x}{t} }
$,
by Lemma~\ref{lemma:size-sub}.

Consider the case for the \textit{switch} rule, where $P \not\equiv \cunit$ and $R \not\equiv \cunit$.
If $Q \not\equiv \cunit$, then, since $R \not\equiv \cunit$ we have $\left\{\left\{ 0 \right\}\right\} \mstrict \occ{R}$ and hence $\occ{P} = \occ{P} \mplus \left\{\left\{ 0 \right\}\right\} \mstrict \occ{P} \mplus \occ{R}$; and therefore the following holds since $\discup$ distributes over $\mplus$.
\[
\begin{array}{rl}
\occ{P \tensor \left(Q \cpar R\right)}
=& \occ{P} \discup \left( \occ{Q} \mplus \occ{R} \right)
\\
\mstrict& \left(\occ{P} \mplus \occ{R}\right) \discup \left( \occ{Q} \mplus \occ{R} \right)
\\
=& \left(\occ{P} \discup \occ{Q}\right) \mplus \occ{R}
= \occ{\left(P \tensor Q\right) \cpar R}
\end{array}
\]
If $Q \equiv \cunit$ then, since $\left\{\left\{0\right\}\right\} \mstrict \occ{P}$ and $\left\{\left\{0\right\}\right\} \mstrict \occ{R}$, the following hold.
\[
\occ{P \tensor \left( \cunit \cpar R\right)}
= \occ{P} \discup \occ{R}
\mstrict \occ{P} \mplus \occ{R}
= \occ{\left(P \tensor \cunit\right) \cpar R}
\]

Consider the case of the \textit{sequence} rule, where $P \not\equiv \cunit$ and $S \not\equiv \cunit$.
If $Q \not\equiv \cunit$ and $R \not\equiv \cunit$, then the following holds since $\discup$ distributes over $\mplus$.
\[
\!\!
\begin{array}{rl}
\occ{\left(P \cpar R\right) \cseq \left(Q \cpar S\right)}
=&\!\!\left( \occ{P} \mplus \occ{R} \right) \discup \left( \occ{Q} \mplus \occ{S} \right)
\\
\mstrict&\!\!\left( \occ{P} \mplus \occ{R} \right) \discup \left( \occ{Q} \mplus \occ{S} \right) \discup \left( \occ{P} \mplus \occ{S} \right) \discup \left( \occ{Q} \mplus \occ{R} \right)
\\
=&\!\!\left( \occ{P} \discup \occ{Q} \right) \mplus \left( \occ{R} \discup \occ{S} \right)
= \occ{\left(P \cseq Q\right) \cpar \left(R \cpar S\right)}
\end{array}
\]
If $Q \equiv \cunit$ and $R \not\equiv \cunit$, then, since $\left\{\left\{0\right\}\right\} \mstrict \occ{R}$, and hence $\occ{S} = \occ{S} \mplus \left\{\left\{0\right\}\right\} \mstrict \occ{S} \mplus \occ{R}$, therefore since $\discup$ distributes over $\mplus$.
\[
\begin{array}{rl}
\!\!
\occ{\left(P \cpar R\right) \cseq \left(\cunit \cpar S\right)}
=\left( \occ{P} \mplus \occ{R} \right) \discup \occ{S}
\mstrict&\!\!\left( \occ{P} \mplus \occ{R} \right) \discup \left(\occ{P} \mplus \occ{S}\right)
\\
=&\!\! \occ{P} \mplus \left( \occ{R} \discup \occ{S} \right)
=
\occ{\left(P \cseq \cunit \right) \cpar \left(R \cpar S\right)}
\end{array}
\]
A symmetric argument holds when $Q \not\equiv \cunit$ and $R \equiv \cunit$.

If $Q \equiv \cunit$ and $R \equiv \cunit$, then $\left\{\left\{0\right\}\right\} \mstrict \occ{P}$ and $\left\{\left\{0\right\}\right\} \mstrict \occ{S}$; hence the following strict inequality holds:
$
\occ{\left(P \cpar \cunit\right) \cseq \left(\cunit \cpar S\right)}
= \occ{P} \discup \occ{S}
\mstrict \occ{P} \mplus \occ{S}
= \occ{\left(P \cseq \cunit \right) \cpar \left(\cunit \cseq S\right)}
$.

Consider the case of the \textit{medial new} rule where $P \not\equiv \cunit$ and $Q \not\equiv \cunit$. For the occurrence count the equality $\occ{ \mathopen{\new x.} \left( P \cseq Q \right) } = \occ{ P } \mplus \occ{ Q } = \occ{ \new x. P \cseq \new x.Q }$ holds.
For the wen count,
$\wensize{ \mathopen{\new x.} \left( P \cseq Q \right) } = 
\wensize{ P } \wensize{ Q } = 
\wensize{ \new x. P \cseq \new x. Q }$.
For the new count the following equality holds:
$
\newsize{ \mathopen{\new x.} \left( P \cseq Q \right) }
=
1 + \mmax{\newsize{ P }, \newsize{ Q }}
=
\mmax{ 1 + \newsize{ P }, 1 + \newsize{ Q }}
=
\newsize{ \new x. P \cseq \new x. Q }
$.

Consider the case for the \textit{medial} rule, where either $P \not\equiv \cunit$ or $R \not\equiv \cunit$ and also either $Q \not\equiv \cunit$ or $S \not\equiv \cunit$.
When all of $P$, $Q$, $R$ and $S$ are not equivalent to the unit, we have the following.
\[
\begin{array}{rl}
\!\!\!
\occ{
 \left(P \wwith R\right) \cseq \left(Q \wwith S\right)
}
=&\!\!\!
\left( \occ{P} \mcup \occ{R} \right) \discup \left( \occ{Q} \mcup \occ{S} \right)
\\
\mstrict
&\!\!\!
\left( \occ{P} \mcup \occ{R} \right) \discup \left( \occ{Q} \mcup \occ{S} \right) \discup \left( \occ{P} \mcup \occ{S} \right) \discup \left( \occ{Q} \mcup \occ{R} \right)
\\
=
&\!\!\!
\left( \occ{P} \discup \occ{Q} \right) \mcup \left( \occ{R} \discup \occ{S} \right)
=
\occ{
 \left(P \cseq Q\right) \wwith \left(R \cseq S\right)
}
\end{array}
\]
For when exactly one of $P$, $Q$, $R$ and $S$ is equivalent to the unit, all cases are symmetric. Without loss of generality suppose that $S \equiv \cunit$ (and possibly also $Q \equiv \cunit$, but $R \not\equiv \cunit$). By distributivity of $\discup$ over $\mcup$ the following holds.
\[
\begin{array}{rl}
\occ{
 \left(P \wwith R\right) \cseq \left(Q \wwith \cunit\right)
}
=&
\left( \occ{P} \mcup \occ{R} \right) \discup \left( \occ{Q} \mcup \left\{\left\{ 0 \right\}\right\} \right)
\\
\mstrict&
\left( \occ{P} \mcup \occ{R} \right) \discup \left( \occ{Q} \mcup \occ{R} \right)
\\
=&
\left( \occ{P} \discup \occ{Q} \right) \mcup \occ{R}
=
\occ{
 \left(P \cseq Q\right) \wwith \left(R \cseq \cunit \right)
}
\end{array}
\]
There is one more form of case to consider for the medial: either $P \not\equiv \cunit$, $Q \equiv \cunit$, $R \equiv \cunit$ and $S \not\equiv \cunit$; or $P \equiv \cunit$, $Q \not\equiv \cunit$, $R \not\equiv \cunit$ and $S \equiv \cunit$. We consider only the former case. The later case, can be treated symmetrically.
Since $P \not\equiv \cunit$ and $S \not\equiv \cunit$, $\left\{\left\{ 0 \right\}\right\} \mstrict \occ{P}$ and $\left\{\left\{ 0 \right\}\right\} \mstrict \occ{S}$. Therefore, $\occ{P} \mcup \left\{\left\{ 0 \right\}\right\} \mstrict \occ{P} \mcup \occ{S}$ and $\occ{Q} \mcup \left\{\left\{ 0 \right\}\right\} \mstrict \occ{P} \mcup \occ{S}$. Hence, we have established that $\left(\occ{P} \mcup \left\{\left\{ 0 \right\}\right\}\right) \discup \left( \occ{Q} \mcup \left\{\left\{ 0 \right\}\right\} \right) \mstrict \occ{P} \mcup \occ{S}$.
Note that the restriction on the \textit{medial} rule, either $P \not\equiv \cunit$ or $R \not\equiv \cunit$ and also either $Q \not\equiv \cunit$ or $S \not\equiv \cunit$, excludes any further cases. Hence we have established that $
\occ{\left(P \wwith R\right) \cseq \left(Q \wwith S\right)}
\mstrict
\occ{\left(P \cseq Q\right) \wwith \left(R \cseq S\right)}$. 

For the \textit{with name} rule $\occ{\quantifier x.P \wwith \quantifier x.Q} = \occ{P} \mcup \occ{Q} = \occ{\mathopen{\quantifier x.}\left(P \wwith Q\right)}$,
where $\quantifier \in \left\{ \new, \wen \right\}$. For the new count $\newsize{\new x.P \wwith \new x.Q} = 2 + \newsize{P} + \newsize{Q} > 1 + \newsize{P} + \newsize{Q} = \newsize{\mathopen{\new x.}\left(P \wwith Q\right)}$
and $\newsize{\wen x.P \wwith \wen x.Q} = \newsize{\mathopen{\wen x.}\left(P \wwith Q\right)}$.
Similarly, $\wensize{\wen x.P \wwith \wen x.Q} > \wensize{\mathopen{\wen x.}\left(P \wwith Q\right)}$.
For \textit{left name}, \textit{right name} and \textit{all name}, the size of formulae are invariant.

The cases for the rules \textit{tidy}, \textit{tidy name}, \textit{left}, \textit{right}, \textit{atomic interact} are established by the following inequalities:
$\occ{ \cunit } \mstrict \occ{ \cunit \wwith \cunit }$,
$\occ{ \cunit } \mstrict \occ{ \new x. \cunit }$,
$\occ{ \cunit } \mstrict \occ{ \co{a} \cpar a }$,
$\occ{ P } \mstrict \occ{ P \ooplus Q }$
and
$\occ{ Q } \mstrict \occ{ P \ooplus Q }$.

Hence the lemma holds by induction on the length of the derivation.
\end{proof}


\begin{lemma}
[Lemma~\ref{lemma:split-exists}]
If $\vdash \exists x. P \cpar Q$, then there exist formulae $V_i$ and values $v_i$ such that $\vdash P\sub{x}{v_i} \cpar V_i$, where $1 \leq i \leq n$, and $n$-ary killing context $\tcontext{}$ such that $\vcenter{\infer{Q}{ \tcontext{V_1, V_2, \hdots, V_n}}}$
and if $\tcontext{}$ binds $y$ then $\nfv{y}{\left(\exists x . P\right)}$.
\end{lemma}

\input{split-exists}


\input{context-invariant}

\begin{lemma}[co-left and co-right: Lemma~\ref{lemma:co-branching}]
If $\vdash \context{ P \wwith Q }$ holds
then both $\vdash \context{ P }$ and $\vdash \context{ Q }$ hold.
\end{lemma}
\input{cobranching}

\begin{lemma}[co-external: Lemma~\ref{lemma:co-external}]
If $\vdash \context{ P \tensor  \left( Q \ooplus R \right) }$ holds
 then $\vdash \context{ \left( P \tensor Q \right) \ooplus \left( P \tensor R \right) }$ holds.
\end{lemma}
\tohide{
\input{coexternal}
}{}

\begin{lemma}[co-sequence: Lemma~\ref{lemma:co-sequence}]
If $\vdash \context{ \left( P \cseq Q \right) \tensor \left( R \andthen S \right) }$ holds
 then $\vdash \context{ \left( P \tensor R \right) \cseq \left( Q \tensor S \right) }$ holds.
\end{lemma}
\tohide{
\input{coseq}

}{}

\begin{lemma}[co-tidy: Lemma~\ref{lemma:co-tidy}]
If $\vdash \context{ \cunit \ooplus \cunit }$ holds, then $\vdash \context{ \cunit }$ holds.
\end{lemma}
\begin{proof}
Assume that $\vdash \left( \cunit \ooplus \cunit \right) \cpar P$ holds. By Lemma~\ref{lemma:split-plus}, there exist killing context $\tcontext{}$
and formulae $U_i$ for $1 \leq i \leq n$ such that $\vdash \cunit \cpar U_i$ or $\vdash \cunit \cpar U_i$ hold, hence $\vdash U_i$ holds, and the following derivation can be constructed.
\[
\infer[.]{
P
}{
 \tcontext{ U_1, \hdots, U_n }
}
\]
Thereby the following proof can be constructed:
\[
\infer[.]{
 P
}{
\infer[]{
 \tcontext{ U_1, \hdots, U_n } 
}{
\infer[]{
 \tcontext{ \cunit, \cunit, \hdots } 
}{ \cunit
}}}
\]
Therefore, by Lemma~\ref{lemma:context}, for any context $\vdash \context{ \cunit \ooplus \cunit }$ yields $\vdash \context{ \cunit }$, as required.
\end{proof}
\begin{lemma}[atomic co-interaction: Lemma~\ref{lemma:co-atoms}]
If $\vdash \context{ \alpha \tensor \co{\alpha} }$ holds
then $\vdash \context{ \cunit }$ holds.
\end{lemma}
\input{coatoms}

%% file: proof_reflexivity.tex
The proof proceeds by induction on the structure of a formula $P$.
The base cases for any atom $\alpha$ follows immediately from the \textit{atomic interaction} rule.
The base case for the unit is immediate by definition of a proof. 
For the following inductive cases assume that $\vdash \co{P} \cpar P$ and $\vdash \co{Q} \cpar Q$ hold.

Consider when the root connective in the formula is the $\tensor$ operator. 
Observe, by definition, $\co{\left(P \tensor Q\right)} \cpar \left(P \tensor Q\right) = \co{P} \cpar \co{Q} \cpar \left(P \tensor Q\right)$ and 
by applying the \textit{switch} rule and then the \textit{induction hypothesis}
we have the following proof:
\[
\infer[.]{
\co{P} \cpar \co{Q} \cpar \left(P \tensor Q\right) 
}{
\infer[
]{
\left( \co{P} \cpar P \right) \tensor \left( \co{Q} \cpar Q\right)
}{
\cunit
}}
\]

The case when the root connective is the \textit{par} operator is symmetric to the case for \textit{times}.

Consider when the root connective in the formula is the \textit{seq} operator. We have, by definition,
$ \co{\left(P \cseq Q\right)} \cpar \left(P \cseq Q\right)
    = \left(\co{P} \cseq \co{Q}\right) \cpar \left(P \cseq Q\right)$
and, by applying the \textit{sequence} rule and then the \textit{induction hypothesis}, the following proof holds:
\[
\infer[.]{
 \left(\co{P} \cseq \co{Q}\right) \cpar \left(P \cseq Q\right) 
}{
\infer[
]{
\left( \co{P} \cpar P \right) \cseq \left( \co{Q} \cpar Q\right)
}{
\cunit
}}
\]

Consider when the root connective in the formula is the \textit{with} operator.
By definition we have 
$\co{\left(P \wwith Q\right)} \cpar \left(P \wwith Q\right)
    =
\left(\co{P} \ooplus \co{Q}\right) \cpar \left(P \wwith Q\right)$
and the following proof holds.
\[
\infer[\mbox{by the \textit{external} rule}]{
\left(\co{P} \ooplus \co{Q}\right) \cpar \left(P \wwith Q\right) 
}{
\infer[\mbox{by the \textit{left} and \textit{right} rules}]{
\left(\left(\co{P} \ooplus \co{Q}\right) \cpar P\right)
\wwith
\left(\left(\co{P} \ooplus \co{Q}\right) \cpar Q\right) 
}{
\infer[\mbox{by the \textit{induction hypothesis}}]{
\left(\co{P} \cpar P\right) \wwith \left(\co{Q} \cpar Q\right) 
}{
\infer[\mbox{by \textit{tidy}}]{
\cunit \wwith \cunit
}{
\cunit
}}}
}
\]
The case for when \textit{plus}, $\ooplus$, is the root connective is symmetric to the case for \textit{with}.

Consider when the root connective in the formula is $\forall$.
By definition,
$ \co{\forall x. P} \cpar \forall x. P
    = \exists x. \co{P} \cpar \forall x. P $
and the following proof holds:
\[
\infer[\mbox{by the \textit{extrude1} rule}]{
\exists x. \co{P} \cpar \forall x. P 
}{
\infer[\mbox{by the \textit{select1} rule}]{
\mathopen{\forall x.} \left(\exists x. \co{P} \cpar P\right) 
}{
\infer[\mbox{by the \textit{induction hypothesis}}]{
\mathopen{\forall x.} \left(\co{P} \cpar P\right) 
}{
\infer[\mbox{by the \textit{tidy1} rule}]{
\forall x. \cunit
}{
 \cunit
}}}}
\]

 The case for when $\exists$ is the root connective is symmetric to the case for $\forall$.

Consider when the root connective in the formula is $\new$.
By definition $ \co{\new{x}. P} \cpar \new{x}. P
    = \wen{x}. \co{P} \cpar \new{x}. P $
and the following proof holds:
\[
\infer[\mbox{by the \textit{close} rule}]{
 \wen{x}. \co{P} \cpar \new{x}. P 
}{
\infer[\mbox{by the \textit{induction hypothesis}}]{
 \mathopen{\new{x}.} \left( \co{P} \cpar P \right) 
}{
\infer[\mbox{by the \textit{tidy name} rule}]{
 \mathopen{\new x.} \cunit
}{
  \cunit
}}}
\]

The case for when the root connective is $\wen$ is symmetric to the case for $\new$.

Hence, by induction on the number of connectives in the formula, reflexivity holds.

%% file: LemmaUniversal.tex
\begin{proof}
We require a function over formulae $s_v(T)$ that replaces a certain universal quantifier in $T$ with a substitution for a value $v$. 
The universal quantifiers to be replaced are highlighted in bold $\boldall$. Note that during a proof the bold operator may be duplicated by the \textit{external} rule and \textit{medial1} rule, hence there may be multiple bold occurrences in a formula.
The function is defined as follows, where $\odot \in \left\{ \cseq, \cpar, \tensor, \ooplus, \wwith \right\}$ is any binary connective, $\quantifier \in \left\{ \forall, \exists, \new, \wen \right\}$ is any quantifier except bold universal quantification and $\kappa \in \left\{ \alpha, \co{\alpha}, \cunit \right\}$ is any constant or atom.
\begin{gather*}
\mathopen{s_v}\left(\boldall x. T\right) = \mathopen{s_v}\left(T \sub{x}{v}\right)
\qquad
\mathopen{s_v}\left(\quantifier x. T\right) =
\quantifier x. \mathopen{s_v}\left(T\right)
\qquad
\mathopen{s_v}\left( T \odot U \right) = 
\mathopen{s_v}\left(T\right) \odot \mathopen{s_v}\left(U\right)
\qquad
\mathopen{s_v}\left(\kappa\right) = \kappa
\end{gather*}
In what follows we use that $\mathopen{s_v}\left(\context{ U }\right) = \contextp{ \mathopen{s_v}\left(U'\right) }$, 
for some context $\context{}$ and $U'$ 
such that $\contextp{}$ is obtained from $\context{}$ by applying the $s_v$ function and
$U'$ is obtained by substituting free variables in $U$, 
that are bound by $\boldall$ quantifiers in the context $\context{}$, with $v$.

We shall prove a stronger statement in the following: 
for every $R$, if $\vdash R$ holds then for all terms $v$, $\vdash \mathopen{s_v}\left(R\right)$ holds.

Without loss of generality, we can assume that the bound and the free variables in $R$ are pairwise
distinct and that the bound variables in $R$ are also distinct from the 
variables in $v$. This simplifies the proof below since substitutions of $\boldall$-quantified variables
commute with other connectives and quantifiers in $R.$

For the base case, $s_v(R) = R$, in which case trivially if $\vdash R$ then $\vdash s_v(R)$, for example where $R \equiv \cunit$.

Consider the case when the bottommost rule in a proof is an instance of the \textit{extrude1} rule involving a bold universal quantifier, as follows, 
$\vcenter{\infer[]{\context{ \boldall x. T \cpar U }}{\context{ \mathopen{\boldall x.} \left( T \cpar U \right) }}}$, 
where $\nfv{x}{U}$ and $\vdash \context{ \mathopen{\boldall x.} \left( T \cpar U \right) }$.

By the induction hypothesis, $\vdash \mathopen{s_v}\left(\context{ \mathopen{\boldall x.} \left( T \cpar U \right) }\right)$ holds.
Now the following equalities hold.
\[
\begin{array}{rl}
\mathopen{s_v}\left(\context{ \mathopen{\boldall x.} \left( T \cpar U \right) }\right)
= & \contextp{ \mathopen{s_v}\left(\left( T' \cpar U' \right) \sub{x}{v}\right)} \\
= & \contextp{ \mathopen{s_v}\left(T'\sub{x}{v}\right) \cpar \mathopen{s_v}\left(U'\right) } \\
= & \mathopen{s_v}\left(\context{ \boldall x. T \cpar U }\right)
\end{array}
\]
Hence $\vdash \mathopen{s_v}\left(\context{ \boldall x. T \cpar U }\right)$ holds as required.

Consider the case where the bottommost rule of a proof is an instance of the \textit{tidy1} rule of the form 
$\vcenter{\infer[]{\context{ \boldall x. \cunit }}{\context{ \cunit }}}$, where $\vdash \context{\cunit}$ holds.
By the induction hypothesis, $\vdash \mathopen{s_v}\left( \context{ \cunit } \right)$ holds.
Since $\mathopen{s_v}\left( \context{ \boldall x. \cunit } \right) = \mathopen{s_v}\left( \context{ \cunit } \right)$, we have $\vdash \mathopen{s_v}\left( \context{ \boldall x. \cunit } \right)$ holds, as required.

Consider the case where the bottommost rule of a proof is an instance of the \textit{all name} rule of the form 
$\vcenter{\infer[]{\context{ \boldall x. \wen y. P }}{\context{ \wen y. \boldall x. P }}}$, where $\vdash \context{ \wen y. \boldall x. P }$ holds.
By the induction hypothesis, $\vdash s_v\left( \context{ \wen y. \boldall x. P } \right)$ holds. Observe that the following equalities hold, by definition of function $s_v$.
\[
\begin{array}{rl}
\mathopen{s_v}\left( \context{  \boldall x. \wen y. P } \right)
=
\contextp{ \mathopen{s_v}\left(\left( \wen y. P' \right)\sub{x}{v}\right) }
=
\contextp{ \mathopen{\wen y. s_v}\left(P'\sub{x}{v} \right)}
=
\mathopen{s_v}\left( \context{ \wen y. \boldall x. P } \right)
\end{array}
\]
Hence $\vdash \mathopen{s_v}\left( \context{ \wen y. \boldall x. P } \right)$ holds, as required.
The case where \textit{all name} involves \textit{new} is similar.

Consider the case when the bottommost rule does not involve a bold universal quantifier.
We show here one instance where the rule involved is \textit{extrude1}; the other cases are similar.
So suppose the bottommost rule instance is 
\[
\infer[.]{
\context{ \forall x. T \cpar U }
}{
\context{ \mathopen{\forall x.} \left( T \cpar U \right) }
}
\]
By the induction hypothesis, $\vdash \mathopen{s_v}\left(\context{ \mathopen{\forall x.} \left( T \cpar U \right) }\right)$.
So, since
\[
\mathopen{s_v}\left(\context{ \mathopen{\forall x.}\left(T \cpar U\right) }\right)
=
\contextp{ \mathopen{\forall x.}\left(\mathopen{s_v}\left(T'\right) \cpar \mathopen{s_v}\left(U'\right)\right) } 
\]
we have $\vdash \contextp{ \mathopen{\forall x.}\left(\mathopen{s_v}\left(T'\right) \cpar \mathopen{s_v}\left(U'\right)\right) }$ also holds.
Hence, since
\[
\mathopen{s_v}\left(\context{ \forall x. T \cpar U }\right) 
= \contextp{ \forall x. \mathopen{s_v}\left(T'\right) \cpar \mathopen{s_v}\left(U'\right) }
\]
and 
\[
\infer[]{
\contextp{ \forall x. \mathopen{s_v}\left(T'\right) \cpar \mathopen{s_v}\left(U'\right) }
}{
\contextp{ \mathopen{\forall x.}\left(\mathopen{s_v}\left(T'\right) \cpar \mathopen{s_v} \left(U'\right)\right) } 
}
\]
we have $\vdash \mathopen{s_v}\left(\context{ \forall x. T \cpar U }\right)$ holds, as required.

The statement of the lemma is then a special case of the stronger statement established by induction. If $\vdash \context{ \boldall x. T }$, where no further bold universal quantifiers occur in the context, then $\vdash \context{ T\sub{x}{v} }$ holds, since in such a scenario $\mathopen{s_v}\left( \context{ \boldall x. T } \right) = \context{ T\sub{x}{v} }$.
\end{proof}

%% file: killing-proof.tex
\begin{proof}
Proceed by induction on the structure of the killing context. The base case is immediate.

Consider a predicate of the form $\new x. \tcontext{P_i \cseq Q_i \colon i \in I }$.
By the induction hypothesis, assume there exists $\tcontextn{0}{}$ and $\tcontextn{1}{}$ such that 
\[
\begin{prooftree}
\tcontextn{0}{P_j \colon j \in J} \cseq \tcontextn{1}{ Q_k \colon k \in K }
\justifies
\tcontext{P_i \cseq Q_i \colon i \in I }
\end{prooftree}
\]
where $J \subseteq I$ and $K \subseteq I$.
There are three cases to consider.

If $\tcontextn{0}{P_j \colon j \in J} \equiv \cunit$, then we have derivation 
\[
\infer[.]{
\new x. \tcontext{P_i \cseq Q_i \colon i \in I }
}{
\infer[\mbox{by using $\equiv$}]{
 \cunit \cseq \new x. \tcontextn{1}{ Q_k \colon k \in K }
}{
  \mathopen{\new x.} \left( \cunit \cseq \tcontextn{1}{Q_k \colon k \in K } \right)
}
}
\]

If $\tcontextn{1}{Q_k \colon k \in K} \equiv \cunit$, then we have derivation
\[
\infer[.]{
\new x. \tcontext{P_i \cseq Q_i \colon i \in I }
}{
\infer[\mbox{using $\equiv$}]{
\new x. \tcontextn{0}{P_j \colon j \in J } \cseq \cunit
}{
  \mathopen{\new x.} \left( \tcontextn{0}{P_j \colon j \in J } \cseq \cunit \right)
}
}
\]

Otherwise, $\tcontextn{0}{P_j \colon j \in J} \not\equiv \cunit$ and $\tcontextn{1}{Q_k \colon k \in K} \not\equiv \cunit$ in which case the \textit{medial new} rule can be applied as follows:
\[
\infer[.]{
 \new x. \tcontext{P_i \cseq Q_i \colon i \in I }
}{
 \infer[\mbox{by the \textit{medial new} rule}]{
  \mathopen{\new x.} \left( \tcontextn{0}{P_j \colon j \in J } \cseq \tcontextn{1}{Q_k \colon k \in K } \right) 
 }{
  \new x. \tcontextn{0}{P_j \colon j \in J } \cseq \new x. \tcontextn{1}{Q_k \colon k \in K }
 }
}
\]

In each of the three cases above, killing contexts of the correct form are obtained. The arguments in the cases of universal quantifiers and with follow a similar pattern.
\end{proof}

%% file: split-exists.tex
\begin{proof}
The proof proceeds by induction on the size of the proof in Definition~\ref{definition:size},
until the principal \textit{exists} operator is removed from the proof, according to the base case.
In the base case, the bottommost rule in a proof is an instance of the \textit{select} rule of the form 
$
\vcenter{
\infer[,]{
\exists x. T \cpar U 
}{
 T\sub{x}{v} \cpar U
}
}$ where $\vdash T\sub{x}{v} \cpar V$ holds; hence splitting is immediately satisfied.
As in every splitting lemma, there are commutative cases for \textit{new}, \textit{wen}, \textit{all}, \textit{with}, \textit{times} and two for \textit{seq}.

Consider the commutative case induced by the \textit{external} rule.
The bottommost rule is the form
\[
\infer[]{
\exists x. T \cpar \left( U \wwith V \right) \cpar W \cpar P
}{
\left(\exists x. T \cpar U \cpar W \wwith
\exists x. T \cpar V \cpar W\right)
\cpar P
}
\]
where it holds that $\vdash \left(\left(\exists x. T \cpar U \cpar W\right) \wwith \left(\exists x. T \cpar V \cpar W\right)\right) \cpar P$.
By Lemma~\ref{lemma:split-times}, $\vdash \exists x. T \cpar U \cpar W \cpar P$ and $\vdash \exists x. T \cpar V \cpar W \cpar P$; and furthermore 
$\size{\exists x. T \cpar U \cpar W \cpar P} \mstrict \size{ \exists x. T \cpar \left( U \wwith V \right) \cpar W \cpar P }$
and
$\size{\exists x. T \cpar V \cpar W \cpar P} \mstrict \size{ \exists x. T \cpar \left( U \wwith V \right) \cpar W \cpar P }$ hold.
Hence, by the induction hypothesis, there exist $Q_i$ and $u_i$ such that $\vdash T\sub{x}{u_i} \cpar Q_i$, for $1 \leq i \leq m$, and $R_j$ and $v_j$ such that
$\vdash T\sub{x}{v_j} \cpar R_j$, for $1 \leq j \leq n$; and $m$-ary and $n$-ary killing contexts $\tcontextn{0}{}$ and $\tcontextn{1}{}$ such that 
the derivations $(1)$ and $(2)$ below hold. 
\[
\begin{array}{ccc}
\infer[]{
U \cpar W \cpar P
}{
 \tcontextn{0}{Q_1, \hdots, Q_m}
}
&
\infer[]{
V \cpar W \cpar P
}{
 \tcontextn{1}{R_1, \hdots, R_n}
}
&
\infer[]{
\left( U \wwith V \right) \cpar W \cpar P
}{
\infer[]{
\left(U \cpar W \cpar P\right) \wwith \left(V \cpar W \cpar P\right) 
}{
\tcontextn{0}{Q_1, \hdots, Q_m} \wwith \tcontextn{1}{R_1, \hdots, R_n}
}} \\
(1) & (2) & (3)
\end{array}
\]
Thus the derivation $(3)$ above can be constructed.
Notice that $\tcontextn{0}{} \wwith \tcontextn{1}{}$ is an $m+n$-ary killing context, as required.

Consider the commutative case induced by the \textit{extrude1} rule. In this case, the bottommost rule is
\[
\infer[]{
\exists x. T \cpar \forall y. U \cpar V \cpar W
}{
\mathopen{\forall y.}\left( \exists x. T \cpar U \cpar V \right) \cpar W
}
\]
 assuming $\nfv{y}{(\exists x. T \cpar V)}$
where $\vdash \mathopen{\forall y.}\left( \exists x. T \cpar U \cpar V \right) \cpar W$ holds.
By Lemma~\ref{lemma:universal}, for every variable $z$, $\vdash \left( \exists x. T \cpar U \cpar V \right)\sub{y}{z} \cpar W$ holds.
Furthermore, by definition of substitution $\left( \exists x. T \cpar U \cpar V \right)\sub{y}{z} \cpar W \equiv \exists x. T \cpar U\sub{y}{z} \cpar V \cpar W$,  since $\nfv{y}{(\exists x. T \cpar V)}$.
Now observe the strict multiset inequality $
\size{ \exists x. T \cpar U\sub{y}{z} \cpar V \cpar W } \mstrict \size{ \exists x. T \cpar \forall y. U \cpar V \cpar W }
$ holds;
hence, by the induction hypothesis, for every variable $z$, there exist formulae $P^z_i$ and values $v^z_i$ such that $\vdash T\sub{x}{v^z_i} \cpar P^z_i$ holds, for $1 \leq i \leq n$, and $n$-ary killing context $\tcontext{}$ such that derivation $(4)$ below can be constructed. Hence, for $\nfv{z}{(\forall y. U \cpar V \cpar W)}$, the  derivation $(5)$
below can be constructed:
\[
\begin{array}{cc}
\infer[]{
U\sub{y}{z} \cpar V \cpar W
}{
\tcontext{P^z_1, \hdots, P^z_n}
}
&
\infer[]{
\forall y. U \cpar V \cpar W
}{
\infer[]{
\mathopen{\forall z.} \left( U\sub{y}{z} \cpar V \cpar W \right) 
}{
\mathopen{\forall z.} \tcontext{P^z_1, \hdots, P^z_n} 
}} \\
(4) & (5) 
\end{array}
\]
Notice that $\mathopen{\forall z.} \tcontext{}$ is a $n$-ary killing context as required.

Consider the commutative cases involving the \textit{sequence} rule.
We present the scenario where the principal formula $\exists x. U$ moves entirely to the left hand side of \textit{seq} operator. The cases where the principal formula moves entirely to the right hand side of the \textit{seq} operator and the commutative case for \textit{times}, are similar to the cases presented below.
In the scenario we consider, the bottommost rule in a proof is of the following form:
\[
\infer[]{
\exists x. U \cpar \left( V \cseq P \right) \cpar W \cpar Q
}{
\left( \left( \exists x. U \cpar V \cpar W\right) \cseq P \right) \cpar Q
}
\]
such that $\vdash \left( \left( \exists x. U \cpar V \cpar W\right) \cseq P \right) \cpar Q$ holds.
By Lemma~\ref{lemma:split-times}, there exist $R_i$ and $S_i$ such that $\vdash \exists x. U \cpar V \cpar W \cpar R_i$ and $\vdash P \cpar S_i$ hold, for $1 \leq i \leq n$, and $n$-ary killing context $\tcontext{}$ such that the derivation
$
\vcenter{
\infer[]{
 Q 
}{
 \tcontext{ R_1 \cseq S_1, \hdots, R_n \cseq S_n }
}}
$ holds, 
and furthermore the size of the proof of ${ \exists x. U \cpar V \cpar W \cpar R_i }$ is bounded above by the size of the proof of $\left( \left( \exists x. U \cpar V \cpar W\right) \cseq P \right) \cpar Q$ hence strictly bounded above by the size of the proof of $\exists x. U \cpar \left( V \cseq P \right) \cpar W \cpar Q$.
By the induction hypothesis, for $1 \leq i \leq n$, there exist formulae $P^i_j$ and terms $t^i_j$ such that $\vdash U\sub{x}{t^i_j} \cpar P^i_j$, for $1 \leq j \leq m_i$, and killing contexts $\tcontextn{i}{}$ such that the derivation $
\vcenter{
\infer[]{
V \cpar W \cpar R_i
}{
 \tcontextn{i}{P^i_1, \hdots, P^i_{m_i}}
}}$ holds.
Hence the following derivation can be constructed, as required.
\[
\infer[]{
\left( V \cseq P \right) \cpar W \cpar Q
}{
\infer[]{
\left( V \cseq P \right) \cpar W \cpar \tcontext{ R_1 \cseq S_1, \hdots, R_n \cseq S_n } 
}{
\infer[]{
\tcontext{ \left( V \cseq P \right) \cpar W \cpar R_i \cseq S_i \colon 1 \leq i \leq n } 
}{
\infer[]{
\tcontext{ \left( V \cpar W \cpar R_i \right) \cseq \left( P \cpar S_i \right) \colon 1 \leq i \leq n }
}{
\infer[]{
\tcontext{ V \cpar W \cpar R_i \colon 1 \leq i \leq n } 
}{
\tcontext{ \tcontextn{1}{P^1_1, \hdots, P^1_{m_1}}, \hdots, \tcontextn{n}{P^n_1, \hdots, P^n_{m_n}} } 
}}}}}
\]
Notice that $\tcontext{ \tcontextn{1}{}, \hdots, \tcontextn{n}{} }$ is a $\sum_{i=1}^n{m_i}$-ary killing context as required.

Consider the commutative case induced by the \textit{extrude new} rule.
In this case, the bottommost rule of a proof is of the form 
\[
\infer[\mbox{, where $\nfv{y}{\exists x. P \cpar R}$ and $\vdash \mathopen{\new y.} \left( \exists x. P \cpar Q \cpar R\right) \cpar S$ holds.}]{
\exists x. P \cpar \new y. Q \cpar R \cpar S
}{
 \mathopen{\new y.} \left( \exists x. P \cpar Q \cpar R\right) \cpar S
}
\] 
By Lemma~\ref{lemma:split-times}, there exist $T$ and $U$ such that $\vdash \exists x. P \cpar Q \cpar R \cpar U$, $\nfv{y}{T}$ holds and either $T = U$ or $T = \wen y. U$, and also
$\vcenter{\infer{S}{T}}$.
Furthermore, the size of the proof of $\exists x. P \cpar Q \cpar R \cpar U$ is bounded above by the size of the proof of $\mathopen{\new y.} \left( \exists x. P \cpar Q \cpar R\right) \cpar S$ and hence strictly bounded above by the size of the proof of $\exists x. P \cpar \new y. Q \cpar R \cpar S$, enabling the induction hypothesis.
Hence, by the induction hypothesis, there exist formulae $V_i$ and terms $t_i$ such that $\vdash P\sub{x}{t_i} \cpar V_i$ holds, for $1 \leq i \leq n$, and $n$-ary killing context $\tcontext{}$ such that the derivation $
\vcenter{
\infer{
Q \cpar R \cpar U
}{
\tcontext{ V_1, \hdots, V_n }
}}$ holds.
Observe that, either $T = U$ and $\nfv{y}{U}$, and hence we have
derivation $(6)$ below; 
or $T = \wen y. U$ and hence we have derivation $(7)$ below. Thereby we can construct the derivation 
$(8)$ below.
\[
\begin{array}{ccc}
\infer[]{
\new y. Q \cpar R \cpar T
}{
 \mathopen{\new y.} \left(Q \cpar R \cpar U\right)
}
&
\infer[]{
\new y. Q \cpar R \cpar \wen y. U
}{
\infer[]{
 \mathopen{\new y.} \left(Q \cpar R\right) \cpar \wen y. U
}{
 \mathopen{\new y.} \left(Q \cpar R \cpar U\right)
}}
&
\infer[]{
\new y. Q \cpar R \cpar S
}{
\infer[]{
\new y. Q \cpar R \cpar T
}{
\infer[]{
\mathopen{\new y.} \left( Q \cpar R \cpar U \right)
}{
\new y. \tcontext{ V_1, \hdots, V_n }
}}} \\
(6) & (7) & (8)
\end{array}
\]
Observe that $\new y. \tcontext{ }$ is a $n$-ary killing context as required.

Consider the commutative case induced by the \textit{right wen} rule.
In this case, the bottommost rule of a proof is of the form
\[
\infer[\mbox{, where $\nfv{y}{\exists x. P \cpar R}$.}]{
\exists x. P \cpar \wen y. Q \cpar R \cpar S
}{
 \wen y \left( \exists x. P \cpar Q \cpar R\right) \cpar S
}
\]
By Lemma~\ref{lemma:split-times}, there exist $T$ and $U$ such that 
$\vdash \exists x. P \cpar Q \cpar R \cpar U$, $\nfv{y}{T}$ and either $T = U$ or $T = \new y. U$, and also
$\vcenter{\infer{S}{T}}$.
Furthermore, the size of the proof of $\exists x.P \cpar Q \cpar R \cpar U$ is bounded above by the size of the proof of $\mathopen{\wen y.} \left( \exists x. P \cpar Q \cpar R\right) \cpar S$ and hence strictly bounded above by the size of the proof of $\exists x. P \cpar \wen y. Q \cpar R \cpar S$, enabling the induction hyothesis.
Hence, by the induction hypothesis, there exist formulae $V_i$ and terms $t_i$ such that $\vdash P\sub{x}{t_i} \cpar V_i$, for $1 \leq i \leq n$, and $n$-ary killing context $\tcontext{}$ such that $
\vcenter{
\infer{
Q \cpar R \cpar U
}{
 \tcontext{ V_1, \hdots, V_n }
}}$.
Observe that either $T = U$ and $\nfv{y}{U}$ hence the derivation $(9)$ below holds; 
or $T = \new y. U$ hence the derivation $(10)$ below holds. Hence the derivation $(11)$ below can be constructed:
\[
\begin{array}{ccc}
\infer{
\wen y. Q \cpar R \cpar T
}{
\infer{
 \mathopen{\wen y.} \left(Q \cpar R \cpar U\right)
}{
 \mathopen{\new y.} \left(Q \cpar R \cpar U\right)
}}
&
\infer{
\wen y. Q \cpar R \cpar \new y. U
}{
\infer[]{
 \mathopen{\wen y.} \left(Q \cpar R\right) \cpar \new y. U 
}{
 \mathopen{\new y.} \left(Q \cpar R \cpar U\right)
}}
&
\infer[]{
\wen y. Q \cpar R \cpar S
}{
\infer[]{
\wen y. Q \cpar R \cpar T
}{
\infer[]{
\mathopen{\new y.} \left( Q \cpar R \cpar U \right)
}{
\new y. \tcontext{ V_1, \hdots, V_n }
}}} \\
(9) & (10) & (11)
\end{array}
\]
Observe that $\new y. \tcontext{ }$ is a $n$-ary killing context as required.

In many commutative cases, the bottommost rule does not interfere with the principal formula. Consider when a rule is applied outside the scope of the principal formula.
In this case, the bottommost rule in a proof is of the form
$
\vcenter{
\infer[]{
 \exists x. U \cpar \context{ V }
}{
 \exists x. U \cpar \context{ W }
}}$
such that $\vdash \exists x. U \cpar \context{ W }$.
By the induction hypothesis, there exist formulae $P_i$ and terms $t_i$, for $1 \leq i \leq n$ such that $\vdash U\sub{x}{t_i} \cpar P_i$, for $1 \leq i \leq n$, and $n$-ary killing context $\tcontext{}$ such that 
$
\vcenter{
\infer[]{
\context{ W }
}{
 \tcontext{ P_1, \ldots, P_n }
}}
$.
Hence $
\vcenter{
\infer[]{
\context{ V }
}{
\infer[]{
 \context{ W } 
}{
 \tcontext{ P_1, \ldots, P_n }
}}}
$ as required.

Consider the following application of any rule 
$
\vcenter{
\infer[]{
\exists x. \context{T} \cpar W
}{
 \exists x. \context{U} \cpar W
}}
$
such that $\vdash \exists x. \context{U} \cpar W$. 
By the induction hypothesis, there exist formulae $P_i$ and terms $t_i$ where $\vdash \context{U}\sub{x}{t_i} \cpar P_i$, for $1 \leq i \leq n$, and $n$-ary killing context $\tcontext{}$ such that 
$
\vcenter{
\infer[]{
W
}{
 \tcontext{ P_1, \ldots, P_n }
}}
$.
Hence, by Lemma~\ref{lemma:substitution}, the proof 
$
\vcenter{
\infer[]{
\context{T}\sub{x}{v_i} \cpar P_i
}{
\infer[]{
 \context{U}\sub{x}{v_i} \cpar P_i 
}{
 \cunit
}}}
$ holds.

All cases have been considered hence the lemma holds by induction on the size of a proof.
\end{proof}

%% file: context-invariant.tex
\begin{lemma}
[Lemma~\ref{lemma:invariant}]
If $\vdash \context{T}$, then there exist formulae $U_i$ and substitutions $\sigma_i$, for $1 \leq i \leq n$, and $n$-ary killing context $\tcontext{}$ such that $\vdash T\sigma_i \cpar U_i$; and, for any formula $V$ there exist $W_i$ such that either $W_i = V\sigma_i \cpar U_i$ or $W_i = \cunit$ and the following holds:
$
\vcenter{
\infer[]{
 \context{V}
}{
 \tcontext{ W_1, W_2, \hdots, W_n }
}}
$.
\end{lemma}

\begin{proof}
The proof proceeds by induction on the size of the formula part of the context (n.b.\ not counting the size of atoms). The base case concerning one hole is immediate.

Consider the case for a context of the form $\exists x. \context{} \cpar P$, where $\vdash \exists x. \context{ T } \cpar P$.
By Lemma~\ref{lemma:split-exists},
there exist formulae $Q_i$ and values $v_i$ such that $\vdash \context{ T } \sub{x}{v_i} \cpar Q_i$, for $1 \leq i \leq n$;
and $n$-ary killing context $\tcontext{}$ such that the following derivation holds.
\[
\infer{P}{
 \tcontext{ Q_1, Q_2, \hdots, Q_n }
}
\]
For context $\context{}$ and any formula $U$, let $\contextn{i}{}$ and $\sigma_{i}$ be such that $\context{U}\sub{x}{v_i} \equiv \contextn{i}{U\sigma_{i}}$.
Notice that for first-order quantifiers, the substitutions does not increase the size of the formula part of the context. It can only increases the size of terms in atoms, which are not counted in this induction.
Since $\vdash \context{ T } \sub{x}{v_i} \cpar Q_i$ holds, then $\vdash \contextn{i}{ T\sigma_i }  \cpar Q_i$ holds.
Therefore, by the induction hypothesis, there exists formula $V^i_j$ such that either $V^i_j = \cunit$ or $V^i_j = \left(U\sigma_i\right) \sigma^i_j \cpar W^i_j$, where $\vdash \left(T \sigma_i\right) \sigma^i_j \cpar W^i_j$, for $1 \leq j \leq m_i$;
and $m_i$-ary killing context $\tcontextn{i}{}$
such that 
$
\context{ U } \sub{x}{v_i} \cpar Q_i
\equiv
\contextn{i}{ U\sigma_i } \cpar Q_i
$
and the following derivation holds:
\[
\infer[.]{
\contextn{i}{ U\sigma_i } \cpar Q_i
}{
\tcontextn{i}{ V^i_1, V^i_2, \hdots, V^i_{m_i} }
}
\]
Hence the following derivation can be constructed for all formulae $U$.
\[
\infer[]{
\exists x. \context{ U } \cpar P
}{
\infer[]{
\exists x. \context{ U } \cpar \tcontext{ Q_1, \hdots, Q_n }
}{
\infer[]{
\exists x. \context{ U } \cpar \tcontext{ Q_i \colon 1 \leq i \leq n }
}{
\infer[]{
\tcontext{ \exists x. \context{ U } \cpar Q_i \colon 1 \leq i \leq n }
}{
\infer[]{
\tcontext{ \context{ U } \sub{x}{v_i} \cpar Q_i \colon 1 \leq i \leq n }
}{
\tcontext{ \tcontextn{i}{ V^i_j \colon 1 \leq j \leq m_i } \colon 1 \leq i \leq n }
}}}}}
\]
Observe $V^i_j = \cunit$ or $V^i_j = U \mathclose{\left(\sigma_i \cdot \sigma^{i}_j\right)} \cpar W^i_j$, such that $\vdash T \mathclose{\left(\sigma_i \cdot \sigma^{i}_j\right)} \cpar W^i_j$, for all $i$ and $j$, as required.

Consider the case for a context of the form $\new x. \context{} \cpar P$, where $\vdash \new x. \context{ T } \cpar P$.
By Lemma~\ref{lemma:split-times},
there exist formulae $Q$ and $\hat{Q}$ such that $\vdash \context{ T } \cpar \hat{Q}$ and either $Q = \hat{Q}$ or $Q$ and $\wen x.\hat{Q}$, and also 
$
\vcenter{
\infer{P}{Q}}
$.
Therefore, by the induction hypothesis,
there exist formulae $V_i$ and $W_i$ and substitutions $\sigma_i$ such that either $V_i = \cunit$ or $V_i = U \sigma_i \cpar W_i$, where $\vdash T \sigma_i \cpar W_i$, for $1 \leq i \leq n$;
and $n$-ary killing context $\tcontext{}$
such that 
\[
\infer[.]{
\context{ U } \cpar \hat{Q}
}{
\tcontext{ V_1, V_2, \hdots, V_n }
}
\]
Hence the following derivation 
\[
\infer[]{
\new x. \context{ U } \cpar P
}{
\infer[]{
\new x. \context{ U } \cpar Q
}{
\infer[]{
\mathopen{\new x.} \left( \context{ U } \cpar \hat{Q} \right)
}{
\new x. \tcontextn{i}{ V_i \colon 1 \leq i \leq n }
}}}
\]
can be constructed for all formulae $U$, as required.

Consider the case for a context of the form $\wen x. \context{} \cpar P$, where $\vdash \wen x. \context{ T } \cpar P$. 
By Lemma~\ref{lemma:split-times},
there exist formulae $Q$ and $R$ such that $\nfv{x}{Q}$ and $\vdash \context{T} \cpar R$ and either $Q = R$ or $Q = \new x. R$, and also
$
\vcenter{\infer{P}{Q}}
$.
Therefore, by the induction hypothesis,
there exist formulae $V_i$ and $W_i$ and substitutions $\sigma_i$ such that either $V_i = \cunit$ or $V_i = U \sigma_i \cpar W_i$, where $\vdash T \sigma_i \cpar W_i$, for $1 \leq i \leq n$;
and $n$-ary killing context $\tcontext{}$
such that 
\[
\infer[.]{
\context{ U } \cpar R
}{
\tcontext{ V_1, V_2, \hdots, V_n }
}
\]
In the former case that $Q = R$, since $\nfv{x}{Q}$, the derivation
\[
\infer[]{
\wen x. \context{ U } \cpar R
}{
\infer[]{
\new x. \context{ U } \cpar R
}{
\new x. \left( \context{ U } \cpar R \right)
}}
\] holds.
In the case, $Q = \new x. R$ the derivation 
\[
\infer[]{
\wen x .\context{ U } \cpar \new x. R
}{
\mathopen{\new x.} \left( \context{ U } \cpar R \right)
}
\] holds.
Hence, for all formulae $U$,
\[
\infer[]{
\wen x. \context{ U } \cpar P
}{
\infer[.]{
\wen x. \context{ U } \cpar Q
}{
\infer[]{
\mathopen{\new x.} \left( \context{ U } \cpar R \right)
}{
\new x. \tcontext{ V_1, V_2, \hdots, V_n}
}}}
\]

Consider the case of a context of the form $\forall x. \context{ } \cpar P$, where $\vdash \forall x. \context{ T } \cpar$ holds. By Lemma~\ref{lemma:universal}, for any variable $y$, $\vdash \context{ T } \sub{x}{y} \cpar P$ holds.
For name $y$, let $\contextn{y}{}$ be such that for any formula $U$, $\context{U} \sub{x}{y} \equiv \contextn{y}{U\sub{x}{y}}$.
For any $y$, by the induction hypothesis, for any formula $U$, there exist formulae $V^y_i$ such that either $V^y_i = \cunit$ or $V^y_i = U\sub{x}{y}\sigma^y_i \cpar W^y_i$, where $\vdash T\sub{x}{y}\sigma^y_i \cpar W^y_i$ holds, for $1 \leq i \leq n$; and $n$-ary killing context $\tcontextn{y}{}$ such that 
$
\context{ U } \sub{x}{y} \cpar P
\equiv
\contextn{y}{ U \sub{x}{y}}  \cpar P$
and the following derivation can be constructed:
\[
\infer[.]{
\contextn{y}{ U \sub{x}{y}}  \cpar P
}{
\tcontextn{y}{ V^{y}_i \colon 1 \leq i \leq n  }
}
\]
Therefore, for  $\nfv{y}{(\forall x. \context{ U } \cpar P)}$ and any $U$,
derivation
\[
\infer[]{
\forall x. \context{ U } \cpar P
}{
\infer[]{
\mathopen{\forall y.} \left(\context{ U } \sub{x}{y} \cpar P \right)
}{
\forall y. \tcontextn{y}{ V^{y}_i \colon 1 \leq i \leq n }
}}
\] 
holds.
In the above $V^{y}_i = \cunit$ or $V^y_i = U\sub{x}{y}\sigma^y_i \cpar W^y_i$, where $\vdash T\sub{x}{y}\sigma^y_i \cpar W^y_i$ holds, for all $i$, as required.

The cases for \textit{plus}, \textit{with}, \textit{tensor} and \textit{seq} do not differ significantly from \textsf{MAV}~\cite{Horne2015}.
\end{proof}

%% file: cobranching.tex
\begin{proof}
Assume that $\vdash \left(P \wwith Q\right)\mathclose\sigma \cpar R$ holds.
By Lemma~\ref{lemma:split-times}, $\vdash P\mathclose\sigma \cpar R$ and $\vdash Q\mathclose\sigma \cpar R$ hold.
Hence by Lemma~\ref{lemma:context}, for any context $\context{}$, if $\vdash \context{ P \wwith Q }$ then $\vdash \context{ P }$ and $\vdash \context{ Q }$.
\end{proof}

%% file: coexternal.tex
\begin{proof}
Assume that $\vdash \left( \left( P \ooplus Q \right) \tensor R\right)\mathclose{\sigma} \cpar W$ holds, for some substitution $\sigma$.
By Lemma~\ref{lemma:split-times}, there exist formulae $T_i$ and $U_i$ such that $\vdash \left( P \ooplus Q \right)\mathclose{\sigma} \cpar T_i$ and $\vdash R\sigma \cpar U_i$, for $1 \leq i \leq n$, and killing context $\tcontext{}$ such that 
\[
\infer[.]{
 W
}{
 \tcontext{ T_1 \cpar U_1, \hdots, T_n \cpar U_n }
}
\]
Now, by Lemma~\ref{lemma:split-plus}, for every $i$, there exists killing context $\tcontextn{i}{}$ and types $V^i_j$ such that either $\vdash P\sigma \cpar V^i_j$ or $\vdash Q\sigma \cpar V^i_j$ holds, for $1 \leq j \leq m_i$, and the derivation
\[
\infer[]{
 T_i 
}{
 \tcontextn{i}{ V^i_1, V^i_2, \hdots, V^i_{m_i} }
}
\] holds.

Notice that if $\vdash P\sigma \cpar V^i_j$ holds then the following derivation can be constructed.
\[
\infer[]{
\left( \left(P \tensor R\right) \ooplus \left(Q \tensor R\right) \right)\mathclose{\sigma} \cpar V^i_j  \cpar U_i
}{
\infer[]{
 \left( P \tensor R \right)\mathclose{\sigma} \cpar V^i_j  \cpar U_i 
}{
\infer[]{
 \left( P\sigma  \cpar  V^i_j \right) \tensor \left( R\sigma  \cpar U_i \right) 
}{
 \cunit
}}}
\]
Otherwise $\vdash Q \cpar V^i_j$ holds, hence the following derivation can be constructed.
\[
\infer[]{
\left( \left(P \tensor R\right) \ooplus \left(Q \tensor R \right)\right)\mathclose{\sigma} \cpar V^i_j  \cpar U_i
}{
\infer[]{
 \left( Q \tensor R \right)\mathclose{\sigma} \cpar V^i_j  \cpar U_i 
}{
\infer[]{
 \left( Q\sigma  \cpar  V^i_j \right) \tensor \left( R\sigma  \cpar U_i \right)
}{ \cunit
}}}
\]
Hence by applying one of the above proofs for each $i$ and $j$ we can construct the following proof.
\[
\infer[]{
\left( \left(P \tensor R\right) \ooplus \left(Q \tensor R\right) \right)\mathclose{\sigma} \cpar W 
}{
\infer[]{
\left(\left( P \tensor R\right) \ooplus \left(Q \tensor R\right) \right)\mathclose{\sigma} \cpar \tcontext{ T_1 \cpar U_1, \hdots, T_n \cpar U_n } 
}{
\infer[]{
\left( \left(P \tensor R\right) \ooplus \left(Q \tensor R\right) \right)\mathclose{\sigma} \cpar \tcontext{ \tcontextn{i}{ V^i_1, V^i_2, \hdots, V^i_{m_i} } \cpar U_i \colon 1 \leq i \leq n } 
}{
\infer[]{
\left( \left(P \tensor R\right) \ooplus \left(Q \tensor R\right) \right)\mathclose{\sigma} \cpar \tcontext{ \tcontextn{i}{ V^i_j  \cpar U_i \colon 1 \leq j \leq m_i } \colon 1 \leq i \leq n }
}{
\infer[]{
 \tcontext{ \left( \left(P \tensor R\right) \ooplus \left(Q \tensor R \right)\right)\mathclose{\sigma} \cpar \tcontextn{i}{ V^i_j  \cpar U_i \colon 1 \leq j \leq m_i } \colon 1 \leq i \leq n } 
}{
\infer[]{
 \tcontext{ \tcontextn{i}{ \left(\left( P \tensor R\right) \ooplus \left(Q \tensor R \right)\right)\mathclose{\sigma} \cpar V^i_j  \cpar U_i \colon 1 \leq j \leq m_i } \colon 1 \leq i \leq n } 
}{
\infer[]{
 \tcontext{ \tcontextn{i}{ \cunit \colon 1 \leq j \leq m_i } \colon 1 \leq i \leq n } 
}{
 \cunit
}}}}}}}
\]
Hence $\vdash \left(\left( P \tensor R\right) \ooplus \left(Q \tensor R\right) \right) \cpar W$.
Therefore, by Lemma~\ref{lemma:context}, for any context $\vdash \context{ \left( P \ooplus Q \right) \tensor R }$ yields $\vdash \context{ \left(P \tensor R\right) \ooplus \left(Q \tensor R\right) }$, as required.
\end{proof}

%% file: coseq.tex
\begin{proof}
Assume that $\vdash \left( \left( P \cseq Q \right) \tensor \left( R \cseq S \right) \right)\mathclose{\sigma} \cpar U$ holds, for some substitution $\sigma$.
By Lemma~\ref{lemma:split-times}, there exist $n$-ary killing context $\tcontext{}$ and $U^0_i$ and $U^1_i$, for $1 \leq i \leq n$, such that $\vdash \left(P \cseq Q\right)\mathclose{\sigma} \cpar U^0_i$ and $\vdash \left(R \cseq S\right)\sigma \cpar U^1_i$ and the derivation
\[
\infer{
 U
}{
 \tcontext{ U^0_1 \cpar U^1_1, U^0_2 \cpar U^1_2, \hdots }
}
\]  holds.

Hence by Lemma~\ref{lemma:split-times}, for $k \in \left\{0, 1\right\}$ there exists $m^k_i$-ary killing context $\tcontextmn{k}{i}{}$ and types $V^k_{i,j}$, $W^k_{i,j}$ for $1 \leq j \leq m_i^k$, such that 
$\vdash P\sigma \cpar V^0_{i,j}$ and $\vdash Q \sigma\cpar W^0_{i,j}$ and $\vdash R\sigma \cpar V^1_{i,j}$ and $\vdash S\sigma \cpar W^1_{i,j}$
and the following derivation
\[
\infer{
U^k_i
}{
 \tcontextmn{k}{i}{V^k_{i,1}~\andthen~W^k_{i,1}, V^k_{i,2}~\andthen~W^k_{i,2} \hdots}
}
\]  holds.

Hence we can construct the following proof.
\[
\infer[]{
\left(\left(P \tensor R\right) \cseq \left(Q \tensor S\right)\right)\mathclose{\sigma} \cpar U
}{
\infer[]{
 \left(\left(P \tensor R\right) \cseq \left(Q \tensor S\right)\right)\mathclose{\sigma} \cpar
 \tcontext{ U^0_1 \cpar U^1_1, U^0_2 \cpar U^1_2, \hdots } 
}{
\infer[]{
 \left(\left(P \tensor R\right) \cseq \left(Q \tensor S\right)\right)\mathclose{\sigma} \cpar
 \tcontext{
 \begin{array}{rl}
      &\tcontextmn{0}{i}{V^0_{i,j} \cseq W^0_{i,j} \colon 1 \leq j \leq m^0_i } \\
 \cpar&\tcontextmn{1}{i}{V^1_{i,k} \cseq W^1_{i,k} \colon 1 \leq k \leq m^1_i } 
 \end{array}
   \colon 1 \leq i \leq n                                               
 }  
}{
\infer[]{
 \left(\left(P \tensor R\right) \cseq \left(Q \tensor S\right)\right)\mathclose{\sigma}
   \cpar
   \tcontext{
\begin{array}{r}
     \tcontextmn{1}{i}{
       \begin{array}{r}
       \tcontextmn{0}{i}{
          V^0_{i,j} \cseq W^0_{i,j}
           \colon 1 \leq j \leq m^0_i
       }
       \cpar
       \left(V^1_{i,k} \cseq W^1_{i,k}\right) 
\\
       \colon 1 \leq k \leq m^1_i
       \end{array}
       }
\\
  \colon 1 \leq i \leq n 
\end{array}
                                       }  
}{
\infer[]{
 \left(\left(P \tensor R\right) \cseq \left(Q \tensor S\right)\right)\mathclose{\sigma}
   \cpar
   \tcontext{
     \begin{array}{r}
     \tcontextmn{1}{i}{
       \begin{array}{r}
       \tcontextmn{0}{i}{
          \begin{array}{l}
          \left(V^0_{i,j} \cseq W^0_{i,j}\right) \cpar
       \left(V^1_{i,k} \cseq W^1_{i,k}\right)
           \colon 1 \leq j \leq m^0_i
          \end{array}
       }
       \\
       \colon 1 \leq k \leq m^1_i
       \end{array} }
     \\ \colon 1 \leq i \leq n 
     \end{array}
                                       }  
}{
\infer[]{
 \left(\left(P \tensor R\right) \cseq \left(Q \tensor S\right)\right)\sigma
   \cpar
   \tcontext{
     \begin{array}{r}
     \tcontextmn{1}{i}{
       \begin{array}{r}
       \tcontextmn{0}{i}{
          \begin{array}{l}
          \left(V^0_{i,j} \cpar V^1_{i,k}\right) \cseq
       \left(W^0_{i,j} \cpar W^1_{i,k}\right)
           \colon 1 \leq j \leq m^0_i
          \end{array}
       }
       \\
       \colon 1 \leq k \leq m^1_i
       \end{array} }
     \\ \colon 1 \leq i \leq n 
     \end{array}
                                       }  
}{
\infer[]{
   \tcontext{
     \tcontextmn{1}{i}{
       \tcontextmn{0}{i}{
          \begin{array}{l}
   \left(\left(P \tensor R\right) \cseq \left(Q \tensor S\right)\right)\mathclose{\sigma}
   \cpar \\
      \left(\left(V^0_{i,j} \cpar V^1_{i,k}\right) \cseq
       \left(W^0_{i,j} \cpar W^1_{i,k}\right)\right)
          \end{array} \colon 1 \leq j \leq m^0_i
       }
       \colon 1 \leq k \leq m^1_i
      }
     \colon 1 \leq i \leq n 
  }  
}{
\infer[]{
   \tcontext{
     \tcontextmn{1}{i}{
       \tcontextmn{0}{i}{
         \begin{array}{l}
         \left(\left(P \tensor R\right)\sigma \cpar V^0_{i,j} \cpar V^1_{i,k}\right)
         \cseq \\ 
         \left(
          \left(Q \tensor S\right)\sigma \cpar
          W^0_{i,j} \cpar W^1_{i,k}\right)
        \end{array}
           \colon 1 \leq j \leq m^0_i 
       }
       \colon 1 \leq k \leq m^1_i
     }
     \colon 1 \leq i \leq n 
   } 
}{
\infer[]{
   \tcontext{
     \tcontextmn{1}{i}{
       \tcontextmn{0}{i}{
          \begin{array}{l}
\left(\left(P\sigma \cpar V^0_{i,j}\right) \tensor \left(R\sigma \cpar V^1_{i,k}\right)\right)
   \cseq \\
      \left(
          \left(Q\sigma \cpar W^0_{i,j}\right) \tensor
       \left(S\sigma \cpar W^1_{i,k}\right)\right)
\end{array}
           \colon 1 \leq j \leq m^0_i          
       }
       \colon 1 \leq k \leq m^1_i
       }
     \colon 1 \leq i \leq n 
                                       }  
}{
\infer[]{
   \tcontext{
     \begin{array}{l}
     \tcontextmn{1}{i}{
       \begin{array}{l}
       \tcontextmn{0}{i}{
          \begin{array}{l}
          \cunit
           \colon 1 \leq j \leq m^0_i
          \end{array}
       }
       \colon 1 \leq k \leq m^1_i
       \end{array} }
      \colon 1 \leq i \leq n 
     \end{array}
                                       }  
}{
 \cunit
}}}}}}}}}}
\]
Therefore, by Lemma~\ref{lemma:context}, for any context $\vdash \context{\left(P \cseq Q\right) \tensor \left(R \cseq S\right)}$ yields $\vdash \context{ \left(P \tensor R\right) \cseq \left(Q \tensor S\right) }$.
\end{proof}

%% file: coatoms.tex
\begin{proof}
Assume for atom $\alpha$ that 
$\vdash \left( \alpha \tensor \co{\alpha} \right)\sigma \cpar P$, for some formula $P$ and some substitution $\sigma$.
By Lemma~\ref{lemma:split-times}, there exist $n$-ary killing context $\tcontext{}$ and formulae $U_i$ and $V_i$ such that $\vdash \alpha \sigma \cpar U_i$ and $\vdash \co{\alpha \sigma} \cpar ~V_i$, for $1 \leq i \leq n$, such that 
\[
\infer[.]{P}{
\tcontext{ U_1 \cpar V_1, U_2 \cpar  V_2, \hdots}
}
\]
By Lemma~\ref{lemma:split-atoms}, for every $i$, there exist $m_i^0$-ary killing contexts $\tcontextmn{0}{i}{}$
such that 
\[
\infer[.]{U_i}{
\tcontextmn{0}{i}{\co{\alpha \sigma}, \hdots, \co{\alpha \sigma}}
}
\]
By Lemma~\ref{lemma:split-atoms}, for every $i$, there exist $m_i^1$-ary killing contexts $\tcontextmn{1}{i}{}$
such that 
\[
\infer[.]{
V_i}{
 \tcontextmn{1}{i}{\alpha \sigma, \hdots, \alpha \sigma}
}
\]
Thereby the following proof can be constructed.
\[
\begin{array}{rl}
\infer[]{
P
}{
\infer[]{
 \tcontext{U_1 \cpar V_1, U_2 \cpar V_2, \hdots} 
}{
\infer[]{
 \tcontext{
                               \tcontextmn{0}{i}{\co{\alpha \sigma} \colon 1 \leq j \leq m_i^0} 
                         \cpar \tcontextmn{1}{i}{\alpha \sigma \colon 1 \leq k \leq m_i^1 }
                             \colon 1 \leq i \leq n 
}
}{
\infer[]{
 \tcontext{
    \tcontextmn{1}{i}{\tcontextmn{0}{i}{\co{\alpha \sigma} \colon 1 \leq j \leq m_i^0} \cpar \alpha \sigma \colon 1 \leq k \leq m_i^1 }
    \colon 1 \leq i \leq n
}
}{
\infer[]{
 \tcontext{
    \tcontextmn{1}{i}{\tcontextmn{0}{i}{\co{\alpha \sigma} \cpar \alpha \sigma \colon 1 \leq j \leq m_i^0} \colon 1 \leq k \leq m_i^1 } 
    \colon 1 \leq i \leq n 
}
}{
\infer[]{
 \tcontext{
    \tcontextmn{1}{i}{\tcontextmn{0}{i}{\cunit \colon 1 \leq j \leq m_i^0} \colon 1 \leq k \leq m_i^1 }
\colon 1 \leq i \leq n }
}{
 \cunit
}}}}}}
\end{array}
\]
Therefore, by Lemma~\ref{lemma:context}, for any context $\context{}$, $\vdash \context{\alpha \tensor \co{\alpha}}$ yields that $\vdash \context{ \cunit }$, as required.
\end{proof}